\documentclass[11pt,noinfoline]{imsart}

\RequirePackage[OT1]{fontenc}
\RequirePackage{amsthm,amsmath}
\RequirePackage{natbib}
\RequirePackage[colorlinks,citecolor=blue,urlcolor=blue,filecolor=blue]{hyperref}
\usepackage{booktabs}

\usepackage{subfigure}
\usepackage{layout}

\usepackage{dsfont}
\usepackage[mathscr]{eucal}
\usepackage[toc,page]{appendix}
\usepackage{mathrsfs}
\usepackage{color}
\usepackage{pifont}
\usepackage{bm}
\usepackage{latexsym}
\usepackage{amsfonts}
\usepackage{amssymb}
\usepackage{epsfig}
\usepackage{graphicx}
\usepackage{multirow}

\usepackage{caption}
\usepackage{fullpage}
\usepackage{enumitem}
\usepackage{algorithm}

\usepackage{verbatim}
\usepackage{xcolor}

\allowdisplaybreaks

\newcommand {\R}{\mathbb{R}}

\newcommand {\Prob}{\mathbb{P}}

\newcommand{\cH}{\mathcal{H}}

\newcommand{\verti}[1]{{\left\vert\kern-0.25ex\left\vert\kern-0.25ex\left\vert #1 
    \right\vert\kern-0.25ex\right\vert\kern-0.25ex\right\vert}}
\newcommand{\vertj}{\vert\kern-0.25ex\vert\kern-0.25ex\vert}

\newcommand{\argmin}{\operatornamewithlimits{arg\,min}}

\DeclareMathOperator{\diag}{diag}

\newcommand{\ip}[1]{{\left\langle\kern-0.5ex\left\langle\kern-0.5ex\left\langle #1 
    \right\rangle\kern-0.5ex\right\rangle\kern-0.5ex\right\rangle}}

\DeclareMathOperator{\tr}{tr}
\newcommand {\bbB}{\mathbb{B}}
\newcommand {\bbE}{\mathbb{E}}
\newcommand {\bbN}{\mathbb{N}}
\newcommand {\bbP}{\mathbb{P}}
\newcommand {\bbQ}{\mathbb{Q}}
\newcommand {\bbR}{\mathbb{R}}

\newcommand {\bbf}{\mathbf{f}}
\newcommand {\bbg}{\mathbf{g}}
\newcommand {\bbh}{\mathbf{h}}
\newcommand {\bbm}{\mathbf{m}}

\newcommand {\bA}{\mathbf{A}}
\newcommand {\bB}{\mathbf{B}}
\newcommand {\bC}{\mathbf{C}}
\newcommand {\bD}{\mathbf{D}}
\newcommand {\bE}{\mathbf{E}}
\newcommand {\bR}{\mathbf{R}}
\newcommand {\bP}{\mathbf{P}}
\newcommand {\bQ}{\mathbf{Q}}
\newcommand {\bH}{\mathbf{H}}
\newcommand {\bI}{\mathbf{I}}
\newcommand {\bJ}{\mathbf{J}}
\newcommand {\bW}{\mathbf{W}}
\newcommand {\bZ}{\mathbf{Z}}

\newcommand {\bzero}{\mathbf{0}}

\newcommand {\cD}{\mathcal{D}}
\newcommand {\cE}{\mathcal{E}}
\newcommand {\cF}{\mathcal{F}}
\newcommand {\cG}{\mathcal{G}}
\newcommand {\cL}{\mathcal{L}}
\newcommand {\cS}{\mathcal{S}}
\DeclareMathOperator{\dg}{dg}
\newcommand{\CI}{\mathrel{\text{\scalebox{1.07}{$\perp\mkern-10mu\perp$}}}}
\newcommand{\supp}{\mathrm{supp}}
\newcommand{\tnorm}[1]
{{\left\vert\kern-0.25ex\left\vert\kern-0.25ex
  \left\vert #1 \right\vert\kern-0.25ex
  \right\vert\kern-0.25ex\right\vert}}

\definecolor{tomas}{rgb}{0.8, 0.5, 0.2}
\definecolor{kartik}{RGB}{0,167,159}
\definecolor{victor}{rgb}{0.8,0, 0}

\newtheorem{theorem}{Theorem}

\newtheorem{corollary}{Corollary}
\newtheorem{lemma}{Lemma}
\newtheorem{definition}{Definition}
\newtheorem{remark}{Remark}

\newtheorem{assumption}{Assumption}
\newtheorem*{assumption-again}{Assumption 1*}

\startlocaldefs
\numberwithin{equation}{section}
\theoremstyle{plain}
\endlocaldefs

\begin{document}

\begin{frontmatter}

\title{{\Large The Functional Graphical Lasso}}

\runtitle{Functional Graphical Lasso}

\begin{aug}
\author{\fnms{Kartik G.} \snm{Waghmare}\ead[label=e1]{kartik.waghmare@epfl.ch}},
\author{\fnms{Tomas} \snm{Masak}\ead[label=e2]{tomas.masak@epfl.ch}}
\and
\author{\fnms{Victor M.} \snm{Panaretos}\ead[label=e3]{victor.panaretos@epfl.ch}}

\thankstext{t1}{Research supported by a Swiss National Science Foundation grant.}

\runauthor{K. Waghmare, T. Masak \& V.M. Panaretos}

\affiliation{Ecole Polytechnique F\'ed\'erale de Lausanne}


\end{aug}

\begin{abstract}
We consider the problem of recovering conditional independence relationships between $p$ jointly distributed Hilbertian random elements given $n$  realizations thereof. We operate in the sparse high-dimensional regime, where $n\ll p$ and no element is related to more than $d\ll p$ other elements. In this context, we propose an infinite-dimensional generalization of the graphical lasso. We prove model selection consistency under natural assumptions and extend many classical results to infinite dimensions. In particular, we do not require finite truncation or additional structural restrictions. The plug-in nature of our method makes it applicable to any observational regime, whether sparse or dense, and indifferent to serial dependence. Importantly, our method can be understood as naturally arising from a coherent maximum likelihood philosophy.
\end{abstract}


\end{frontmatter}

\tableofcontents

\renewcommand{\baselinestretch}{1.1}

\section{Introduction}

Let $X = (X_{1}, \dots, X_{p})^\top$ where $\{X_{j}\}_{j=1}^{p}$ are jointly distributed second-order random elements in the Hilbert spaces $\{\cH_{j}\}_{j=1}^{p}$, respectively. The conditional independence structure of $X$ can be thought of as an undirected graph $G$ with the vertices $\{X_{j}\}_{j=1}^{p}$, where for $i \neq j$, $X_{i}$ and $X_{j}$ are adjacent unless they are conditionally independent given the rest of the vertices $\{X_{k}\}_{k\neq i,j}$, that is 
\begin{equation*}
    X_{i} \CI X_{j} ~|~ \{X_{k}\}_{k\neq i,j}.
\end{equation*}
The maximum degree $d$ of $G$ is defined as the maximum number of neighbours (adjacent vertices) of a vertex of $G$. We are interested in determining the edges of the graph $G$ from $n$ independent realizations $\{X^{k}\}_{k=1}^{n}$ of $X$ in the sparse high-dimensional regime, where $n \ll p$ and $d \ll p$. The standard (multivariate) version of the problem can be seen as a special case where $\cH_{j} = \bbR$ for every $j$ and consequently $\{X_{j}\}_{j=1}^{p}$ are real-valued random variables. But our framework is considerably more general, in that the spaces $\cH_j$ can be different and infinite-dimensional.

In the multivariate setting, the problem has been studied comprehensively and many methods have been devised. Of these, \emph{precision thresholding} is the simplest as it merely requires thresholding the entries of the inverse of the empirical covariance matrix $\hat{\bC}$. If the absolute value of the $(i,j)$-th entry of $\hat{\bC}^{-1}$ is below the threshold, then the corresponding edge is understood as being absent in the graph. The motivation for this comes directly from a classic result in the theory of Gaussian graphical models that we will call the \emph{inverse zero characterization}, which states that if $X$ is Gaussian with an invertible covariance $\bC$, the $(i,j)$-th entry of $\bC^{-1}$ is non-zero if and only if $X_{i}$ and $X_{j}$ are adjacent \citep{lauritzen1996,meinshausen2006,drton2017}. The method does not perform well in the sparse high-dimensional regime ($n \ll p$) because it cannot exploit the sparsity in the graph structure.

Fortunately, there are methods which are consistent in high-dimensions. One such method, known as \emph{neighbourhood selection} \citep{meinshausen2006}, involves performing $\ell_{1}$-penalized linear regression on each of the random variables against the rest with the non-zero coefficients in the regression corresponding to the neighbours of the random variable. A second such method, called the \emph{graphical lasso} \citep{yuan2007,friedman2008} combines the sparsity-exploiting properties of the $\ell_{1}$ penalty along with the inverse zero characterization and is known to be consistent in high-dimensional settings \citep{rothman2008,ravikumar2011}. In practice, the graphical lasso is arguably the method of choice for Gaussian graphical models, likely due to its conceptual simplicity and ability to perform estimation and model selection (i.e.~support estimation) in a single step \citep{yuan2007}. The method involves estimating the precision matrix by maximizing the appropriately penalized Gaussian log-likelihood:
\begin{equation}\label{eq:graphical_lasso}
\hat{\bQ} = \argmin_{\bQ} \tr\big( \hat{\bC} \bQ \big) - \log \det(\bQ) + \lambda \|\bQ\|_{1-},
\end{equation}
where $\bQ = [q_{ij}]_{i,j=1}^{p}$ is positive-definite, $\lambda > 0$ is a tuning parameter and $\|\bQ\|_{1-} = \sum_{i\neq j} |q_{ij}|$ is the penalty term which promotes sparsity in $\bQ$ by driving the the less significant of its off-diagonal entries to zero. 

We will extend graphical lasso to the general Hilbertian setting by reformulating the optimization problem (\ref{eq:graphical_lasso}) in infinite-dimensional terms. Our primary concern is the \emph{multivariate functional data} setting in which the $X_{j}$ are real-valued random functions on compact intervals. The distinctive feature of functional data \citep{ramsay2005,hsing2015}, as opposed to multivariate data, is the fact that the covariance operator is trace-class and thus not boundedly invertible, obscuring the relationship between the graphical model and the support of the inverse covariance. While we focus on multivariate functional data, our approach can in principle be used to recover relationships between diverse types of random objects be they variables, vectors, functions, or surfaces, so long as they can be represented as second-order random elements in a Hilbert space. We will mostly restrict ourselves to the classical setting where the $X_{j}$ are jointly Gaussian. For non-Gaussian $X_{j}$, our method recovers relationships based on conditional uncorrelatedness, which reflects purely linear relationships between the random elements.

Recovering conditional independence graphs of multivariate functional data by means of extending the graphical lasso has been attempted before in the literature.
\citet{qiao2019} proposed an intuitive approach that proceeds by representing every random function $X_{j}$ as a random vector of a chosen number of its principal component scores. Then, the conditional independence graph of the resulting representations is recovered using the joint graphical lasso \citep{danaher2014}, which ensures that the procedure recovers relationships between the different random functions while ignoring those between principal scores corresponding to the same random function. While the method is sensible in its conception, \citet{zapata2022} noted that connecting conditional independence relationships between the random function with the zeros of the precision matrix of their principal component representations seems to require that every random function can be represented as a finite linear combination of a fixed number of deterministic functions with random coefficients. In other words, the functional data has to be exactly finite dimensional. 

Observing that this assumption is impractical and unrealistic for certain applications, \cite{zapata2022} advance a novel assumption of their own, called partial separability, under which they link the conditional independence graph of the random functions with the zero entries of a suitably defined precision matrix, while allowing the data to be infinite dimensional. As an interesting generalization of the separability assumption that is popular in multi-way functional data \citep{aston2017}, partial separability could be of interest even in areas other than functional graphical models. However, it still constitutes a serious structural assumption as it postulates that the covariance operators $\{\bC_{jj}\}_{j=1}^{p}$ of the random functions $\{X_{j}\}_{j=1}^{p}$ are simultaneously diagonalizable, that is, they have the same eigenfunctions. Moreover, evaluating the plausibility of the assumption for a given data set on an intuitive basis is difficult and a statistical test for the same has not yet been developed.

Both of these approaches to functional graphical models are based on functional principal components analysis, treating functional realizations in terms of their truncated principal component representations, and recovering conditional independence relationships between random elements from these representations demands imposing structural assumptions. The necessity of dimensionality reduction can be understood as stemming from the absence of a determinant-like functional on the space of covariance operators. Indeed, every functional defined as the product of eigenvalues must be uniformly zero because the eigenvalues of covariance operators converge to zero by virtue of compactness. In this article, we will present an approach that circumvents this problem by reformulating \eqref{eq:graphical_lasso} in terms of correlation operators, which are operator analogues of correlation matrices, and a regularized infinite-dimensional generalization of the matrix determinant known in functional analysis as the Carleman-Fredholm determinant (or Hilbert-Carleman determinant). Although the use of correlation in place of covariance is not uncommon in functional data \citep{lee2021, li2018} and is, in fact, standard practice for multivariate graphical lasso \citep{kovacs2021}, it arises naturally in our treatment -- namely it emerges naturally from a coherent maximum likelihood philosophy. The key idea is to use the product measure of the ``coordinates'' $\{X_{j}\}_{j=1}^{p}$ as a reference measure. 

\subsection{Related Work}
Other well-known methods of graph recovery for multivariate data already mentioned have also been generalized to he multivariate functional setting. Inverse thresholding \citep{li2018,lee2021}, in this case, requires thresholding the entries of the inverse of a certain correlation operator that can be computed from the data. This requires the said operator to be invertible. Like its multivariate counterpart, and for the same reasons, this method is not expected to perform well in the sparse high-dimensional setting. Naturally, all results proving model selection consistency for this approach assume that $p$ is fixed. 

Functional generalizations of the neighborhood selection approach \citep{kolar2021,lee2022} involve performing appropriately penalized functional regression on each of the random elements against every other random element. Unlike inverse thresholding, these methods do work well in the high-dimensional settings and they also possess the computational advantage of being amenable to parallel implementation. But due to their reliance on functional regression, they require the corresponding regression operators to be Hilbert-Schmidt. This constitutes a substantial structural assumption since regression operators have no reason to be bounded in general \citep{kneip2020}. For example, if one of the random elements is a linear combination of some other random elements the corresponding regression operator will be proportional to identity, which is not a Hilbert-Schmidt operator. 

A number of constributions in the literature, including those discussed here, deal with the more complicated setting where $\{X_{j}\}_{j=1}^{p}$ are non-Gaussian or have non-linear relationships, making them very different in flavor from the work presented here. They exhibit a complex variety in the details of the structural assumptions they make. As an interesting development, we mention here \cite{soleadette2022},  but a comprehensive review of these different approaches is beyond the scope of this article.

\subsection{Contributions}

We extend graphical lasso to a general infinite-dimensional Hilbertian setting. Under rather minimal and intuitive functional counterparts of the multivariate assumptions, we prove functional analogues of state-of-the-art results in the form of finite-sample guarantees concerning the family-wise error rate of model selection and the rates of convergence for precision estimation known for multivariate graphical lasso \citep{ravikumar2011}. As a result, we establish model selection consistency. 

Our method can be motivated in a very natural manner from the maximum likelihood principle which is uncommon, to say the least, for functional data, due to the lack of a suitable replacement for the Lebesgue measure in function spaces. In doing so, we demonstrate what might be the right approach to applying likelihood methodology to multivariate functional data. This development could be of wider interest.

Furthermore, we extend classical results concerning the equivalence of graphical lasso to penalized log-likelihood maximization, Kullback-Leibler divergence minimization, and determinant maximization which were known in the multivariate setting, to the general setting of infinite-dimensional Hilbert spaces. From an analytical perspective, methods in functional data analysis are often infinite-dimensional reformulations of their counterparts in multivariate analysis. While reformulating functions such as the trace and Frobenius norm is almost trivial, we show here how to achieve the same for the nontrivial and rather subtle case of the determinant. This is of interest in its own right, given that the determinant is an important measure of the joint dispersion in multivariate analysis and appears in many other problems. 

Our treatment also clarifies certain elements of the multivariate functional data literature. For example, the hitherto ad-hoc concept of correlation operator arises naturally from the likelihood approach, while the assumption of eigenvalue gap, previously made only in order to ensure that the correlation operator is invertible, now admits a concrete interpretation in terms of the supports of the measures involved. Moreover, we show that the inverse zero characterization holds in complete generality, without the need for any structural assumptions.

In the  tradition of simplification by abstraction, the functional graphical lasso also contributes to our understanding of its multivariate counterpart by identifying the analytical properties of the random objects involved, which make the method work. It also suggests new ways of using the graphical lasso in the multivariate setting. In principle, using tools such as kernel or graph embeddings, the method can be extended to other classes of random objects such as distributions and networks, and to nonlinear relationships.

Parallel to its attractive theoretical properties, the approach also has practical merits. It greatly eases the burden of parameter tuning, requiring only the lasso-type penalty parameter to be chosen, which is well-understood and easily interpreted using the method's divergence minimization characterization. The use of truncated representations is also entirely optional (unnecessary in principle but possible if there are computational constraints). Indeed we observe in our simulations that dimensionality reduction can be counterproductive if the underlying functions do not admit efficient representations, as is the case when the sample paths of the random functions $\{X_{j}\}_{j=1}^{p}$ are rough. Furthermore, the coordinate-free operator formulation of the method allows the user to choose whichever discretization scheme they deem fit, be it basis representation, point evaluation or cell averaging for reasons of  representation accuracy or efficiency. For the same reason, working with heterogeneous data, where $\{X_{j}\}_{j=1}^{p}$ comprises of different kinds of random objects such as variables, vectors, curves or surfaces, is as simple as working with homogeneous data. Finally, the plug-in nature of the method permits the user to choose the covariance estimation procedure which is appropriate given the nature of the available observations, thus making it applicable to functional time series and sparsely observed functional data as well.

\subsection{Structure of the Article}

We begin by describing important concepts and introducing our notation in Section \ref{sec:background}. This is followed by the problem formulation and a discussion of the assumptions in Sections \ref{sec:problem} and \ref{sec:assumptions}, respectively. In Section \ref{sec:method}, we describe our methodology, its motivation and interpretations. Our main results, including finite sample results concerning the estimation of the precision operator and model selection consistency, are stated in Section \ref{sec:theory}. The proofs are deferred to the supplementary material. Section \ref{sec:implementation} contains the details of how the method is implemented, and Section \ref{sec:simulations} presents simulation studies to assess our method's performance. 

\section{Background and Notation}
\label{sec:background}

For $1 \leq j \leq p$, let $\cH_{j}$ be separable Hilbert spaces equipped with the inner products $\langle\cdot, \cdot\rangle_{j}$. The subscript $j$ will always be clear from the context and we will avoid writing it explicitly, preferring $\langle f, g\rangle$ instead, for $f, g \in \cH_{j}$. We will denote by $\cH$, the \emph{product Hilbert space} denoted by $\cH_{1} \times \cdots \times \cH_{p}$ or $\times_{j=1}^{p} \cH_{j}$ equipped with the inner product given by
\begin{equation*}
    \langle \bbf, \bbg\rangle = \sum_{j=1}^{p} \langle f_{j}, g_{j} \rangle
\end{equation*}
for $\bbf, \bbg \in \cH$, where $\bbf = (f_{1}, \dots, f_{p})$ and $\bbg = (g_{1}, \dots, g_{p})$. 

\subsection{Operators and Operator Matrices on Hilbert Spaces}
Operators between Hilbert spaces will be denoted using boldface, as in $\bA$ with the corresponding operator norm and adjoint being written as $\|\bA\|$ and $\bA^{\ast}$ as usual. We define the \emph{spectrum} $\sigma(\bA)$ of $\bA$ as the set of $\lambda \in \bbR$ for which the operator $\bA - \lambda\bI$ does not admit a bounded inverse. The notation $\bA^{-1}$ will denote the inverse of the operator $\bA$ or its pseudoinverse, in case it is not invertible.

We will mostly work with spaces of Hilbert-Schmidt operators. The space of Hilbert Schmidt operators on a Hilbert space $\cH$ will be denoted as $\cL_{2}(\cH)$. The space $\cL_{2}(\cH)$ forms a Hilbert space under the inner product $\langle \cdot, \cdot\rangle_{2}$ induced by the Hilbert-Schmidt norm $\|\cdot\|_{2}$ given by 
\begin{equation*}
    \|\bH\|_{2}^{2} = \sum_{j=1}^{\infty} \sigma_{j}^{2}(\bH)
\end{equation*}
where $\{\sigma_{j}\}_{j=1}^{\infty}$ are the singular values of $\bH$, or equivalently, the eigenvalues of $|\bH| = \sqrt{\bH^{\ast}\bH}$. It is a well-known fact that if $\cH$ is the space of square-integrable functions, Hilbert-Schmidt operators can be represented as an integral operators corresponding to square-integrable kernels.

An \emph{operator matrix} is a matrix of the form $\bA = [\bA_{ij}]_{i,j=1}^{p}$ where the $ij$th entries are operators $\bA_{ij} : \cH_{j} \to \cH_{i}$. For an operator matrix $\bA$, we define the diagonal part $\dg \bA$ of $\bA$ as the diagonal matrix $\bD = [\bD_{ij}]_{i,j=1}^{p}$ given by $\bD_{ij} = \bA_{ij}$ for $i = j$ and $\bzero$ otherwise. The off-diagonal part $\bA - \dg \bA$ will be denoted as $\bA_{0}$. Operator matrices can be thought of as operators on the product Hilbert space $\cH$, as given by
\begin{equation*}
    \bA \bbf = \left[\sum_{j=1}^{p}\bA_{ij}f_{j}\right]_{i=1}^{p}
\end{equation*}
for $\bbf = (f_{1}, \dots, f_{p}) \in \cH$. The adjoint of an operator matrix $\bA$ will be denoted as $\bA^{\top}$. The trace $\tr \bA$ and Hilbert-Schmidt norm $\|\bA\|_{2}$ of an operator matrix $\bA = [\bA_{ij}]_{i,j=1}^{p}$ can be written in terms of the traces and Hilbert-Schmidt norms of the entries as 
\begin{eqnarray*}
    \tr(\bA) = \sum_{i=1}^{p} \tr(\bA_{ii}) \quad \mbox{ and } \quad \| \bA \|_{2}^{2} = \sum_{i,j=1}^{p} \| \bA_{ij} \|_{2}^{2}.
\end{eqnarray*}
To mirror the behaviour of Euclidean spaces in the product Hilbert space, we devise some additional norms. The operator counterparts $\|\cdot\|_{2,1}$ and $\|\cdot\|_{2,\infty}$ of the $\ell_{1}$ and $\ell_{\infty}$ norms are given by $\| \bA \|_{2, 1} = \sum_{i,j=1}^{p} \| \bA_{ij} \|_{2}$ and $\| \bA \|_{2, \infty} = \max_{i,j} \| \bA_{ij} \|_{2}$. In the same way, we define the operator analogues of matrix norms: $\tnorm{\bA}_{2, \infty} = \max_{i}\sum_{j} \|\bA_{ij}\|_{2}$ (maximum column sum) and $\tnorm{\bA}_{2, 1} = \max_{j}\sum_{i} \|\bA_{ij}\|_{2}$ (maximum row sum). Note that $\tnorm{\bA}_{2, 1} = \tnorm{\bA^{\top}}_{2, \infty}$ and that these norms are sub-multiplicative (see Appendix). The tensor product $\bA \otimes \bB$ of operator matrices $\bA$ and $\bB$ is defined as the linear map $\bD \mapsto \bB \bD \bA$ and can also be expressed as an array $[\bA_{ij} \otimes \bB_{kl}]_{i,j,k,l=1}^{p}$ of the tensor products of their entries. The action $\bD \mapsto \bB \bD \bA$ can be imitated by a matrix $[\bA_{ij} \otimes \bB_{kl}]_{(i,j),(k,l)}$ (indexed by the pairs $(i,j)$ and $(k,l)$) acting on vectorized version of $\bD = [\bD_{ij}]_{(i,j)}$ (indexed by $(i,j)$). As a result, we can simultaneously think of the tensor product $\bA \otimes \bB$ as a linear map and as a matrix with tensor product entries.

{The symbols $\bI$ and $\bzero$ will denote the identity and zero elements of their ambient spaecs, which will be clear from the context, according to which they can be elements, operators or operator matrices.} 

\subsection{Second-Order Random Elements in Hilbert Space} 
A random element $X$ is said to be second-order if $\bbE[\|X\|^{2}] < \infty$. For such random elements, we can define the mean and the covariance operator
\begin{equation*}
    \bbm = \bbE[X] \quad\mbox{ and }\quad \bC = \bbE[(X - \bbE[X]) \otimes (X - \bbE[X])],
\end{equation*}
respectively. If $\cH$ is the product Hilbert space of certain Hilbert spaces $\cH_{j}$ for $1 \leq j \leq p$, then we can write $X$ as a random tuple, as in $X = (X_{1}, \dots, X_{p})$ where $\{X_{j}\}_{j=1}^{p}$ are jointly distributed random elements on their respective Hilbert spaces. The covariance operator $\bC$ can then be thought of as an operator matrix, as in $\bC = [\bC_{ij}]_{i,j=1}^{p}$ where the $(i,j)$-th entry is given by the operator $\bC_{ij} = \bbE[(X_{i} - \bbm_{i}) \otimes (X_{j} - \bbm_{j})]$. 

By a well-known result of \cite{baker1973}, for every $1 \leq i,j \leq p$ for $i \neq j$ there exists a unique bounded linear operator $\bR_{ij}: \cH_{j} \to \cH_{i}$ with $\|\bR_{ij}\| \leq 1$ such that $\bC_{ij} = \smash{\bC_{ii}^{1/2}\bR_{ij}\bC_{jj}^{1/2}}$ and $\bR_{ij} = \Pi_{i}\bR_{ij}\Pi_{j}$ where $\Pi_{i}$, $\Pi_{j}$ are projections to the closures of the images of $\bC_{ii}$ and $\bC_{jj}$ in their respective co-domains, which is to say that $\bR_{ij}$ maps the closure of the range of $\bC_{jj}$ to that of $\bC_{ii}$. Accordingly, we define the correlation operator matrix as $\bR = \left[ \bR_{ij} \right]_{i,j=1}^{n}$ where $\bR_{ii} = \bI$ and $\bR_{ij} = \smash{\bC_{ii}^{-1/2}\bC_{ij}\bC_{jj}^{-1/2}}$ for $i \neq j$, where $\smash{\bC_{ii}^{-1/2}}$ and $\smash{\bC_{jj}^{-1/2}}$ are understood to be the operator pseudoinverses of $\bC_{ii}^{1/2}$ and $\bC_{jj}^{1/2}$, respectively. It can be shown that $\bR$ is always a positive semi-definite operator. Theorem \ref{thm:CI-elements} specifies a sufficient condition on the random element $X$, under which $\bR$ is strictly positive-definite and hence invertible. 
 
If $\bR$ is invertible, we can write its inverse as $\bR^{-1} = \bI + \bH$ where $\bI$ is understood as the identity operator matrix and $\bH = [\bH_{ij}]_{i,j=1}^{p}$ is a bounded operator matrix. We will refer to $\bH$ as the \emph{precision operator matrix} or simply, the \emph{precision operator} of $X$. 

We will describe the dispersion of the distributions of our random elements using the notion of sub-Gaussian and sub-exponential norms of random variables. The sub-Gaussian and sub-exponential norms of a random variable $Z$ are respectively given by
\begin{equation*}
    \| Z \|_{\psi_{2}} = \inf \{ t > 0: \bbE\big[\exp(Z^{2}/t^{2})\big] \leq 2 \} \mbox{ and }
    \| Z \|_{\psi_{1}} = \inf \{ t > 0: \bbE\big[\exp(|Z|/t)\big] \leq 2 \}.
\end{equation*}

\subsection{Conditional Independence Graphs of Random Elements}

Let  $X = (X_{1}, X_{2}, X_{3})$ be random element in the product Hilbert space $\cH = \cH_{1} \times \cH_{2} \times \cH_{3}$. We say that $X_{1}$ and $X_{2}$ are conditionally independent given $X_{3}$, or alternatively, $X_{1} \CI X_{2} ~|~ X_{3}$ if the conditional measures $\Prob_{X_{1}|X_{3}}$, $\Prob_{X_{2} | X_{3}}$ and $\Prob_{X_{1}, X_{2} | X_{3}}$ satisfy
\begin{equation*}
    \Prob_{X_{1}, X_{2} | X_{3}} = \Prob_{X_{1} | X_{3}} \otimes \Prob_{X_{2} | X_{3}}. 
\end{equation*}
The $\sigma$-algebra generated by the random variables $\{\langle h, X_{3} \rangle: h \in \cH_{3}\}$ is same as the Borel $\sigma$-algebra associated with $\cH_{3}$ \cite[Theorem 7.1.1]{hsing2015}. Therefore, the above statement can be interpreted in terms of the familiar notion of conditional independence for real-valued random variables as follows: for every $f \in \cH_{1}$ and $g \in \cH_{2}$, we have that
\begin{equation}\label{eqn:ci_rvs}
    \langle f, X_{1} \rangle \CI \langle g, X_{2} \rangle ~|~ \{\langle h, X_{3} \rangle: h \in \cH_{3}\}.
\end{equation} 
Thus $X_{1}$ and $X_{2}$ are conditionally independent given $X_{3}$ if and only if any two linear functionals of $X_{1}$ and $X_{2}$ are conditionally independent given every linear functional of $X_{3}$. If $X$ is second-order, we can define a purely second-order counterpart of the notion of conditional independence called \emph{conditional uncorrelatedness} which we denote as $X_{1} \CI_{2} X_{2} ~|~ X_{3}$ and define as 
\begin{equation}\label{eqn:ci_rvs1}
    \bbE\Big[\Big(\langle f, X_{1} \rangle - \bbE_{2}\big[\langle f, X_{1} \rangle \big| L(X_{3})\big]\Big) 
    \Big(\langle g, X_{2} \rangle - \bbE_{2}\big[\langle g, X_{2} \rangle \big| L(X_{3})\big]\Big)\Big] = 0
\end{equation} 

or equivalently,
\begin{equation}\label{eqn:ci_rvs2}
    \bbE\Big[\langle f, X_{1} \rangle\langle g, X_{2} \rangle\Big] = \bbE\Big[\bbE_{2}\big[\langle f, X_{1} \rangle ~|~ L(X_{3})\big] \cdot  
    \bbE_{2}\big[\langle g, X_{2} \rangle ~|~ L(X_{3})\big]\Big]
\end{equation} 
where $\bbE_{2}[Z ~|~ L(X_{3})]$ denotes the best linear unbiased predictor of the random variable $Z$ from the closed linear span $L(X_{3})$ of the random variables $\{\langle h, X_{3} \rangle: h \in \cH_{3}\}$. For zero-mean Gaussian random elements, the two notions of independence and uncorrelatedness coincide (see \cite{loeve2017}).

Consider a random element $X = (X_{1}, \dots, X_{p})$ on a product Hilbert space $\cH$. Let $G$ be an undirected graph with the vertex set $\{1, \dots, p\}$. By convention, every vertex is understood to be adjacent to itself. We say that $X$ has the graph $G$ if it satisfies the \emph{pairwise Markov property}, that is, for every $1 \leq i, j \leq p$ such that $i$ and $j$ are not adjacent in $G$ we have  
\begin{equation}\label{eqn:pairwise-markov}
        X_{i} \CI X_{j} ~|~ X_{k} : k \neq i, j
\end{equation}
or equivalently, we have
\begin{equation*}
    \langle f_{i}, X_{i} \rangle \CI \langle f_{j}, X_{j} \rangle ~|~ \{ \langle f_{k}, X_{k} \rangle : f_{k} \in \cH_{k}, k \neq i,j \}
\end{equation*}
for every $f_{i} \in \cH_{i}$ and $f_{j} \in \cH_{j}$. We refer to $G$, thus defined, as the conditional independence graph of $X$. For a second-order $X$, we can similarly define the conditional uncorrelation graph of $X$, by simply replacing $\CI$ in (\ref{eqn:pairwise-markov}) with $\CI_{2}$. We will see eventually that the graph of a second-order random element $X$ is intimately related to the entries of the precision operator matrix $\bH$.

\subsection{The Carleman-Fredholm Determinant}

In order to properly generalize the graphical lasso to random elements, we will need to reformulate the graphical lasso objective function in terms of operator matrices. The extension of the terms of the objective corresponding to trace and  $\ell_{1}$ penalty is straightforward, the determinant term determinant is  more involved and non-standard:
\begin{definition}
    Let $\bH \in \cL_{2}(H)$ with  eigenvalues $\{\lambda_{j}\}_{j=1}^{\infty}$. We define the Carleman-Fredholm determinant of $\bH$ as 
    \begin{equation}
        \det\nolimits_{2} (\bI + \bH) = \prod_{j=1}^{\infty} (1 + \lambda_{j})e^{-\lambda_{j}}
    \end{equation}
\end{definition}
It can be shown that the infinite product converges when $\sum_{j=1}^{\infty} \lambda_{j}^{2} < \infty$ and thus, the Carleman-Fredholm determinant is well-defined for all Hilbert-Schmidt operators. It is also known that the map $\bH \mapsto \det\nolimits_{2} (\bI + \bH)$ is strictly log-concave, continuous everywhere in $\|\cdot\|_{2}$ norm and Gateaux differentiable on $\{\bH: -1 \notin \sigma(\bH)\} \subset \cL_{2}(H)$ (see the Appendix). Note that it is not simply the product of the eigenvalues $1+\lambda_{j}$ of $\bI + \bH$. Defining the determinant simply as the product of the eigenvalues leads to what is known as the \emph{Fredholm determinant}, which is defined only for trace-class operators. We will see that the Carleman-Fredholm determinant appears most naturally when one attempts to correctly generalize the multivariate graphical lasso optimization function to covariance operators. For a more in-depth discussion on the generalization of determinants to operators, the interested reader is invited to consult \cite{gohberg2012} and \cite{simon1977}.

\section{Problem Statement}\label{sec:problem}

Let $X = (X_{1}, \dots, X_{n})$ be a second-order random element in $\cH$ and $\smash{\{X^{k}\}_{k=1}^{n}}$ be (not necessarily independent) realizations of $X$. Given an estimate $\hat{\bC} = \hat{\bC}_{n}(X^{1}, \dots, X^{n})$ of the covariance $\bC$ of $X$, we are interested in estimating the graph $G$ of the random elements $\{X_{j}\}_{j=1}^{p}$, which is given by the adjacency matrix $A = [A_{ij}]_{i,j=1}^{p}$ where
\begin{equation*}
    A_{ij} = \begin{cases}
        1 &\mbox{ if } i = j,\\
        1_{\{\bH^{\ast}_{ij} \neq \bzero\}} &\mbox{ otherwise}
    \end{cases}
\end{equation*}
and $\bH^{\ast} = \bR^{-1} - \bI$, with $\bR$ being the correlation operator matrix of $X$. Essentially, we are interested in determining the non-zero off-diagonal entries of $\bH^{\ast}$. Of particular interest is the sparse high-dimensional setting, where the number $p$ of random elements $X_{j}$ can be much larger than the  number $n$ of samples, and the graph $G$ is known to be sparse in the sense that the maximum degree $d$ of a vertex in $G$ is much smaller than $p$.

When $X$ is Gaussian, the graph $G$ is identical to the conditional independence graph of $\{X_{j}\}_{j=1}^{p}$ and if $X$ is not Gaussian, we can still interpret the off-diagonal zero entries of $\bH$ in terms of the alternative notion of conditional uncorrelatedness (see Theorem \ref{thm:CI-elements}). In either case, the graph describes the dependence structure of the random elements $\{X_{j}\}_{j=1}^{p}$ in the following sense: $X_{i}$ is adjacent to $X_{j}$ if and only if $X_{j}$ can tell us something about $X_{i}$ that other elements $\{X_{k}\}_{k\neq i,j}$ put together cannot. We use linear relationships between the random elements to judge what they tell us about each other and therefore, the graph $G$ can be regarded as the graph of linear relationships between $X_{j}$. For Gaussian $X$, the relationships are always linear and as a result the graph $G$ is equal to the conditional independence graph.  

\section{Assumptions}
\label{sec:assumptions}

In this section, we discuss the relatively small number of conditions we assume to prove the consistency and rates of convergence of functional graphical lasso. In particular, we will attempt to explain how they can be interpreted in terms of properties of the distribution of $X$ and how they may break down in certain cases. 

\subsection{Equivalence}

The \emph{support} of a measure (denoted by $\supp$) is the largest closed set such that each of its open subsets have positive measure. Two measures are said to be \emph{equivalent}, if they have the same support, and \emph{singular}, if they have disjoint supports. In general, it is possible for two measures to be neither equivalent nor singular. But according to the Feldman-H\'{a}jek theorem, two Gaussian measures on a locally convex space must be either equivalent or singular.

Let $\bbP_{X_{j}}$ and $\bbP_{X}$ denote the random measures corresponding to $X_{j}$ in the space $\cH_{j}$ for every $1 \leq j \leq p$ and $X$ in the product space $\cH = \otimes_{j=1}^{p} \cH_{j}$ respectively. We can also view the components $X_{j}$ of $X$ separately, and they would correspond to the product measure $\otimes_{j=1}^{p} \bbP_{X_{j}}$. If $X$ is Gaussian, we will make the following assumption:
\begin{assumption}[Equivalence]\label{asm:equiv}
$\bbP_{X}$ is equivalent to the product measure $\otimes_{j=1}^{p} \bbP_{X_{j}}$, that is,
\begin{equation*}
    \supp ~\bbP_{X} = \supp ~\otimes_{j=1}^{p} \bbP_{X_{j}}.
\end{equation*}
\end{assumption}
According to Corollary 6.4.11. of \cite{bogachev1998}, this seemingly innocuous statement is actually equivalent to saying that: (a) the off-diagonal entries of the correlation operator matrix $\bR$ are Hilbert-Schmidt and (b) that there is a gap between the eigenvalues of $\bR$ and $0$, that is, $1 + \inf\nolimits_{j} \lambda_{j}(\bR_{0}) > 0.$ This ensures that the correlation operator matrix $\bR$ is invertible and that the operator $\bH = \bR^{-1} - \bI$ is Hilbert-Schmidt (Lemma \ref{lem:H-is-Hilbert-Schmidt}), implying that our optimization functional is well-defined at $\bH$. It is important to note that for us this is a consequence of a ``first principles" assumption imposed upon the observed random element $X$ itself, namely the support condition. In contrast, previously this was an operationally convenient assumption imposed upon intermediate quantities such as $\bR$ so as to make certain operations (such as operator inversion or evaluation of the Hilbert-Schmidt norm) well-defined. 

If $X$ is not Gaussian, we will assume the properties (a) and (b) directly through the following assumption instead:
\begin{assumption-again}[Eigenvalue Gap]
The cross-correlation operator matrix $\bR_{0} = \bR - \bI$ is Hilbert-Schmidt and the eigenvalues $\{ \lambda_{j} (\bR_{0}) \}_{j=1}^{\infty}$ of 
$\bR_{0}$ satisfy
\begin{equation*}
    1 + \inf\nolimits_{j} \lambda_{j}(\bR_{0}) > 0.
\end{equation*}
\end{assumption-again}
Define $\rho = 1 + \tnorm{\bR_{0}}_{2, \infty}$. It is worth pointing out that the assumption of an eigenvalue gap is not so stringent considering that $1 + \lambda_{k}(\bR_{0}) \geq 0$ for $k \geq 1$ anyway since $\bR$ is non-negative and $\lambda_{k}(\bR_{0}) \to 0$ as $k \to \infty$ because $\bR_{0}$ is Hilbert-Schmidt.
\begin{lemma}\label{lem:H-is-Hilbert-Schmidt}
    If $1 + \inf\nolimits_{j} \lambda_{j}(\bR_{0}) > 0$ and $\bR_{0} = \bR - \bI$ is Hilbert-Schmidt, then so is $\bH^\ast = \bR^{-1} - \bI$.
\end{lemma}
\begin{proof}[Proof of Lemma \ref{lem:H-is-Hilbert-Schmidt}]
Let $c = 1 + \inf\nolimits_{j} \lambda_{j}(\bR_{0})$. By the spectral mapping theorem, $\lambda_{k}(\bH^\ast) = [1 + \lambda_{k}(\bR_{0})]^{-1} - 1 = - \lambda_{k}(\bR_{0})[1 + \lambda_{k}(\bR_{0})]^{-1}$ and therefore, 
\begin{equation*}
    \|\bH^\ast\|_{2}^{2} = \sum_{k=1}^{\infty} \frac{\lambda_{k}^{2}(\bR_{0})}{[1 + \lambda_{k}(\bR_{0})]^{2}} 
    \leq \frac{1}{c^{2}} \sum_{k=1}^{\infty} \lambda_{k}^{2}(\bR_{0}) 
    = \frac{\|\bR_{0}\|_{2}^{2}}{c^{2}} < \infty.
\end{equation*}
\end{proof}

In fact, we will see in Section \ref{sec:penalized-log} that Assumption \ref{asm:equiv} allows us to treat the product measure $\otimes_{j=1}^{p} \bbP_{X_{j}}$ as a \emph{reference measure} to describe the distribution of $X$ much like the Lebesgue measure serves to do the same in Euclidean spaces.

\begin{remark}
It is not difficult to imagine a scenario where Assumption \ref{asm:equiv} fails to hold. Consider a Gaussian process $X$ on the unit interval $[0, 1]$ with continuous sample paths which corresponds to a Gaussian measure in the space $L^{2}[0, 1]$. Then $X$ can be thought of as a pair $(X_{1}, X_{2})$ where $X_{1}$ and $X_{2}$ are the processes (and random elements) corresponding to the restrictions of $X$ to the intervals $[0, 1/2]$ and $(1/2, 1]$ respectively. Let $Y$ be the random element (process) corresponding to the product measure $\bbP_{X_{1}} \otimes \bbP_{X_{2}}$. The sample paths of the process $Y$ are almost surely discontinuous at $t = 1/2$ while that of the process $X$ are almost surely continuous throughout. Thus the measures $\bbP_{X}$ and $\bbP_{X_{1}} \otimes \bbP_{X_{2}}$ are singular.
\end{remark}

\begin{remark}
\label{rmk:lee-assumption}
Assumption $1^{\ast}$ here is strictly weaker than Assumption 1 of \cite{lee2022} which states that for every $1 \leq i \leq p$, the regression operator $\bC_{-i,-i}^{\dagger}\bC_{-i,i}^{}$ is Hilbert-Schmidt (see the appendix for a proof). 
\end{remark}

\subsection{Incoherence}

Let $\Gamma$ denote the outer product of the operator matrix $\bR$ with itself, i.e.
\begin{equation*}
    \Gamma = \bR \otimes \bR = \left[ \bR_{ij} \otimes \bR_{kl} \right]_{i,j,k,l = 1}^{p}.
\end{equation*}
Equivalently, $\Gamma$ can be thought of as an operator matrix indexed by the pairs $(i, j)$ and $(k, l)$ with $\Gamma_{(i, j)(k, l)} = \bR_{ij} \otimes \bR_{kl}$ just as $\bR$ is indexed by the vertices $i, j$ in $\bR_{ij}$. For two sets $A$ and $B$ of vertex pairs we can write the \emph{submatrix}  $\Gamma_{AB}$ as
\begin{equation*}
    \Gamma_{AB} = \left[ \Gamma_{(i, j)(k, l)} \right]_{(i,j) \in A, (k,l) \in B}.
\end{equation*}
Finally, observe that $\Gamma$ can be thought of as an operator on the product space of the tensor product spaces $\cH_{i} \otimes \cH_{j}$ with $\Gamma \bA = \bR \bA \bR$. Because $\bR$ is invertible (under Assumption \ref{asm:equiv}), it follows that so is $\Gamma$ with $\Gamma^{-1} = \bR^{-1} \otimes \bR^{-1}$. 

Let $S$ denote the set of $(i, j)$ such that $i$ and $j$ are adjacent in $G$ or equivalently, $(i, j)$ corresponds to an edge. Naturally, $S^{c}$ denotes its complement. Then $\Gamma_{SS}$ can be shown to be invertible by virtue of being a principal submatrix of $\Gamma$. The following assumption will serve as the functional analogue of the familiar mutual incoherence condition from \cite{ravikumar2011}.
\begin{assumption}[Incoherence]\label{asm:incoherence}
For some $\alpha > 0$, we have 
\begin{equation}\label{eqn:incoherence}
    \max_{e \in S^{c}} \|\Gamma_{eS}^{\phantom{-1}}\Gamma_{SS}^{-1}\|_{2, 1} \leq 1 - \alpha.
\end{equation}
\end{assumption}
Notice that like $\bR$, the inverse of submatrix $\Gamma_{SS}$ can be written as the sum of identity and a Hilbert-Schmidt operator matrix. Indeed, for $\bA = (\bR_{0})_{S}$, where $\bR_{0}=\bR - \mathbf I$, we can write
\begin{equation*}
    \Gamma_{SS}^{-1} = \left[\bI \otimes \bI + \bI \otimes \bA + \bA \otimes \bI + \bA \otimes \bA\right]^{-1} = \bI + \Lambda_{SS}
\end{equation*}
where $\Lambda_{SS}$ is Hilbert-Schmidt by Lemma \ref{lem:H-is-Hilbert-Schmidt} because the operator $\bI \otimes \bA + \bA \otimes \bI + \bA \otimes \bA$ inside the inverse is itself Hilbert-Schmidt. Define $\gamma = 1 + \tnorm{\Lambda_{SS}}_{2, \infty}$.

Intuitively speaking, if we could think of $\bR$ as the covariance operator of a zero mean random element $Z = (Z_{j})_{j=1}^{p}$ with $\bR_{ij} = \bbE\left[ Z_{i} \otimes Z_{j} \right]$, we would consider the random elements
\begin{equation*}
    Y_{(i, j)} = Z_{i} \otimes Z_{j} - \bbE\left[ Z_{i} \otimes Z_{j} \right]
\end{equation*}
for $1 \leq i, j \leq p$. Using the same tools as in multivariate analysis (Taylor expanding the moment generating function), it can be shown that $\Gamma_{(i,j)(k, l)} = \bbE\left[Y_{(i,j)} \otimes Y_{(k, l)}\right]$ when $Z_{j}$ are Gaussian. Let $Y_{S} = \{Y_{e}: e \in S\}$. Assumption \ref{asm:incoherence} can now be expressed as
\begin{equation*}
    \max_{e \in S^{c}} \left\| \bbE\left[ Y_{e}^{\phantom{.}} \otimes Y_{S} \right] \bbE\left[ Y_{S}^{\phantom{.}} \otimes Y_{S} \right]^{-1} \right\|_{2, 1} \leq 1 - \alpha.
\end{equation*}
Notice that $\Gamma_{eS}$ is the cross-covariance of $Y_{e}$ with $Y_{S}$ and $\Gamma_{SS}$ is the covariance of $Y_{S}$. If we were to find the best linear predictor of $Y_{e}$ using $Y_{S}$, the linear coefficients would be given by $\Gamma_{eS}^{\phantom{-1}}\Gamma_{SS}^{-1}$. Assumption \ref{asm:incoherence} is essentially saying that these coefficients cannot be to large: none of the ``non-edges" $Y_{e}$ (with $e \in S^{c}$) are highly correlated with the ``edges" $Y_{S}$ and therefore one cannot predict the ``non-edges" $Y_{e}$ from the ``edges" $Y_{S}$ too well.

Of course, strictly speaking, the operator $\bR$ is not the covariance operator of any second-order random element in Hilbert space due to not being trace-class. We believe that this explanation can be made rigorous by treating $\bR$ as the covariance operator of Gaussian random element on a suitably chosen locally convex topological vector space (where covariance operators do not have to be trace-class). Even without the technical details formalizing this, the intuition remains useful.

Incoherence is the assumption that enables us to exploit sparsity. It seems that incoherence is an indispensable assumption for the multivariate graphical lasso and weaker assumptions lead to substantially weaker rates of convergence (\cite{rothman2008}). One can only expect the same of the functional graphical lasso.

\subsection{Regularity}

To execute our method, we will need to estimate the correlation operator matrix $\bR$. We do this by first estimating the covariance operator $\bC$ and then solving the following linear problem for $\bR$:
\begin{equation}\label{eqn:correln-defn}
    [\dg \bC]^{1/2} \bR [\dg \bC]^{1/2} = \bC.
\end{equation}
Note that the problem is ill-posed because $\bC$ (and $\dg \bC$) are compact operators. To ensure reasonable rates of convergence for this estimation procedure we need to impose the following condition on $\bR$:
\begin{assumption}[Regularity]\label{asm:regularity}
    For some $0 < \beta \leq 1$, we have $\bR_{0} = [\dg \bC]^{\beta} \Phi_{0} [\dg \bC]^{\beta}$ for some Hilbert-Schmidt operator matrix $\Phi_{0}$, whose diagonal entries are zero.
\end{assumption}
In principle, it is possible that $\bR_{0} = [\dg \bC]^{\beta} \Phi_{0} [\dg \bC]^{\beta}$ with $\beta > 1$ but for our purpose this situation is essentially identical to the case $\beta = 1$. Note that for $0 < \beta' < \beta$, the condition $\bR_{0} = [\dg \bC]^{\beta} \Phi_{0} [\dg \bC]^{\beta}$ implies $\bR_{0} = [\dg \bC]^{\beta'} \Phi_{0}' [\dg \bC]^{\beta'}$ for some $\Phi_{0}'$, and therefore, Assumption \ref{asm:regularity} holds for $\beta = 1$ if it holds for $\beta > 1$. Note that Assumption \ref{asm:regularity} is equivalent to saying that for every $i \neq j$,
\begin{equation*}
    \sum_{k,l=1}^{\infty} \left[ \frac{1}{\mu_{k}\lambda_{l}} \right]^{1+2\beta} |\langle e_{k}, \bC_{ij}f_{l}\rangle|^{2} < \infty
\end{equation*}
for some $\beta > 0$, where $\{(\mu_{k}, e_{k})\}_{k=1}^{\infty}$ and $\{(\lambda_{l}, f_{l})\}_{l=1}^{\infty}$ are the eigenpairs of $\bC_{ii}$ and $\bC_{jj}$, respectively. Essentially, this means that $\bC_{ij}$ admits an efficient or sparse representation in the eigenbases of $\bC_{ii}$ and $\bC_{jj}$.

In fact, it is a classical result in inverse problem theory that in the absence of such \emph{source conditions}, the rate of convergence for the solution of an infinite-dimensional linear inverse problem can be arbitrarily slow (\cite{hanke2017}). In the language of numerical linear algebra, Assumption \ref{asm:regularity} means that the operator $\bC$ is intrinsically preconditioned for inversion by $\dg \bC$. Moreover, the usage of such regularity conditions is standard in the literature \citep[c.f.][]{li2018}. We will see that the performance of our procedure will depend critically on the maximum degree $d$ of the graph and that this dependence is mediated by $\beta$.

\section{Methodology and Philosophy}
\label{sec:method}
We now describe our two-step estimation procedure to recover the graph $G$ of $X$ given an estimate $\hat{\bC}$ of the covariance $\bC$ of $X$.

Firstly, we estimate the correlation operator matrix $\bR$ of $X$.
Because $\bC$ is known only approximately, estimating $\bR$ using Equation (\ref{eqn:correln-defn}) presents an ill-posed linear problem. We use the regularized estimator $\hat{\bR} = [\hat{\bR}_{ij}]_{i,j=1}^{p}$ given by
\begin{equation*}
    \hat{\bR}_{ij} = \begin{cases}
    \bI &\mbox{ for } i = j, \mbox{ and }\\
    [\epsilon_{n}\bI + \dg \hat{\bC}_{ii}]^{-1/2}\hat{\bC}_{ij}[\epsilon_{n}\bI + \dg \hat{\bC}_{jj}]^{-1/2} &\mbox{ for } i \neq j,
    \end{cases}
\end{equation*}
where $\epsilon_{n}$ serves as a tuning parameter. Recall that the diagonal entries $\bR_{jj}$ are all equal to $\bI$, so we need not burden ourselves with their estimation.

Secondly, we minimize the proposed objective functional $\cF$ over the space of Hilbert-Schmidt operators $\bH$ on $\cH$ given by
\begin{equation}\label{eqn:opt-problem}
    \cF[\bH] = \begin{cases}
        \tr(\bH\hat{\bR}_{0}) - \log \det\nolimits_{2} (\bI + \bH) + \lambda_{n}\|\bH_{0} \|_{2, 1} &\mbox{ if } \bI + \bH > \bzero, \mbox{ and}  \\
        \infty & \mbox{ otherwise.}
    \end{cases}
\end{equation}
where the trace $\tr(\bH\hat{\bR}_{0})$ can be expressed as $\sum_{i \neq j} \tr(\bH_{ij}\hat{\bR}_{ij}^{\ast})$ and $\|\bH_{0} \|_{2, 1} = \sum_{i \neq j} \| \bH_{ij} \|_{2}$ is the $\ell_{1}$-norm of the Hilbert-Schmidt norms of the off-diagonal entries of $\bH$ and can be likened to a group lasso penalty proposed in \cite{yuan2006}. Note that the trace is well-defined since both $\bH$ and $\hat{\bR}_{0}$ are Hilbert-Schmidt implying that the product $\bH\hat{\bR}_{0}$ is trace-class. Thus, $\cF[\bH]$ is well-defined for a Hilbert-Schmidt operator $\bH$. Furthermore as an eigenvalue of $\bH$ approaches $-1$ from above, $\det\nolimits_{2} (\bI + \bH)$ converges to $0$ and its logarithm grows without bound implying that $\cF[\bH] \to \infty$. The piece-wise definition is thus quite reasonable and in fact, makes $\cF$ a coercive, strictly convex functional which is continuous in the extended sense. This will ensure that $\cF$ always has a unique minimum and minimizer (Theorem \ref{thm:existence}). The nonzero entries of the minimizer $\hat{\bH} = \argmin_{\bH} \cF[\bH]$ describe the graph $G$ in that $\hat{\bH}_{ij} \neq \bzero$ if and only if $i$ and $j$ are adjacent. 

To summarize, given an estimator $\hat{\bC}$ of the covariance $\bC$ we have the following procedure to estimate the graph $G$ of $X$:
\begin{enumerate}[label = {\bfseries Step \arabic*.}, leftmargin = 2cm] 
    \item \textbf{Estimation.} Estimate     
    the correlation operator matrix $\hat{\bR}$ as follows:
    \begin{eqnarray}        
        &\hat{\bR} &= \bI + [\epsilon_{n}\bI + \dg \hat{\bC}]^{-1/2}\hat{\bC}_{0}[\epsilon_{n}\bI + \dg \hat{\bC}]^{-1/2}.
    \end{eqnarray} 
    \item \textbf{Minimization.} Compute $\hat{\bH}$: 
    \begin{eqnarray}\label{eqn:optimisation-problem}
        &\hat{\bH} &= \argmin \left[\tr(\bH\hat{\bR}_{0}) - \log \det\nolimits_{2} (\bI + \bH) + \lambda_{n}\|\bH_{0} \|_{2, 1}\right],      
    \end{eqnarray}
where the minimum is taken over all Hilbert-Schmidt operators $\bH$ such that $\bI + \bH > \bzero$. The adjacency matrix $\hat{A} = [\hat{A}_{ij}]_{i,j=1}^{p}$ of the estimate $\hat{G}$ of the graph $G$ is given by $\hat{A}_{ij} = 1_{\{\hat{\bH}_{ij} \neq \bzero\}}$ for $i \neq j$ and $1$ otherwise.
\end{enumerate}

A variety of methods have been used in the literature for choosing the tuning parameter $\epsilon_{n}$. For example, \cite{li2018} use generalized cross validation. We prefer a simpler approach inspired by \cite{waghmare2023}. 
The tuning parameter $\lambda_{n}$ can be chosen using stability selection proposed by \cite{meinshausen2010}. We show in Section \ref{sec:finite-sample-theory}, that the rate at which $\epsilon_{n} \to 0$ and $\lambda_{n} \to 0$ is related to how well $\hat{\bC}$ concentrates around $\bC$ together with the regularity $\beta$.

In Section \ref{sec:implementation}, we describe how the quantities involved are actually calculated in practice and how the minimization procedure is implemented using the Alternating Direction Method of Multipliers (ADMM) algorithm.

\subsection{Penalized Log-likelihood Maximization}\label{sec:penalized-log} 
The use of likelihood  techniques in functional data analysis is largely impeded by the absence of a compelling reference measure in infinite dimensions, playing the role of Lebesgue measure. The latter serves as a \emph{de facto} reference measure in finite dimensions due to its translation invariance and  accompanying \emph{indifference} to points in the space. However, translation invariance forces a Borel measure on an infinite-dimensional Banach space to assign infinite measure to every open set. 

We propose to use the product $\bbQ = \otimes_{j=1}^{p} \bbP_{X_{j}}$ as our reference measure for multivariate functional data. It is not translation invariant, but unlike Lebesgue measure, it is a probability measure representing an actually possible scenario, i.e. when $X_{j}$ are all independent --- in finite dimensions the product measure is equivalent to Lebesgue measure. Of course, we do not know $\bbQ$ a priori but it turns out that we can, in a certain sense, evaluate the corresponding log-likelihood without knowing $\bbQ$ exactly using what amounts to a renormalization technique.

Let $\bbP$ and $\tilde{\bbP}$ be two zero-mean Gaussian measures with the marginals $\{\bbP_{X_{j}}\}_{j=1}^{p}$ which are equivalent to $\bbQ = \otimes_{j=1}^{p} \bbP_{X_{j}}$. By Corollary 6.4.11 of \cite{bogachev1998}, we can write the covariance operators of $\bbP$ and $\tilde{\bbP}$ as
\begin{eqnarray*}
    &\bC_{\bbP} &= \bC_{\bbQ}^{1/2} (\bI + \bR_{0}) \bC_{\bbQ}^{1/2} \\
    &\bC_{\tilde{\bbP}} &= \bC_{\bbQ}^{1/2} (\bI + \tilde{\bR}_{0}) \bC_{\bbQ}^{1/2}
\end{eqnarray*}
where $\bC_{\bbQ} = \dg \bC = \dg \bC_{\bbP} = \dg \bC_{\tilde{\bbP}}$, and $\bR_{0}, \tilde{\bR}_{0}$ are Hilbert-Schmidt operators with diagonal entries all zero and eigenvalues separated from $-1$ in the sense of Assumption $1^{\ast}$. 
The average log-likelihood of $\bbP$ with respect to $\bbQ$ evaluated with an infinite number of samples drawn from $\tilde{\bbP}$ evaluates to the following expectation:
\begin{lemma}\label{thm:likelihood}
We have
\begin{equation}\label{eqn:likelihood}
    \int \log \left[ \frac{d\bbP}{d\bbQ} \right] d\tilde{\bbP} = -\frac{1}{2} [\tr(\bH\tilde{\bR}_{0}) - \log \det\nolimits_{2}(\bI + \bH)]
\end{equation}
where $\bH = (\bI + \bR_{0})^{-1} - \bI$.
\end{lemma}
If we think of $\tilde{\bbP}$ as the empirical measure generated from the samples $\{X^{k}\}_{k=1}^{n}$ drawn from $\bbP$, the expression on the left of Equation (\ref{eqn:likelihood}) is exactly the log-likelihood of $\bbP$ with respect to $\bbQ$. Of course, we require $\tilde{\bbP}$ to be a Gaussian measure which an empirical measure cannot be. 
So, we treat $\tilde{\bbP}$ as the zero-mean Gaussian measure with the covariance operator $\bC_{\tilde{\bbP}} \approx \hat{\bC}$. 
Roughly speaking, this would mean that $\tilde{\bR}_{0} \approx \hat{\bR}_{0}$ and the right hand side of Equation (\ref{eqn:likelihood}) becomes
\begin{equation}\label{eqn:likelihood-1}
    -\frac{1}{2} [\tr(\bH\hat{\bR}_{0}) - \log \det\nolimits_{2}(\bI + \bH)]
\end{equation}
which makes for a compelling substitute for the sought after log-likelihood and corresponds to the first two terms of our objective functional $\cF[\bH]$. 
The idea of using a Gaussian measure $\tilde{\bbP}$ corresponding approximately to the empirical covariance operator $\hat{\bC}$ instead of the empirical measure to evaluate the log-likelihood is reminiscent of the idea behind the parametric bootstrap. 

Essentially, we have conditioned the log-likelihood on the prior knowledge that $\tilde{\bbP}$ is a Gaussian. 
Deriving expression (\ref{eqn:likelihood-1}) is also possible with a more direct approach using the empirical measure, but requires using an onerous amount of unpleasant renormalization techniques such as truncation and regularization needed to deal with the infinities arising from the difference in the supports of the empirical and true versions of the measures involved, all of which we manage to avoid here.

Since we know that the sparse edge structure of the graph is encoded in the non-zero off-diagonal entries of the Hilbert-Schmidt operator matrix $\bH$, it seems natural to penalize (\ref{eqn:likelihood-1}) with the $\ell_{1}$ norm of the norms $\|\bH_{ij}\|_{2}$. This gives
\begin{equation}\label{eqn:likelihood-2}
    -\frac{1}{2}\cF[\bH] = -\frac{1}{2} \left[\tr(\bH\hat{\bR}_{0}) - \log \det\nolimits_{2}(\bI + \bH)\right] - \frac{\lambda_{n}}{2} \sum_{i \neq j} \| \bH_{ij} \|_{2}.
\end{equation}
We have thus shown that our method can be understood as penalized log-likelihood maximization. 

\subsection{Constrained Divergence Minimization} 

The continuity and convexity of the functional $\cF$ implies that the optimization problem (\ref{eqn:optimisation-problem}) has an equivalent dual formulation.
In fact, minimizing $\cF$ is actually equivalent to evaluating the convex conjugate $\cG^{\ast}$ of the functional $\cG[\bA] = \log\det\nolimits_{2}(\bI + \bA) - \lambda_{n}\|\bA_{0}\|_{2,1}$ at $-\hat{\bR}_{0}$. Indeed,
\begin{eqnarray*}
    &\min_{\bA} \cF[\bA] &= - \max_{\bA} \left[ \tr(\bA(-\hat{\bR}_{0})) + \log\det\nolimits_{2}(\bI + \bA) - \lambda_{n}\|\bA_{0}\|_{2,1} \right] \\
    &&= - \cG^{\ast}[-\hat{\bR}_{0}]. 
\end{eqnarray*}
Now, the convex conjugate of $\bA \mapsto - \log\det\nolimits_{2}(\bI + \bA)$ is actually twice the Kullback-Leibler divergence of the Gaussian measure with the correlation operator $\bI -\bB$ assuming $\dg \bB = \bzero$. This can be verified from Equation (\ref{eqn:likelihood}) which yields for $\tilde{\bbP} = \bbP$, the Kullback-Leibler divergence $\cD[\bR_{0}]$ of a Gaussian measure $\bbP$ with correlation operator $\bR = \bI + \bR_{0}$ with respect to its product measure $\bbQ$, to be
\begin{equation*}
    \cD[\bR_{0}] = - \frac{1}{2} \log \det\nolimits_{2}(\bI + \bR_{0})
\end{equation*}
Using infimal convolution we can combine this convex conjugate with that of $\bA \mapsto \|\bA_{0}\|_{2, \infty}$ and rewrite the convex conjugate of $\cG$ as the solution of a constrained optimization problem.
\begin{theorem}\label{thm:dual-problem}
The optimization problem (\ref{eqn:optimisation-problem}) satisfies 
\begin{alignat}{1}
    \min_{\bA} \{\cF[\bA]: \bI + \bA > \bzero\} &= - 2 \min_{\bB} \{\cD[\bB_{0}]: \dg \bB = \bzero \mbox{ and }\|\bB_{0} - \hat{\bR}_{0}\|_{2, \infty} \leq \lambda_{n}\} 
\end{alignat}
Furthermore, if $\hat{\bH} = \argmin_{\bA} \cF[\bA]$ and $\tilde{\bR} = \argmin_{\bB} \cD[\bB_{0}]$, then $\tilde{\bR}^{-1} = \bI + \hat{\bH}$.
\end{theorem}
In other words, the optimization problem (\ref{eqn:optimisation-problem}) is equivalent to minimizing the Kullback-Leibler divergence $\cD[\bB_{0}]$ with respect to $\bbQ$ or maximizing the Carleman-Fredholm determinant $\det\nolimits_{2}(\bI + \bB_{0})$ under the constraint that every off-diagonal entry $\bB_{ij}$ stays within $\lambda_{n}$ in Hilbert-Schmidt distance from $\hat{\bR}_{ij}$. Our method can thus be seen as constrained minimization of the Kullback-Leibler divergence.

\subsection{Limit of Multivariate Graphical lasso}\label{sec:limit-mvlasso}

The functional graphical lasso can also be motivated from its multivariate counterpart by gridding, and correcting the multivariate objective function so as to obtain meaningful limiting behaviour as the grid resolution increases and finding an appropriate replacement for the penalty term.

Suppose that the spaces $\cH_{i}$ are composed of continuous functions on the sets $U_{i}$ and $X= (X_{1}, \dots, X_{p})$ is a Gaussian random element on $\cH$. We construct a grid $\{u_{ij}\}_{j=1}^{m}$ of $m$ points on each set $U_{i}$. A natural way to study the graph of $X = (X_{1}, \dots, X_{p})$ is to study the graph of $$\mathsf{X} = (X_{1}(u_{11}), \dots, X_{1}(u_{1m}), X_{2}(u_{21}), \dots, X_{2}(u_{2m}), \dots, X_{p}(u_{p1}), \dots, X_{p}(u_{pm})).$$
Thus we can apply the multivariate graphical lasso to the random vectors $\mathsf{X^{k}}$ corresponding to the independent realizations $X^{k}$ of $X$ and see what happens as $m \to \infty$.

To this end, we evaluate the empirical covariance estimator $\mathsf{\hat{C}}$ of the covariance $\mathsf{C}$ of $\mathsf{X}$. Notice that $\mathsf{C}$ and $\mathsf{\hat{C}}$ are simply restrictions to the grid $\{u_{ij}\}_{i=1, j=1}^{p,m}$ of the continuous integral kernels of the operators $\bC$ and $\smash{\hat{\bC}}$. Now, consider the objective function given by 
\begin{equation*}
    \mathsf{F(\Theta)} = \mathsf{tr(\Theta \hat{C})} - \mathsf{log~det(\Theta)} + \lambda \mathsf{\sum\nolimits_{i \neq j} |\Theta_{ij}|},
\end{equation*}
where $\mathsf{\Theta}$ is a possible candidate for the precision matrix $\mathsf{C^{-1}}$. Observe that we can write 
$$\mathsf{\hat{C}} = \mathsf{\tfrac{1}{m}D_{c}^{1/2}(I + \tfrac{1}{m}\hat{R}_{0})D_{c}^{1/2}}$$ 
where $\mathsf{D_{c}}$ and $\mathsf{\hat{R}_{0}}$ are approximately the restrictions to the grid of the integral kernels of the Hilbert-Schmidt operators $\dg \hat{\bC}$ and $\hat{\bR}_{0}$. The factors $\mathsf{\tfrac{1}{m}}$ are a result of having to replicate the operations $\bA_{ij}f_{j}(x) = \int_{U_{j}}A(x, y)f(y) dy$ whose discrete approximation is given by $l \mapsto \smash{\tfrac{1}{m}\sum_{k=1}^{m}A(u_{il}, u_{jk})f(u_{jk})}$. This suggests that we should parametrize $\mathsf{\Theta}$ as in 
$$\mathsf{\Theta}^{-1} = \mathsf{\tfrac{1}{m}D_{c}^{1/2}(I + \tfrac{1}{m}H)^{-1}D_{c}^{1/2}}$$ 
in terms of a rough approximation $\mathsf{H}$ to the grid $\{u_{ij}\}_{i=1, j=1}^{p,m}$ of $\bH$. We can now write 
\begin{eqnarray*}
    \mathsf{tr(\Theta \hat{C})} - \mathsf{log~det(\Theta)}
    &=& \mathsf{tr([\tfrac{1}{m}D_{c}^{1/2}(I + \tfrac{1}{m}H)^{-1}D_{c}^{1/2}]^{-1} \hat{C})} - \mathsf{log~det([\tfrac{1}{m}D_{c}^{1/2}(I + \tfrac{1}{m}H)^{-1}D_{c}^{1/2}]^{-1})} \\
    &=& \mathsf{tr((I + \tfrac{1}{m}H) [mD_{c}^{-1/2} \hat{C} D_{c}^{-1/2}])} - \mathsf{log~det(I + \tfrac{1}{m}H)} + \mathsf{log~det(\tfrac{1}{m}D_{c})}\\
    &=& \mathsf{tr((I + \tfrac{1}{m}H) (I + \tfrac{1}{m}\hat{R}_{0}))} - \mathsf{log~det(I + \tfrac{1}{m}H)} + \mathsf{log~det(\tfrac{1}{m}D_{c})}\\  
    &=& \mathsf{tr(\tfrac{1}{m^{2}}H\hat{R}_{0})} + \mathsf{tr(\tfrac{1}{m}H)} - \mathsf{log~det(I + \tfrac{1}{m}H)} + \mathsf{pm} + \mathsf{log~det(\tfrac{1}{m}D_{c})}   
\end{eqnarray*}
using the fact that $\mathsf{\tr \hat{R}_{0} \approx 0}$ and $\mathsf{\tr(I) = pm}$. Assuming polynomial decay of eigenvalues, $\mathsf{log~det(\tfrac{1}{m}D_{c})} \approx \textstyle  \mathsf{Cp\sum_{j=1}^{m}log(1/j^{\alpha})} \approx \mathsf{-C\alpha pm~log(m)}$ which diverges to $-\infty$ faster than the penultimate term $\mathsf{pm}$ diverges to $\infty$. This suggests that the above expression diverges to $-\infty$ and explains why the multivariate graphical lasso is not stable with respect to grid resolution. If we ignore the terms $\mathsf{pm}$ and $\mathsf{log~det(\tfrac{1}{m}D_{c})}$ with an ill-defined limits or alternatively, introduce the correction term $-\mathsf{pm} - \mathsf{log~det(\tfrac{1}{m}D_{c})}$, we can obtain a nontrivial limiting behaviour from the above expression, which gives 
\begin{eqnarray*}
    \mathsf{tr(\tfrac{1}{m^{2}}H\hat{R}_{0})} + \mathsf{tr(\tfrac{1}{m}H)} - \mathsf{log~det(I + \tfrac{1}{m}H)} 
    &=& \mathsf{tr(\tfrac{1}{m^{2}}H\hat{R}_{0})} - \left[ \mathsf{log~det(I + \tfrac{1}{m}H)} - \mathsf{tr(\tfrac{1}{m}H)} \right]\\
    &\to& \tr(\bH\hat{\bR}_{0}) - \log \det\nolimits_{2}(\bI + \bH)
\end{eqnarray*}
as $\mathsf{m} \to \infty$. Because we know from Theorem \ref{thm:CI-elements} how the information about the graph is in the off-diagonal entries of $\bH$, it now makes sense to penalize the above expression accordingly, thus recovering our objective functional $\cF[\bH]$. Section \ref{sec:discretization} contains a longer discussion on how operators are discretized and on how the above formulas may be obtained. 

\section{Theoretical Guarantees}
\label{sec:theory}
\subsection{Identifiability and Well-posedness}
If $X$ is a Gaussian random element and the correlation operator matrix $\bR$ is invertible, the pairwise Markov property can be expressed succinctly in terms of the precision operator matrix $\bH$. Even if $X$ is not Gaussian, the same applies for the conditional uncorrelatedness version of the pairwise Markov property.
\begin{theorem}[Precision Operator and Conditional Independence]\label{thm:CI-elements}
Let $X = (X_{1}, \dots, X_{p})$ be a second-order random element in the product Hilbert space $\cH = \cH_{1} \times \cdots \times \cH_{p}$.
\begin{enumerate}
    \item Under Assumption \ref{asm:equiv}, if $X$ is Gaussian then the correlation operator matrix $\bR$ is invertible and for $1 \leq i, j \leq p$ with $i \neq j$, we have the correspondence
    \begin{equation*}
        X_{i} \CI X_{j} ~|~ \{X_{k} : k \neq i, j\} \quad\mbox{ if and only if }\quad \bH_{ij}^{\ast} = \bzero,
    \end{equation*}
    \item Under Assumption 1$^{\ast}$, the correlation operator matrix $\bR$ is invertible and for $1 \leq i, j \leq p$ with $i \neq j$, we have the correspondence
    \begin{equation*}
        X_{i} \CI_{2} X_{j} ~|~ \{X_{k} : k \neq i, j\} \quad\mbox{ if and only if }\quad \bH_{ij}^{\ast} = \bzero,
    \end{equation*}
\end{enumerate}
where $\bH^{\ast} = [\bH_{ij}^{\ast}]_{i,j=1}^{p}$ is the precision operator matrix of $X$.
\end{theorem}
In other words, the off-diagonal zero entries of the adjacency matrix of $G$ (which represent edges) correspond precisely to the off-diagonal zero entries of the precision operator matrix $\bH^{\ast}$. This is the Hilbert space generalization of the familiar result for Gaussian graphical models where the role of the precision matrix is served by the inverse of the covariance instead.

As mentioned before, $\cF$ is coercive, strictly convex and continuous in the extended sense and these are sufficient conditions for a functional to admit a unique minimum and minimizer in Hilbert space. 
Our theoretical analysis depends critically on exploiting the stationary condition (\ref{eqn:stationary_condn}) below:
\begin{theorem}\label{thm:existence}
    For any $\lambda_{n} > 0$ and estimated correlation operator $\hat{\bR}$, the optimization problem (\ref{eqn:optimisation-problem}) admits a unique solution $\hat{\bH}$  which furthermore satisfies
    \begin{equation}\label{eqn:stationary_condn}
        \hat{\bR} - (\bI + \hat{\bH})^{-1} + \lambda_{n} \hat{\bZ} = \bzero
    \end{equation}
    for some $\hat{\bZ} \in \partial\|\hat{\bH}_{0} \|_{2, 1}$, where $\partial\|\hat{\bH}_{0} \|_{2, 1}$ denotes the subdifferential of $\bH \mapsto \|\bH_{0} \|_{2, 1}$ at $\bH = \hat{\bH}$.
\end{theorem}

\subsection{Finite Sample Theory}\label{sec:finite-sample-theory}

We begin by introducing some terminology from \cite{ravikumar2011} which will be useful for describing the tail behaviour of our estimators. For $\delta_{\ast} > 0$ and a function $f: \bbN \times \bbR_{+} \to \bbR_{+}$, which is monotonically increasing in both arguments, we say that an estimator $\hat{\bA} = [\hat{\bA}_{ij}]_{i,j=1}^{p}$ of an operator matrix $\bA = [\bA_{ij}]_{i,j=1}^{p}$ satisfies a \emph{tail condition with the parameters} $f$ \emph{and} $\delta_{\ast}$ if for every $n \geq 1$ and $0 < \delta < \delta_{\ast}$ we have
\begin{equation*}
        \bbP [\|\hat{\bA}_{ij} - \bA_{ij}\|_{2} \geq \delta] \leq 1/f(n, \delta).
\end{equation*}
To handle the behaviour of $\hat{\bA}$ under such tail conditions, we define
\begin{equation*}
    \bar{n}_{f}(\delta, r) = \max \{n: f(n, \delta) \leq r\} \quad \mbox{ and } \quad \bar{\delta}_{f}(r, n) = \max \{\delta: f(n, \delta) \leq r \}.
\end{equation*}
Essentially, $\bar{n}_{f}(\delta, r)$ is the smallest $n$ and $\bar{\delta}_{f}(r, n)$ is the smallest $\delta$ for which $\|\hat{\bA} - \bA\|_{2} < \delta$ with probability at least $1 - 1/r$.

Let $X = (X_{1}, \dots, X_{p})$ be a second-order random element in $\cH$ with the covariance $\bC$. Let $\{X^{k}\}_{k=1}^{n}$ be $n$ independent realizations of $X$ and $\hat{\bC} = \hat{\bC}(X^{1}, \dots, X^{n})$ be an estimator of $\bC$ which satisfies the tail condition with the parameters $f$ and $\delta_{\ast}$. Then these tail conditions together with the regularity conditions of Assumption \ref{asm:regularity}, naturally lead to similar conditions on the corresponding estimator $\hat{\bR}$ of the correlation operator $\bR$. Recall that $\rho = 1 + \tnorm{\bR_{0}}_{2, \infty}$ and $\gamma = 1 + \tnorm{\Lambda_{SS}}_{2, \infty}$.
\begin{theorem}\label{thm:tail-correlation}
    Let $\hat{\bC}$ be an estimator of $\bC$ satisfying the tail condition with the parameters $f$ and $\delta_{\ast}$ such that $\dg(\hat{\bC})$ is non-negative. 
    Under Assumption \ref{asm:regularity}, for $\epsilon_{n} = \delta^{\frac{1}{1 + \beta}}$, the corresponding estimator $\hat{\bR}$ for the correlation $\bR$ satisfies for $i \neq j$
    \begin{equation*}
        \bbP\left[\|\hat{\bR}_{ij} - \bR_{ij}\|_{2} \geq \kappa \delta^{\frac{\beta}{1 + \beta}}\right] \leq \frac{1}{f(n, \delta)}
    \end{equation*}
    for $0 < \delta \leq \delta_{\ast}$, where $\kappa = 16\sqrt{2} \left(\left[1 \vee \max_{i \neq j}\| \bR_{ij} \|_{2}] \vee [\max_{i \neq j} \{\|\Phi_{ij}\|_{2}\}[2 \vee 2\max_{j}\{ \|\bC_{jj}\|\}]\right]\right)$. Consequently, $\|\hat{\bR}_{ij} - \bR_{ij}\|_{2} \leq \kappa \bar{\delta}_{f}(n, r)^{\frac{\beta}{1 + \beta}}$ with probability at least $1 - 1/r$ when $\epsilon_{n} = \bar{\delta}_{f}(n, r)^{\frac{1}{1 + \beta}}$.
\end{theorem}
In other words, if $\hat{\bC}$ satisfies tail conditions with the parameters $f$ and $\delta_{\ast}$, then $\hat{\bR}$ satisfies a tail condition with the parameters $g(n, \delta) = f(n, [\delta/\kappa]^{1 + 1/\beta})$ and $[\delta_{\ast}/\kappa]^{1 + 1/\beta}$. Note that 
\begin{equation*}
\bar{\delta}_{g}(n, r) = \kappa \bar{\delta}_{f}(n, r)^{\frac{\beta}{1 + \beta}} \quad\mbox{ and }\quad \bar{n}_{g}(\delta, r) = \bar{n}_{f}([\delta/\kappa]^{1 + 1/\beta}, r).
\end{equation*}
Theorem \ref{thm:tail-correlation} tells us how well we can estimate the correlation operator $\bR$, given an estimator $\hat{\bC}$ of the covariance operator $\bC$. The performance depends crucially on the regularity $\beta$, with smaller values of $\beta$ requiring higher sample sizes $n = \bar{n}_{g}(\delta, r)$ to estimate $\bR_{ij}$ up to the same error with high probability. 

\subsubsection{Estimation of the Precision Operator}

We begin by describing the entry-wise Hilbert-Schmidt deviation of the estimator $\hat{\bH}$ from its estimand. In the sequel, the parameter $\tau$ is user-defined and can be increased to get better concentration of $\hat{\bH}$ near $\bH^{\ast}$ in exchange for more demanding requirements on the sample size $n$.
\begin{theorem}\label{thm:general}
    Let $X = (X_{1}, \dots, X_{p})$ be a second-order random element in the Hilbert space $\cH$ with the covariance $\bC$ and let $\hat{\bC}$ be an estimator of $\bC$ satisfying the tail condition with parameters $f$ and $\delta_{\ast} > 0$, and let $\tau > 2$. Under Assumptions \ref{asm:equiv}/1*, \ref{asm:incoherence} and \ref{asm:regularity}, and the conditions for Theorem \ref{thm:tail-correlation}, if $\epsilon_{n} = \bar{\delta}_{f}(n, p^{\tau})^{\frac{1}{1 + \beta}}$, $\lambda_{n}= \frac{8}{\alpha}\kappa \bar{\delta}_{f}(n, p^{\tau})^{\frac{\beta}{1+\beta}}$ and the sample size $n$ satisfies
    \begin{equation}\label{eqn:sample-size-requirement}
        n \geq \bar{n}_{f} \left(1/\max\left\lbrace\frac{1}{\delta_{\ast}}, \left[12d \kappa \left(1 + \frac{8}{\alpha}\right)^{2} \left[\rho \gamma \vee \rho^{3} \gamma^{2}\right]\right]^{1+\frac{1}{\beta}}\right\rbrace, p^{\tau}\right),
    \end{equation}
    then with probability at least $1-1/p^{\tau-2}$, we have:
    \begin{enumerate}
        \item The estimator $\hat{\bH}$ satisfies
        \begin{equation}\label{eqn:entry-wise-bound}
            \|\hat{\bH} - \bH^{\ast}\|_{2, \infty} \leq 2 \gamma \left(1 + \frac{8}{\alpha}\right) \kappa \bar{\delta}_{f}(n, p^{\tau})^{\frac{\beta}{1+\beta}}.
        \end{equation}
        \item If for some $(i,j)$, $\|\bH_{ij}^{\ast}\|_{2} > 2 \gamma (1 + 8/\alpha)\kappa \bar{\delta}_{f}(n, p^{\tau})^{\frac{\beta}{1+\beta}}$, then $\hat{\bH}_{ij}$ is nonzero.
    \end{enumerate}
\end{theorem}
Notice how the sample size requirement (\ref{eqn:sample-size-requirement}) depends on the incoherence parameter $\alpha$, the degree $d$ of the graph, and the parameters $\rho$ and $\gamma$ in essentially the same way as Theorem 1 of \cite{ravikumar2011}, except for the coefficient $\kappa$ and the power $1 + 1/\beta$. As the regularity $\beta \to 0$, the sample size requirement increases while the bound (\ref{eqn:entry-wise-bound}) on the entry-wise deviation $\|\hat{\bH} - \bH^{\ast}\|_{2, \infty}$ weakens. The parameters $\kappa$, $\rho$ and $\gamma$ can be said to capture the sizes of quantities involved while the degree $d$ describes the sparsity of the graph. The factor $(1 + 8/\alpha)$ quantifies the dependence of our sample size requirement and bound on entry-wise deviation on Assumption \ref{asm:incoherence}. Decreasing the incoherence $\alpha$ unsurprisingly weakens the bound. Interestingly, the dependence on the degree $d$ of the graph is through the regularity $\beta$.

The bound \eqref{eqn:entry-wise-bound} naturally yields a bound on the Hilbert-Schmidt distance of $\hat{\bH}$  from $\bH^{\ast}$.
\begin{corollary}
Let $s$ denote the total number of nonzero off-diagonal entries in $\bH^{\ast}$. Under the same conditions and choices of $\epsilon_{n}$, $\lambda_{n}$ and $n$ as Theorem \ref{thm:general}, we have
\begin{equation*}
    \|\hat{\bH} - \bH^{\ast}\|_{2} \leq 2 \gamma \left(1 + \frac{8}{\alpha}\right) \kappa \sqrt{p + s} \bar{\delta}_{f}(n, p^{\tau})^{\frac{\beta}{1+\beta}}.
\end{equation*}
 with probability at least $1-1/p^{\tau-2}$.
\end{corollary}
\begin{proof}
Clearly, $\|\hat{\bH} - \bH^{\ast}\|_{2}^{2} 
    = \sum_{i=1}^{p} \|\hat{\bH}_{ii} - \bH^{\ast}_{ii}\|_{2}^{2} + \sum_{i \neq j} \|\hat{\bH}_{ij} - \bH^{\ast}_{ij}\|_{2}^{2}
    \leq (p + s) \|\hat{\bH} - \bH^{\ast}\|_{2,\infty}^{2}$.
\end{proof}

\subsubsection{Model Selection Consistency}
For sufficiently large sample size, we can ensure that we recover the whole graph exactly with high probability. Specifically, we need $n$ large enough to ensure that the $\hat{\bH}_{ij}$ is nonzero for the smallest $\|\bH_{ij}^{\ast}\|_{2}$ with high probability. Naturally, the smaller $\theta$ is, the larger $n$ will have to be.  
\begin{corollary}[Model Selection Consistency]
    Let $\theta = \min \{\|\bH_{ij}^{\ast}\|_{2}: \bH_{ij}^{\ast} \neq \bzero \}$ and $\tau > 2$. Under the same conditions as Theorem \ref{thm:general}, if $\epsilon_{n} = \bar{\delta}_{f}(n, p^{\tau})^{\frac{1}{1 + \beta}}$, $\lambda_{n}= \frac{8}{\alpha}\kappa \bar{\delta}_{f}(n, p^{\tau})^{\frac{\beta}{1+\beta}}$ and the sample size $n$ satisfies 
    \begin{equation*}
        n \geq \bar{n}_{f} \left(1/\max\left\lbrace \frac{1}{\delta_{\ast}}, \left[\frac{2\gamma \kappa}{\theta} \left(1 + \frac{8}{\alpha}\right) \right]^{1+\frac{1}{\beta}}, \left[12d\kappa \left(1 + \frac{8}{\alpha}\right)^{2} \left[\rho \gamma \vee \rho^{3} \gamma^{2}\right]\right]^{1+\frac{1}{\beta}}\right\rbrace, p^{\tau}\right)
    \end{equation*}
    then $\bbP[\hat{G} = G] \geq 1 - 1/p^{\tau-2}.$
\end{corollary}
\begin{proof}
In addition to the sample size requirement of Theorem \ref{thm:general}, we require that $\theta > 2 \gamma (1 + 8/\alpha)\kappa \bar{\delta}_{f}(n, p^{\tau})^{\frac{\beta}{1+\beta}}$. By Theorem \ref{thm:general} $(2)$, this ensures that the entry $\hat{\bH}_{ij}$ is nonzero for every nonzero entry $\bH_{ij}^{\ast}$ thus implying exact recovery of the graph $G$. 
\end{proof}

\subsection{Sub-Gaussian Random Elements}
\label{sec:finite-sample-theory-subgaussian}
In this section, we will work out the finite sample theory for sub-Gaussian random elements $X = (X_{1}, \dots, X_{p})$, if the covariance $\bC$ is specifically estimated using the empirical covariance operator $\hat{\bC} = \frac{1}{n}\sum_{k=1}^{n}X^{k} \otimes X^{k} - \bar{X} \otimes \bar{X}$ where $\bar{X} = \frac{1}{n}\sum_{k=1}^{n} X^{k}$. 

There are different definitions of sub-Gaussianity for random elements in Hilbert spaces (cf. \cite{chen2021}, \cite{antonini1997}). For our purpose, the sub-Gaussianity of the norms of the constituent random elements provides a natural generalization of the definition used in \cite{ravikumar2011} which required the coordinates to be sub-Gaussian random variables.
\begin{definition}
We will say that $X = (X_{1}, \dots, X_{p})$ is a sub-Gaussian random element in the product space $\cH$ if the norms $\|X_{j}\|$ are sub-Gaussian random variables for $1 \leq j \leq p$. Equivalently, $X$ is a sub-Gaussian random element if 
\begin{equation*}
    \|X\|_{\infty} = \max\nolimits_{j} \| \|X_{j}\|\|_{\psi_{2}} < \infty.
\end{equation*}
\end{definition}
The above definition is weaker than an alternative definition proposed by \cite{chen2021} but stronger than the one suggested by \cite{vershynin2018}.
Using the Karhunen-Lo\`{e}ve expansion, it can be shown that for a Gaussian $X = (X_{1}, \dots, X_{p})$, the sub-Gaussian norms of $X_{j}$ satisfy
\begin{equation*}
    \|\|X_{j}\|\|_{\psi_{2}}^{2} = \|\|X_{j}\|^{2}\|_{\psi_{1}} \leq \tfrac{8}{3}\tr(\bC_{jj})
\end{equation*}
and therefore, $\|X\|_{\infty} \leq \sqrt{8/3} \max\nolimits_{j} \left[ \tr(\bC_{jj}) \right]^{1/2} \leq \sqrt{8/3} \left[ \tr(\bC)\right]^{1/2}$. Thus our definition includes all Gaussian random elements $X$ as sub-Gaussian.  

Using Bernstein's inequality (cf. Theorem 2.8.1 of \cite{vershynin2018}), we can show that
\begin{lemma}\label{thm:conc-subgauss}
If $\hat{\bC} = \frac{1}{n}\sum_{k=1}^{n} X^{k} \otimes X^{k} - \bar{X} \otimes \bar{X}$ is the empirical covariance estimator, then for $0 < \delta \leq \delta_{\ast}$, we have
\begin{equation*}
    \bbP\{ \| \hat{\bC}_{ij} - \bC_{ij} \|_{2} \geq \delta \} \leq 2 \exp \left[ -\frac{cn\delta^{2}}{\|X\|_{\infty}^{4}} \right]
\end{equation*}
where $$\delta_{\ast} = \min_{ij} \|\| X_{i} \otimes X_{j} - \bbE\left[ X_{i} \otimes X_{j} \right] - \bar{X}_{i} \otimes \bar{X}_{j} + \bbE[X_{i}] \otimes \bbE[X_{j}] \|\|_{\psi_{1}}.$$ If $\bbE[X] = \bzero$, the statement continues to hold even for $\hat{\bC} = \frac{1}{n}\sum_{k=1}^{n} X^{k} \otimes X^{k}$ and $\delta_{\ast} = \min_{ij} \|\| X_{i} \otimes X_{j} - \bbE\left[ X_{i} \otimes X_{j} \right] \|\|_{\psi_{1}}$.
\end{lemma}
If the mean is zero, which is the case addressed in \cite{ravikumar2011}, then $\delta_{\ast}$ simply does not depend on $n$. If the mean $\bbE[X]$ is not zero, the dependence of the quantity $\delta_{\ast}$ on $n$ is not as pronounced as it may seem because for large $n$, $\delta_{\ast} \approx \min_{ij} \|\| X_{i} \otimes X_{j} - \bbE\left[ X_{i} \otimes X_{j} \right] \|\|_{\psi_{1}}$ which can be shown to be always majorize $\min_{ij} \|\|X_{i}\|\|X_{j}\| - \bbE\left[ \|X_{i}\|\|X_{j}\| \right] \|_{\psi_{1}}$ by Jensen's inequality. 

Now, Lemma \ref{thm:conc-subgauss} essentially says that $\hat{\bC}$ satisfies the tail condition for $f(n, \delta) = \frac{1}{2}\exp\left[\frac{cn\delta^{2}}{\|X\|_{\infty}^{4}} \right]$ and $\delta_{\ast} > 0$, which implies that
\begin{equation*}
    \bar{n}_{f}(\delta, r) = \left\lfloor \frac{\|X\|_{\infty}^{4}\log(2r)}{c \delta^{2}} \right\rfloor 
    \quad \mbox{ and } \quad 
    \bar{\delta}_{f}(n, r) = \sqrt{\frac{\|X\|_{\infty}^{4}\log(2r)}{cn}}.
\end{equation*}
Applying Theorem \ref{thm:general} to our special case now yields explicit parameter choices, sample size requirements and upper bounds on the entry-wise deviations.
\begin{theorem}[Sub-Gaussian Random Elements] \label{thm:sub-gaussian}
Assume that $X = (X_{1}, \dots, X_{p})$ is such that the norms $\|X_{j}\|$ are sub-Gaussian and let $\hat{\bC} = \frac{1}{n}\sum_{k=1}^{n}X^{k} \otimes X^{k} - \bar{X} \otimes \bar{X}$. Under the same conditions as Theorem \ref{thm:general}, if the parameters $\epsilon_{n}$ and $\lambda_{n}$ are chosen as
\begin{equation*}
    \epsilon_{n} = \left[\frac{\|X\|_{\infty}^{4}(\log 2 + \tau \log p)}{cn}\right]^{\frac{1}{2(1 + \beta)}}, \quad 
    \lambda_{n} = \frac{8}{\alpha}\kappa \left[\frac{\|X\|_{\infty}^{4}(\log 2 + \tau \log p)}{cn}\right]^{\frac{\beta}{2(1+\beta)}}
\end{equation*}
and the sample size $n$ satisfies
\begin{equation*}
    n >  [\log 2 + \tau \log p]\max\left\lbrace\frac{1}{\delta_{\ast}^{2}}, \left[12d \kappa \left(1 + \frac{8}{\alpha}\right)^{2} \left[\rho \gamma \vee \rho^{3} \gamma^{2}\right]\right]^{2+\frac{2}{\beta}}\right\rbrace \frac{\|X\|_{\infty}^{4}}{c}
\end{equation*}
then with probability at least $1- 1/p^{\tau-2}$ we have that 
\begin{eqnarray*}
    \|\hat{\bH} - \bH^{\ast}\|_{2, \infty} 
    &\leq&
    2 \gamma \left(1 + \frac{8}{\alpha}\right) \kappa \|X\|_{\infty}^{2} \left[ \frac{\log 2 + \tau \log p}{cn} \right]^{\frac{\beta}{2(1+\beta)}}       
\end{eqnarray*}
and $\|\hat{\bH} - \bH^{\ast}\|_{2} \leq \sqrt{s + p} \|\hat{\bH} - \bH^{\ast}\|_{2, \infty}$ where $s$ denotes the number of nonzero off-diagonal entries of $\bH$. Here $\delta_{\ast}$ is as in Lemma \ref{thm:conc-subgauss}.
\end{theorem}

Assuming the parameters $\kappa$, $\rho$, $\gamma$ and $\alpha$ do not change much with respect to $p$, this suggests that
\begin{equation*}
    \|\hat{\bH} - \bH^{\ast}\|_{2} = \mathcal{O}_{\mathbb{P}} \left( \sqrt{s + p} \left[ \frac{\log p}{n} \right]^{\frac{\beta}{2(1+\beta)}} \right).
\end{equation*}
It must be noted that even under the most favourable regularity, when $\beta = 1$, we cannot recover the bound for the multivariate case, which is $\mathcal{O}_{\mathbb{P}}(\sqrt{(s+p)(\log p)/ n})$. This an intrinsic feature of the functional case. It stems from having to estimate correlation operator matrices, which is necessary in the functional setting because covariance operators never admit bounded inverses, but optional in the multivariate setting where the inverse of a full-rank covariance matrix is always bounded. Fortunately, the sample size requirement is still reasonable in that it only requires
\begin{equation*}
    n = \Omega((\delta_{\ast}^{-2} + d^{2 + 2/\beta})\tau\log p),
\end{equation*}
implying that estimation with a modest sample size is still feasible so long as $d \ll p$ and $s \ll p^{2}$.

\begin{corollary}[Model Selection Consistency for Sub-Gaussians]
Let $\theta = \min \{\|\bH_{ij}^{\ast}\|_{2}: \bH_{ij}^{\ast} \neq \bzero \}$. Under the same conditions and parameter choices of $\epsilon_{n}$ and $\lambda_{n}$ as in Theorem \ref{thm:sub-gaussian}, 
if the sample size $n$ satisfies,
\begin{equation*}
    n >  [\log 2 + \tau \log p]\max\left\lbrace \frac{1}{\delta_{\ast}^{2}}, \left[\frac{2\gamma \kappa}{\theta} \left(1 + \frac{8}{\alpha}\right) \right]^{2+\frac{2}{\beta}}, \left[12d \kappa \left(1 + \frac{8}{\alpha}\right)^{2} \left[\rho \gamma \vee \rho^{3} \gamma^{2}\right]\right]^{2+\frac{2}{\beta}}\right\rbrace \frac{\|X\|_{\infty}^{4}}{c}
\end{equation*}
we have $\bbP[\hat{G} = G] \geq 1 - 1/p^{\tau-2}.$
\end{corollary}

\noindent The sample size required for model selection consistency is thus higher. In Big-$\Omega$ notation, we need
\begin{equation*}
    n = \Omega((\delta_{\ast}^{-2} + \theta^{-2 - 2/\beta} + d^{2 + 2/\beta})\tau\log p)
\end{equation*}
to recover the graph with probability at least $1-p^{\tau-2}$. Model selection is thus feasible even when $n \ll p$ so long as we have $n \approx \log p$ samples. In contrast, estimation of the true precision operator $\bH^{\ast}$ is not even consistent in the standard high-dimensional regime where $n/p \to \alpha \in (0, 1)$, since it requires $n \approx (s + p)^{1 + 1/\beta} \log p$ samples, implying that $n > p^{2}$.

\begin{remark}
     Our theoretical analysis also provides some additional insight into the multivariate case by showing that a partial recovery might still be feasible even when the graph is not sparsely connected so long as it can be partitioned into sparsely connected subgraphs. Given a random vector $Y = \{Y_{j}\}_{j=1}^{P}$ for some $P \gg 1$, this would correspond to partitioning $Y$ into sparsely related smaller random \emph{subvectors} $X_{j} \subset Y$ while the entries of $Y$ withing an individual random vector $X_{j}$ are allowed to be densely related. This ensures that the maximum degree $d$ for the graph of $\{X_{j}\}_{j=1}^{p}$ is small even if that of $\{Y_{j}\}_{j=1}^{P}$ is not, thus creating the sufficient conditions for our finite sample arguments to work. We would expect the standard multivariate graphical lasso to fail here because it does not leverage the latent sparsity in the form of sparsely related subvectors of $Y$. 
\end{remark}

\section{Implementation}
\label{sec:implementation}

\subsection{Discretization}\label{sec:discretization}

The operator formalism used in (\ref{eqn:optimisation-problem}) has so far allowed us to postpone the delicate issue of how the involved quantities are discretized for the purpose of computation. It turns out that this is mostly a question of finding the correct discrete equivalents of the objects and operations involved. As in Section \ref{sec:limit-mvlasso}, we will denote the discrete counterparts of functional quantities (eg. $X$ or $\bC$) using the sans serif font (eg. $\mathsf{X}$ or $\mathsf{C}$).

We will discretize every element $\bbf = (f_{1}, \dots, f_{p})$ in $\cH$ as a column vector $\mathsf{f}$ of length $\mathsf{K}$ indexed by $\mathsf{J} = \mathsf{[1, \dots, K]}$. The coordinates $f_{i}$ in $\cH_{i}$ will be discretized as subvectors $\mathsf{f_{i}} = \mathsf{[f_{ij}: j \in J_{i}]}$ of length $\mathsf{K_{i}}$ indexed by $\mathsf{J_{i}} \subset \mathsf{J}$ such that $\mathsf{J}$ is the concatenation of the sets $\mathsf{J_{i}}$. There are many different ways of doing this. For example, if $\cH_{i}$ is the space $L^{2}(U_{i}, \mu_{i})$ of square-integrable functions on some space $U_{i}$ equipped with the measure $\mu_{i}$, we can generate a mesh $\{U_{ij}\}$ of $\mathsf{K_{i}}$ cells of roughly equal measure on $U_{i}$ and take $\mathsf{f_{ij}}$ to be the average value $\tfrac{1}{\mu_{i}(U_{ij})}\int_{U_{ij}} f_{i}$ of $f_{i}$ in the $\mathsf{j}$th cell. Often, $\cH_{i}$ is composed of continuous functions on $U_{i}$ and we can take $\mathsf{f_{ij}}$ to be the value $f_{i}(u_{ij})$ at some fixed point $u_{ij} \in U_{ij}$. These discretizations schemes can be described as \emph{discretization by cell averaging} and \emph{discretization by point evaluation} respectively \citep[c.f.][]{masak2019}.

Another recourse is to take $\mathsf{f_{ij}}$ to be the $\mathsf{j}$th coefficient $\langle f_{i}, e_{j} \rangle$ in the basis expansion of $f_{i}$ with respect to a fixed basis $\{e_{j}\}_{j=1}^{\infty}$ on $\cH_{i}$. This is \emph{discretization by basis representation}. An element $\bbf \in \cH$ can thus be represented in terms of the tensor product basis formed from bases on the spaces $\cH_{j}$. There are plenty of different ways of doing this. One can use pre-specified bases such as B-splines or empirical bases corresponding to Karhunen-Lo\`{e}ve type expansions, be it a one-dimensional expansion in every node like in \citet{qiao2019}, a two-dimensional expansion under additional structural assumptions like in \citet{zapata2022}, or any other version of multivariate functional PCA \citep{chiou2014}. 

Let $\bbf, \bbg, \bbh \in \cH$ where $\bbf = (f_{1}, \dots, f_{p})$ and $\bbg = (g_{1}, \dots, g_{p})$, with the discretizations $\mathsf{f, g}$ and $\mathsf{h}$. The tensor or outer product $\bbf \otimes \bbg$  of $\bbf, \bbg \in \cH$ is to be discretized simply as the matrix $\mathsf{f}\mathsf{g^{\top}}$. The inner product $\langle \bbf , \bbg \rangle$, however, is to be discretized as $\mathsf{f}^{\top}\mathsf{Mg}$, where the $\mathsf{K \times K}$ matrix $\mathsf{M}$ is the discrete equivalent of the inner product operation, which actually depends on the scheme of discretization employed. If we are averaging on cells or using point evaluations as discussed before, $\mathsf{M}$ is given by $\mathsf{M_{ij} = 1/K_{l}}$ if both $\mathsf{i = j \in J_{l}}$ and is $\mathsf{0}$ otherwise. On the other hand, if we are using a basis representation, $\mathsf{M}$ is same as the $\mathsf{K \times K}$ identity matrix $\mathsf{I_{K}}$. This difference follows from the observation that for $f_{i}, g_{i} \in \cH_{i}$ and their discretizations $\mathsf{f_{i}}, \mathsf{g_{i}}$, we have under the former schemes of observation
\begin{equation*}
    \langle f_{i}, g_{i} \rangle = \textstyle \int_{U_{i}} f_{i}(u) g_{i}(u) d\mu_{i}(u) \approx \mathsf{\frac{1}{K_{i}}\sum_{j=1}^{K_{i}} f_{ij}g_{ij}}
\end{equation*}
while under the latter scheme, we have 
\begin{equation*}
    \langle f_{i}, g_{i} \rangle = \textstyle \sum_{j=1}^{\infty} \langle f_{i}, e_{j}\rangle \langle g_{i}, e_{j}\rangle \approx \mathsf{\sum_{j=1}^{K_{i}} f_{ij}g_{ij}}
\end{equation*}
instead. The correct way to represent $(\bbf \otimes \bbg)\bbh = \langle \bbg, \bbh \rangle\bbf$ is thus $\mathsf{(g^{\top}Mh)f} = \mathsf{(fg^{\top})Mh}$.

Because compact operators are infinite sums of tensor products of elements, we can discretize them in essentially the same way as elements themselves. Thus the discretization of a Hilbert-Schmidt operator matrix $\bA = [\bA_{ij}]_{i,j=1}^{p}$ is a $\mathsf{K \times K}$ matrix $\mathsf{A}$, with the entries of $\bA$ being represented by the submatrices of $\mathsf{A}$ in the same way as with the element $\bbf$ and its coordinates. And just like the outer products, we can represent $\bA \bbh$ as $\mathsf{AMh}$. The same applies to operator-operator multiplication and the correct representation of the product $\bA\bB$ is $\mathsf{AMB}$, where $\bB$ is another Hilbert-Schmidt operator with the discretization $\mathsf{B}$.

Non-compact operators such has $\bI$ on the other hand, cannot be discretized like compact operators. For such operations, it is best to find the discrete equivalent of their action on the elements directly. For $\bI$, notice that $\bI \bbf = \bbf$. The linear operation that,  applied to any $\mathsf{f}$, returns $\mathsf{f}$ is the  $\mathsf{K \times K}$ identity matrix $\mathsf{I_{K}}$, of course. So the correct way to represent the operation $\bI \bbf$ is $\mathsf{I_{K}f}$. Trivial as it may appear, understanding this is what leads to the correct representation of the operation $(\bI + \bA)^{-1/2}\bbf$ which is $\mathsf{(I_{K} + AM)^{-1/2}f}$, as can be inferred from the binomial expansion of $(\bI + \bA)^{-1/2}$.

Using the same principle, we can work out that the trace $\tr(\bA)$ and Carleman-Fredholm determinant $\det\nolimits_{2}(\bI + \bA)$ of $\bA$ can be represented as 
\begin{equation*}
    \mathsf{tr[MA]} \qquad\mbox{and}\qquad \mathsf{det[I_K + MA]\cdot exp\left(-tr[MA]\right)},
\end{equation*}
respectively. Note also that the action of taking the diagonal part $\dg(\bA)$ of $\bA$ is equivalent to taking the Hadamard product $\mathsf{D \circ A}$ with matrix $\mathsf{D = [D_{ij}]}$ where $\mathsf{D_{ij} = 1}$ if both $\mathsf{i,j \in J_{l}}$ for some $\mathsf{1 \leq l \leq p}$ and is $\mathsf{0}$ otherwise. We are now going to describe the discretized version of our algorithm.

Given $n$ realizations of $X$ in the form of vectors $\mathsf{\{X_{k}: k = 1, \dots, n\}}$ we compute the discretized version $\mathsf{C}$ of the estimated covariance $\hat{\bC}$. For example, if $\hat{\bC}$ is the empirical covariance estimator, 
\begin{equation*}
    \mathsf{C = \textstyle \frac{1}{n} \sum_{k=1}^{n} X_{k}^{\phantom{.}} X_{k}^{\top} - \left[ \frac{1}{n} \sum_{k=1}^{n} X_{k}^{\phantom{.}} \right]\left[ \frac{1}{n} \sum_{k=1}^{n} X_{k}^{\phantom{.}} \right]^{\top}}
\end{equation*}
The off-diagonal part $\hat{\bC}_{0}= \hat{\bC} - \dg \hat{\bC}$ is discretized as $\mathsf{C - D \circ C}$. Since $\dg \hat{\bC} \; \bbf$ is given by $(\mathsf{D \circ CM})\mathsf{f}$,
the estimated cross-correlation operator matrix $\hat{\bR}_{0}$ (which is compact) thus corresponds to 
\begin{equation*}
    \mathsf{R_{0} =} \left[\epsilon_{n} \mathsf{I_K + D \circ CM}\right]^{-1/2} \cdot \mathsf{\left[  C - D \circ C \right]} \cdot \left[\epsilon_{n} \mathsf{I_K + D \circ MC}\right]^{\mathsf{-1/2}}.
\end{equation*}
The operator trace $\tr (\bH \hat{\bR}_{0})$ and the Carleman-Fredholm determinant $\det\nolimits_{2}(\bI + \bH)$ can be evaluated in terms of the matrix trace $\mathsf{tr}$ and determinant $\mathsf{det}$ as 
\begin{equation*}
    \mathsf{tr [MHMR_{0}]} \qquad\mbox{and}\qquad \mathsf{det[I_{K} + MH]\cdot exp\left(-tr[MH]\right)},
\end{equation*}
respectively. Finally, recall that the discretization of $\bA = [\bA_{ij}]_{i,j=1}^{p}$ is defined as a $\mathsf{K \times K}$ matrix $\mathsf{A = [A_{ij} ]_{i,j=1}^K}$ where each of the operators $\bA_{ij}$ is discretized as $\mathsf{A[J_{i}, J_{j}]} = \mathsf{[A_{kl}]_{k \in J_{i}, l \in J_{j}}}$. Because $\|\bA\|_{2,1} = \sum_{i,j=1}^{p}\|\bA_{ij}\|_{2} = \sum_{i,j=1}^{p} [\tr(\bA_{ij}\bA_{ij}^{\ast})]^{1/2}$, the discretized counterpart is given by 
\begin{eqnarray*}
    &&\mathsf{\sum_{i,j = 1}^{p} tr\left[M[J_{i}, J_{i}]A[J_{i}, J_{j}]M[J_{j}, J_{j}]A[J_{i}, J_{j}]^{\top}\right]}\\ &=& \mathsf{\sum_{i,j = 1}^{p} tr\left[(M[J_{i}, J_{i}]^{1/2}A[J_{i}, J_{j}]M[J_{j}, J_{j}]^{1/2})(M[J_{i}, J_{i}]^{1/2}A[J_{i}, J_{j}]M[J_{j}, J_{j}]^{1/2})^{\top}\right]} \\
    &=&\mathsf{\|M^{1/2}AM^{1/2}\|_{2,1}} 
\end{eqnarray*}
where the norm $\mathsf{\| \cdot\|_{2,1}}$ is defined as $ \mathsf{\|A\|_{2,1}} = \mathsf{\sum_{i,j = 1}^{p} \|A[J_{i}, J_{j}]\|_{F}}$. Altogether, the optimization functional $\cF$ can thus be written as
\begin{align}\label{eq:optim_discrete0}
    \textstyle \mathsf{F[H]} &= \mathsf{tr [MHMR_0] + tr[MH]} - \mathsf{log~ det[I_{K} + MH]} + \lambda_{n} \cdot \mathsf{\|M^{1/2}(H - D \circ H)M^{1/2}\|_{2,1}}.
\end{align}
Now, using the cyclic property of the trace and multiplicativity of the determinant, we can write
\begin{eqnarray*}
   \mathsf{tr [MHMR_0] + tr[MH]} 
   &=& \mathsf{tr [(M^{1/2}HM^{1/2})(M^{1/2}R_0M^{1/2})] + tr[M^{1/2}HM^{1/2}]} \\
   &=& \mathsf{tr [(I_{K} + M^{1/2}HM^{1/2})(I_{K} + M^{1/2}R_0M^{1/2})] - tr[I_{K} + M^{1/2}R_0M^{1/2}]} \\
   \mathsf{log~ det[I_{K} + MH]} &=& \mathsf{log~ det[I_{K} + M^{1/2}HM^{1/2}]} \\
   \mathsf{\|M^{1/2}(H - D \circ H)M^{1/2}\|_{2,1}} &=& \mathsf{\|(I_{K} + M^{1/2}HM^{1/2}) - D \circ (I_{K} + M^{1/2}HM^{1/2})\|_{2,1}}
\end{eqnarray*}

Ignoring the constant term $- \mathsf{tr [I_{K} + M^{1/2}R_0M^{1/2}]}$ in the second equation, the problem reduces to minimizing 
\begin{equation}\label{eq:optim_discrete1}
    \mathsf{\tilde{F}(Q)} = \mathsf{tr [QR]} - \mathsf{log~det[Q]} + \lambda_{n} \cdot \mathsf{\|Q - D \circ Q\|_{2,1}}
\end{equation}
with respect to $\mathsf{Q}$ under the constraint $\mathsf{Q} > 0$, where $\mathsf{Q = I_{K} + M^{1/2}HM^{1/2}}$ and $\mathsf{R = I_{K} + M^{1/2}R_{0}M^{1/2}}$. Note that the matrix determinant should be evaluated directly as the product of the eigenvalues of the matrix calculated using the eigendecomposition rather than using cofactor expansion. The former is vastly superior in terms of computational efficiency and numerical precision. 

The operator formalism used in (\ref{eqn:optimisation-problem}) faithfully encapsulates the three different discretization techniques (averaging on cells, evaluation at points, or an orthonormal basis representation) discussed above, in a \emph{coordinate-free} way \citep{stone1987}. As a consequence of this faithful representation, the above formulas are comparable across (high enough) resolutions. Their value does not change drastically if one increases the number of points, cells or basis functions without bound and in fact tends to the exact values of the corresponding quantities as the number of samples increases.  

\subsection{Optimization}
We use the Alternating Direction Method of Multipliers \citep[ADMM,][]{boyd2011} to solve the convex optimization problem \eqref{eq:optim_discrete1}. 
The basic idea of ADMM is to introduce an auxiliary variable $\mathsf{Z}$ to separate the loss (or likelihood) term from the penalty term,
\begin{eqnarray*}
    \argmin_{\mathsf Q,\mathsf Z} \mathsf{ tr [QR]} - \mathsf{log~ det[Q]} + \lambda_{n} \cdot \mathsf{\sum_{i \neq j} 
    \left[ \sum_{u \in I_{i}} \sum_{v \in I_{j}} Z_{uv}^{2} \right]^{1/2}} \quad \mbox{s.t.} \quad \mathsf{Q=Z}.
\end{eqnarray*}
The augmented Lagrangian can then be written as
\begin{eqnarray*}
    \argmin_{\mathsf Q,\mathsf Z} \mathsf{ tr [QR]} - \mathsf{log~ det[Q]} + \lambda_{n} \cdot \mathsf{\sum_{i \neq j} 
    \left[ \sum_{u \in I_{i}} \sum_{v \in I_{j}} 
    Z_{uv}^{2} 
    \right]^{1/2}} + \frac{\rho}{2} \mathsf{\|Q-Z\|_F^2} + \langle \mathsf Y, \mathsf Q-\mathsf Z \rangle
\end{eqnarray*}
and subsequently minimized w.r.t.~$\mathsf Q$ and $\mathsf Z$ in an alternating fashion, with the dual variable $\mathsf{Y}$ updated after every iteration. In the above, $\rho$ is a small positive constant affecting the convergence speed, not the convergence itself, which is guaranteed irrespective of the choice \citep{boyd2011}. We use the default $\rho=1$ in our applications of the algorithm. It is customary to perform another variable change: $\mathsf{U:=Y/\rho}$. The augmented Lagrangian then becomes
\begin{eqnarray*}
    \textstyle \mathsf L_\rho(\mathsf Q,\mathsf Z,\mathsf U) = \mathsf{tr [QR]} - \mathsf{log~ det~Q} + \lambda_{n} \cdot \mathsf{\sum_{i \neq j} 
    \left[ \sum_{u \in I_{i}} \sum_{v \in I_{j}} 
    Z_{uv}^{2} 
    \right]^{1/2}} + \frac{\rho}{2} \mathsf{\|Q-Z+U\|_F^2},
\end{eqnarray*}
which is equal to the original one up to a constant, and hence the optimal $\mathsf H$ can be obtained easily from the optimal $\mathsf Q$, indeed providing a solution to the original problem \eqref{eq:optim_discrete0}. Overall, the $\mathsf{l}$-th iteration of the ADMM algorithm consists of the following three steps, iterated until confergence for $\mathsf{m=1,2,\ldots}$ starting from an initial point $\mathsf{Z^{(0)}}, \mathsf{U^{(0)}}$:
\begin{equation*}
\begin{split}
     \mathsf{Q^{(m)}} &:= \argmin_{\mathsf Q} \mathsf{L_\rho(Q,Z^{(m-1)}, U^{(m-1)})}\\
     \mathsf{Z^{(m)}} &:= \argmin_{\mathsf Z} \mathsf L_\rho(\mathsf{Q^{(m)}},\mathsf Z,\mathsf{U^{(m-1)})}\\
     \mathsf{U^{(m)}} &:= \mathsf{U^{(m-1)} + (Q^{(m)} - Z^{(m)})}.
\end{split}
\end{equation*}

The first step has an analytic solution. Equating the derivative of $\mathsf{L_\rho}(\mathsf Q,\mathsf{Z^{(m-1)}},\mathsf{U^{(m-1)})}$ w.r.t.~$\mathsf Q$ to zero, one obtains the following non-linear system:
\[
\rho \mathsf{Q - Q^{-1}} = \rho(\mathsf{Z^{(m-1)}} - \mathsf{U^{(m-1)})} - \mathsf R.
\]
Denoting by $\mathsf{E^{(m-1)}} \Gamma^{(\mathsf m-1)} (\mathsf{E^{(m-1)}})^\top$ the eigendecomposition on the right-hand side and changing the variable to $\widetilde{\mathsf Q} = \mathsf{(E^{(m-1)})^\top Q E^{(m-1)}}$, the non-linear system becomes
\[
\rho \mathsf{\widetilde{Q} - \widetilde{Q}^{-1}} =\mathsf{\Gamma}^{(\mathsf m-1)}.
\]
Note that $\widetilde{\mathsf Q}$ and $\widetilde{\mathsf Q}^{-1}$ have the same eigenvectors, and $\mathsf{\Gamma^{(l-1)}}$ is diagonal, i.e.~the eigenvectors form the canonical basis of $\R^K$. Hence the solution is given by matching the eigenvalues only: for $i=1,\ldots,K$ it is sufficient to have
$\rho \widetilde{q}_{ii} - 1/\widetilde{q}_{ii} = \gamma_{ii}^{(l-1)}$. These equations are respectively solved by 
\[
\widetilde{\mathsf q}_{ii}^{(\mathsf m)} = \frac{\gamma_{ii}^{(\mathsf m-1)} + \sqrt{\left[\gamma_{ii}^{(\mathsf m-1)}\right]^2 + 4 \rho}}{(2\rho)}.
\]
With these forming the diagonal of $\widetilde{\mathsf Q}^{(\mathsf m)}$, we obtain $\mathsf Q^{(\mathsf m)} = \mathsf{E^{(m-1)} \widetilde{Q}^{(m)} (E^{(m-1)})^\top}$. Note that we chose the negative sign above to obtain a positive semi-definite solution, which is naturally the one sought even though we do not make this constraint explicit. 

In the second step, the problem separates in variables $\mathsf{Z_{i,j}} := \mathsf{Z[I_i,I_j]}$ with the group lasso penalizing only off-diagonal blocks. Hence, using the shorthand notation, the solution is given by
\[
\mathsf{Z_{i,j}^{(m)}} = \begin{cases}
    \mathsf{Q^{(m)}_{i,j} + U^{(m-1)}_{i,j} \qquad\qquad\, \mbox{for} \quad i=j}, \\
    \mathcal{S}_{\mathsf{\lambda_n/\rho}} \mathsf{(Q^{(m)}_{i,j} + U^{(m-1)}_{i,j}) \quad \mbox{for} \quad i\neq j},
\end{cases}
\]
where $\mathcal{S}_{\mathsf t}(\mathsf M) = (1-\tfrac{\mathsf t}{\|\mathsf M\|_F})_+ \mathsf M$ is the group-wise soft-thresholding operator \citep{friedman2010}.
We always use $\mathsf{Z^{(0)}=U^{(0)}=\diag(R)}$ as the starting point and iterate until the relative residual is small, namely until $\|\mathsf Q^{(m)}-\mathsf Z^{(m)}\|_F/\|\mathsf Q^{(m)}\|_F \leq 10^{-4}$.

\section{Simulation Study}
\label{sec:simulations}
We now explore finite sample performance of the proposed methodology  in a small simulation study. We devise three simulation setups underlining several claims we intend to make. Below, we describe the three setups briefly, while a full description is available in the supplementary material.

\begin{description}
    \item[Setup 1] is closely related to Model 1 of \cite{qiao2019}, which generates the functional datum in every node as a zero-mean Gaussian with the covariance being rank 5 with Fourier eigenfunctions and equal eigenvalues, with the precision matrix chosen such that a functional AR(2) process is formed between the nodes. This is done in a perfectly regular way such that all rank-one projections of the processes form AR(2) processes on their own, and the dependencies are created over the whole functional domain. We change this slightly to rank 10, eigenvalues $\lambda_j = 1/l$ for $k=l,\ldots,10$, and create the AR(2) dependencies only between the eigenfunctions corresponding to $\lambda_6,\ldots,\lambda_{10}$.
    \item[Setup 2] also utilizes Fourier eigenfunctions but differs from Setup 1 in two aspects. Firstly, the rank is not finite, with eigenvalues decaying quadratically as $\lambda_l = 1/l^2$. Secondly, the functional AR(2) dependencies are not flat, they are created only locally on one tenth of the functional domain corresponding to every node. This makes them harder to discover after a projection. We consider these local dependencies in the time domain more realistic as opposed to the perfectly global spectral dependencies in Setup 1, where eigenfunctions directly influence themselves across different nodes.
    \item[Setup 3] superposes independent Fourier rank-5 processes with fractional Brownian motions (with the parameter $H=0.2$, i.e.~a relatively slow eigendecay). But here, the dependency is only formed between the fractional Brownian motions. In other words, every functional datum has an independent smooth component and a dependent but rough component. We believe such rough short-scale dependencies could be interesting e.g.~in portfolio optimization \citep{carvalho2007} with a short time horizon \citep{lin2021}.
\end{description}

The proposed functional graphical lasso is implemented using the ADMM algorithm described in Section \ref{sec:implementation} and compared against the functional graphical lasso of \cite{qiao2019} implemented using a block coordinate gradient descent. The computer code for the latter was kindly provided to us by the authors, and we slightly modified it to allow for non-regular settings (namely Setup 3). Note that the very fact that this modification can be done and is guaranteed to work, stems from the theoretical development in this paper. We do not compare against other, possibly non-functional approaches, since these have been shown inferior by \cite{qiao2019}.

The results are averages of 24 independent simulation runs. They are reported in terms of mean ROC curves, showing the performance across all values of the the lasso penalty parameter $\lambda_n$ leading to different sparsity levels. This not only leads to fair comparisons, but also note that $\lambda_n$ is typically chosen in practice in order to obtain a desired sparsity level \citep{danaher2014}. Alternatively, the stability selection approach of \cite{meinshausen2010} can be used.

The competing projection-based approach of \citet{qiao2019} requires a user to choose the projection levels, i.e.~the rank and the number of B-splines. The authors provide a standard prediction-based cross-validation approach to choose first the number of B-splines and then the rank. While their approach works reasonably well for the former, we see no reason why it should work for the latter. In fact, it often doesn't in our experience, e.g.~in Setup 2. Hence we show two versions of the algorithm in our simulations, one with the cross-validated tuning parameters, and the other with the number of B-splines fixed at 15 and the rank fixed at 5. Those are two arbitrary and rather low choices one might think about given the other parameters of the problems, and can be interpreted easily for comparison purposes. On the other hand, the proposed methodology does not require a choice of any tuning parameters, which is a genuine practical advantage.

We fix $n=p=100$ in order to facilitate comparisons with \citet{qiao2019} or even \citet{zapata2022}. Still, we do not emulate specifically the simulation setups of \citet{zapata2022} or include their method in our comparisons for the following reasons. While the approach of \citet{qiao2019} needs a choice of two tuning parameters, they are both easily interpretable, and a relatively simple way of choosing them is provided. On the other hand, the approach of \citet{zapata2022} also requires a choice of two tuning parameters: the number of partially separable components and a tuning parameter weighing their sparsity levels together. But the first one is chosen arbitrarily (as the proportion of variance explained) while the second one is chosen in an oracle fashion, and does not have a straightforward interpretation. The point of this simulation study is not to demonstrate a general superiority of our approach, there is in fact no reason why our methodology should outperform that of \citet{qiao2019} or \citet{zapata2022} for well chosen values of their respective tuning parameters. But we rather aim to demonstrate the advantages of not being forced to choose any tuning parameters, which is an implication of the theoretical development in this paper, free of any structural assumptions \citep[a self-evident point in the case of][]{zapata2022}.

Figure \ref{fig:sims} displays the results of our simulation study. We can see that in Setup 1, the proposed methodology matches that of \citet{qiao2019}. Even though we increased the number of Fourier eigenfunctions to 10 and only created dependencies between the second group of five, the cross-validation approach of \citet{qiao2019} correctly identifies the number of components needed, and matches the performance of the proposed method. On the other hand, the poor performance of the fixed pre-chosen projection in Setup 1 shows the dangers of choosing the projection level too low. In Setup 2, on the other hand, the cross-validation approach of \citet{qiao2019} underestimates the rank, leading to a worse performance than with the pre-chosen values of the projection levels. Still, the proposed approach clearly outperforms both of its competitors. Finally, in Setup 3, the proposed approach vastly outperform its competitors, because this simulation setup generally disfavors projections. While cross-validation leads to higher projection levels than the pre-chosen ones in this case, it does not retain a sufficient number of components.

\begin{figure}[!t]
   \advance\leftskip-0.3cm
   \begin{tabular}{ccc}
   (a) Setup 1 & (b) Setup 2 & (c) Setup 3 \\
   \includegraphics[width=0.32\textwidth]{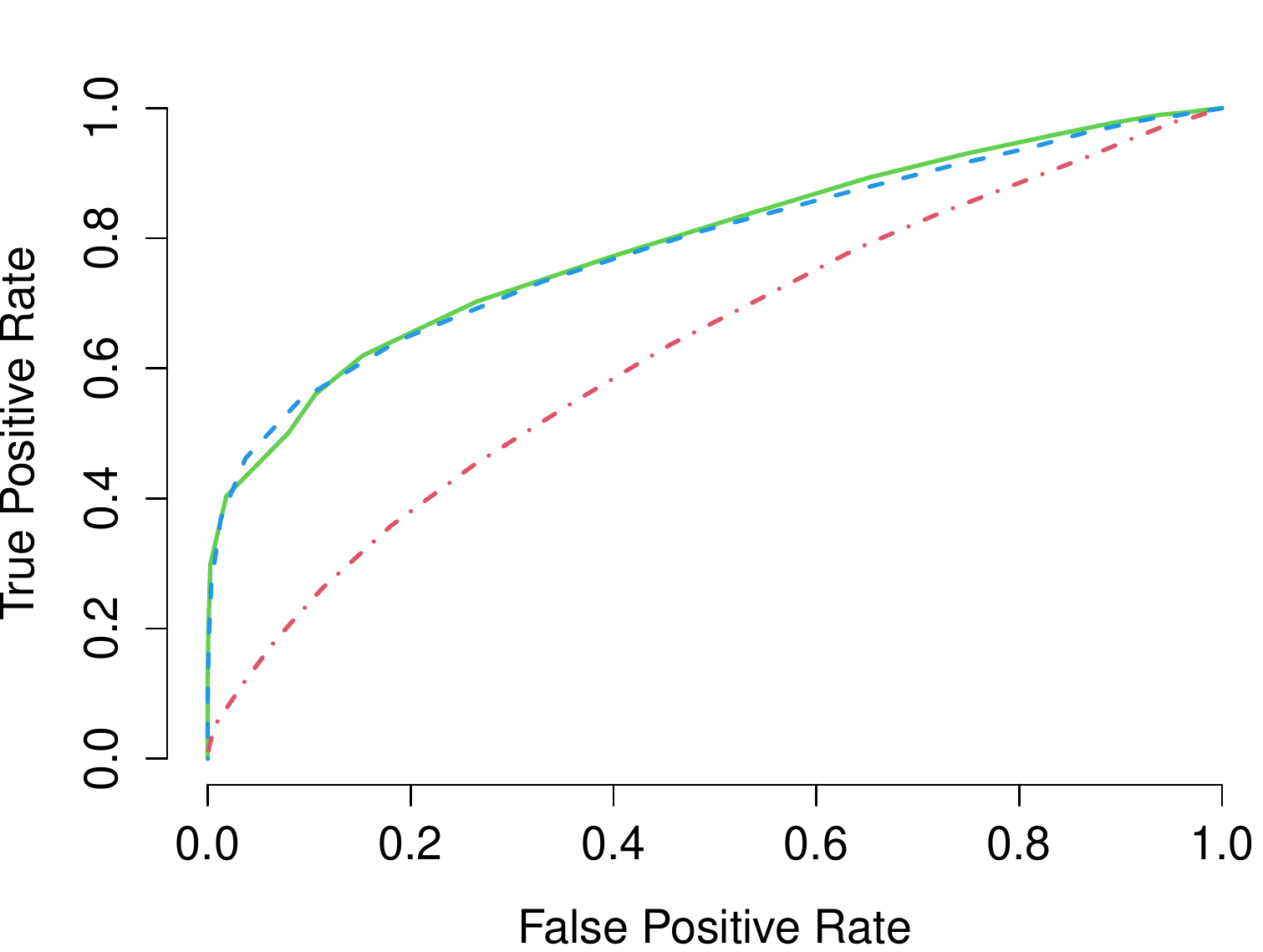} &
   \includegraphics[width=0.32\textwidth]{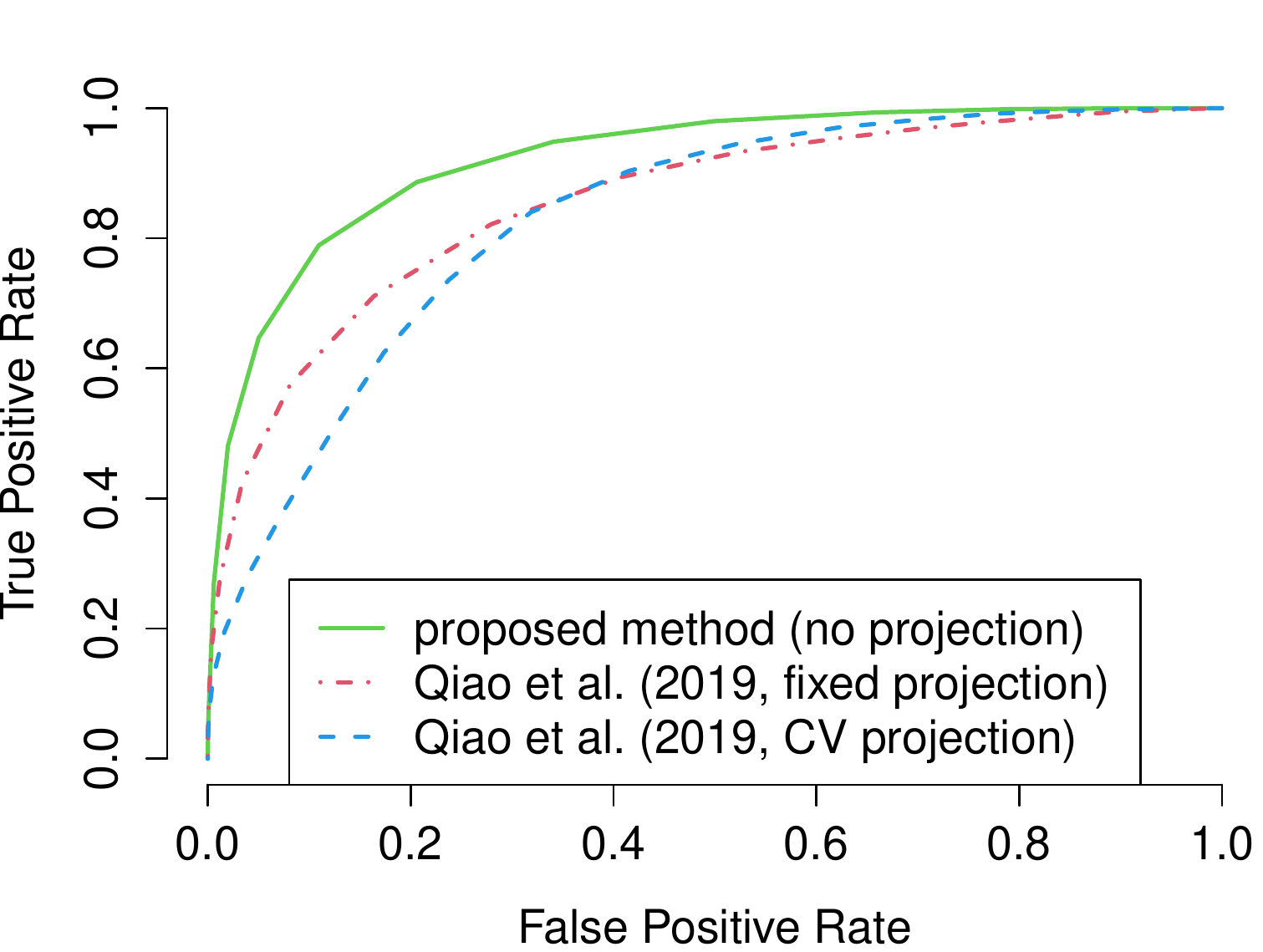} &
   \includegraphics[width=0.32\textwidth]{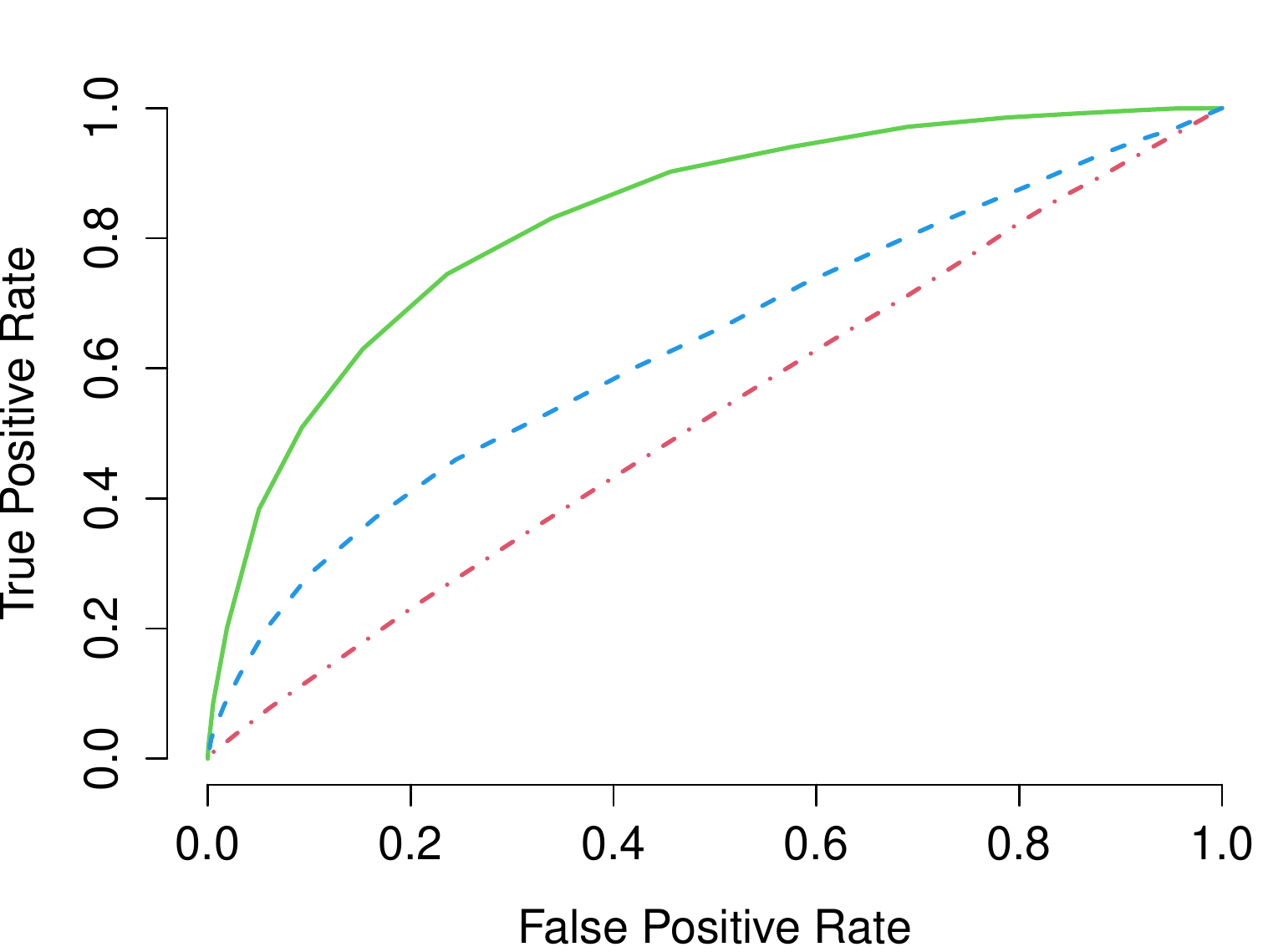} 
   \end{tabular}  
   \caption{ROC curves in the three simulation setups for the proposed method and the projection approach of \citet{qiao2019} with pre-chosen and cross-validated choice of projection dimensions.}
    \label{fig:sims} 
\end{figure}

Overall, Setup 1 constitutes an example where not performing projections even in a perfectly low-dimensional case poses no issues. Setup 2 illustrates that projecting data in a not perfectly low-dimensional case can lead to loss of information. Finally, Setup 3 constitutes a case where projections are simply not advisable.

\section{Appendix}

The appendix contains proofs of formal statements appearing in the original paper and additional details on the simulation study.

\subsection{Background and Notation}

\begin{lemma}
    The functionals $\tnorm{\cdot}_{2, \infty}$ and $\tnorm{\cdot}_{2, 1}$ are sub-multiplicative norms.
\end{lemma}
\begin{proof}
For $\bA = [\bA_{ij}]_{i,j = 1}^{p}$, $\bB = [\bB_{ij}]_{i,j = 1}^{p}$, $A = [\|\bA_{ij}\|_{2}]_{i,j=1}^{p}$ and $B = [\|\bB_{ij}\|_{2}]_{i,j=1}^{p}$, we have
\begin{equation*}
    \tnorm{\bA + \bB}_{2, \infty} 
    \leq \tnorm{A + B}_{\infty} 
    \leq \tnorm{A}_{\infty} + \tnorm{B}_{\infty} 
    = \tnorm{\bA}_{2, \infty} + \tnorm{\bB}_{2, \infty},
\end{equation*}
and similarly, 
\begin{equation*}
    \tnorm{\bA\bB}_{2, \infty} 
    \leq \tnorm{AB}_{\infty} 
    \leq \tnorm{A}_{\infty} \cdot\tnorm{B}_{\infty} 
    = \tnorm{\bA}_{2, \infty}\cdot \tnorm{\bB}_{2, \infty}.
\end{equation*}
with the same conclusion following for $\tnorm{\cdot}_{2,1}$ from $\tnorm{\bA}_{2,1} = \smash{\tnorm{\bA^{\top}}_{2,\infty}}$. The first inequalities in both the cases follow from the sub-additivity and sub-multiplicativity of the Hilbert-Schmidt norm.
\end{proof}

\subsection{Conditional Independence for Random Elements}
In the following proof we will use ``sub-setted" matrix $\bA_{ss}$ to mean the matrix $\bA$ with $ij$th entries with neither $i$ nor $j$ in $s$ being equal to $\bzero$. The symbols $\bA_{si}$ and $\bA_{js}$ are defined accordingly.
\begin{proof}[Proof of Theorem \ref{thm:CI-elements}]
    Pick $i \neq j$ such that $\bP_{ij} = \bzero$. Let $s = \{k: 1 \leq k \leq p \mbox{ for } k \neq i,j\}$. By Theorem 2.2.3 of \cite{bakonyi2011}, this is equivalent to 
    \begin{equation*}
        \bR_{ij}^{\phantom{.}} = \bR_{is}^{\phantom{.}}\bR_{ss}^{-1}\bR_{sj}^{\phantom{.}}.
    \end{equation*}
    We will show that this is in turn equivalent to saying that for every $f_{i} \in \cH_{i}$ and $f_{j} \in \cH_{j}$, we have
    \begin{equation*}
        \langle f_{i}, X_{i} \rangle \CI \langle f_{j}, X_{j} \rangle ~|~ \{ \langle f_{k}, X_{k} \rangle: f_{k} \in \cH_{k} \mbox{ where } k \neq i,j \}.
    \end{equation*} 
    Because of Gaussianity, this is equivalent to saying that 
    \begin{eqnarray*}
        &&\mathrm{Cov}[\langle f_{i}, X_{i} \rangle, \langle f_{j}, X_{j} \rangle | \langle f_{k}, X_{k} \rangle : f_{k} \in \cH_{k} \mbox{ where } k \neq i,j] \\
        &&= \bbE[\langle f_{i}, X_{i} \rangle \langle f_{j}, X_{j} \rangle | \langle f_{k}, X_{k} \rangle :  f_{k} \in \cH_{k} \mbox{ where } k \neq i,j ] \\
        &&\quad- \bbE[\langle f_{i}, X_{i} \rangle | \langle f_{k}, X_{k} \rangle : f_{k} \in \cH_{k} \mbox{ where } k \neq i,j] \cdot \bbE[\langle f_{j}, X_{j} \rangle | \langle f_{k}, X_{k} \rangle : f_{k} \in \cH_{k} \mbox{ where } k \neq i,j] \\
        &&= 0.
    \end{eqnarray*}
    Taking the expectation gives that $\bbE[\langle f_{i}, X_{i} \rangle \langle f_{j}, X_{j} \rangle]$ is equal to
    \begin{equation}\label{eqn:ci-elements}
        \bbE\left[\bbE[\langle f_{i}, X_{i} \rangle | \langle f_{k}, X_{k} \rangle : f_{k} \in \cH_{k} \mbox{ where } k \neq i,j] \bbE[\langle f_{j}, X_{j} \rangle | \langle f_{k}, X_{k} \rangle : f_{k} \in \cH_{k} \mbox{ where } k \neq i,j]\right]
    \end{equation}
    
    Define $\bbf = (f_{1}, \dots, f_{p})^{\top}$, $\bbf_{i} = (\bzero, \dots, \bzero, f_{i}, \bzero, \dots, \bzero)^{\top}$, $\bbf_{j} = (\bzero, \dots, \bzero, f_{j}, \bzero, \dots, \bzero)^{\top}$ and 
    $$\bbf_{ij} = (f_{1}, \dots, f_{i-1}, \bzero, f_{i+1}, \dots, f_{j-1}, \bzero, f_{j+1}, \dots, f_{p})^{\top}.$$
    Then $\bbf$, $\bbf_{i}$, $\bbf_{j}$ and $\bbf_{ij}$ can be thought of as elements of the product space $\cH$ and the random variables $\langle f_{i}, X_{i} \rangle$, $\langle f_{j}, X_{j} \rangle$ and $\langle f_{k}, X_{k} \rangle$ can be written as $\langle \bbf_{i}, X\rangle$, $\langle \bbf_{j}, X\rangle$ and $\langle \bbf_{ij}, X\rangle$ respectively.
    
    Notice that the space of random variables $\langle \bbf_{ij}, X \rangle$ under the inner product $(\langle \bbf_{ij}, X \rangle, \langle \bbg_{ij}, X \rangle) \mapsto \bbE[\langle \bbf_{ij}, X \rangle\langle \bbg_{ij}, X \rangle]$ is isomorphic to the reproducing kernel Hilbert space $\mathfrak{H}$ generated by the kernel $K(\bbf_{ij}, \bbg_{ij}) = \langle \bbf_{ij}, \bC \bbg_{ij} \rangle = \langle \bbf_{ij}, \bC_{ss} \bbg_{ij} \rangle$. By Lo\`{e}ve isometry, the expression (\ref{eqn:ci-elements}) can be rewritten as the inner product in $\mathfrak{H}$ of the elements $\bbf_{ij} \mapsto \langle \bC \bbf_{i}, \bbf_{ij} \rangle = \langle \bC_{si} \bbf_{i}, \bbf_{ij} \rangle, \bbf_{ij} \mapsto \langle \bC \bbf_{j}, \bbf_{ij} \rangle = \langle \bC_{sj} \bbf_{j}, \bbf_{ij} \rangle \in \mathfrak{H}$ which can be expressed as
    \begin{equation*}
        \langle \bC_{ss}^{-1/2} \bC_{si}^{\phantom{/}} \bbf_{i}, \bC_{ss}^{-1/2} \bC_{sj}^{\phantom{/}} \bbf_{j} \rangle 
        = \langle \bbf_{i}, [\bC_{ss}^{-1/2}\bC_{si}^{\phantom{.}}]^{\ast} [\bC_{ss}^{-1/2}\bC_{sj}^{\phantom{.}}] \bbf_{j} \rangle
    \end{equation*}
    which means that we can write $\bbE[\langle f_{i}, X_{i} \rangle \langle f_{j}, X_{j} \rangle] = \langle \bbf_{i}, \bC\bbf_{j} \rangle$ as
    \begin{equation*}\langle \bbf_{i}, \bC\bbf_{j} \rangle
        \langle \bbf_{i}, \bC\bbf_{j} \rangle = \langle \bbf_{i}, \bC_{ij}\bbf_{j} \rangle = \langle \bbf_{i}, [\bC_{ss}^{-1/2}\bC_{si}^{\phantom{.}}]^{\ast} [\bC_{ss}^{-1/2}\bC_{sj}^{\phantom{.}}]\bbf_{j} \rangle
    \end{equation*}
    or equivalently, $\bC_{ij} = [\bC_{ss}^{-1/2}\bC_{si}]^{\ast} [\bC_{ss}^{-1/2}\bC_{sj}]$. This can be written as $\bC_{ij} = \bC_{is}[\bC_{ss}]^{-1}\bC_{sj}$ if $\bC_{sj}$ is invertible with respect to $\bC_{ss}$ which can be easily shown to be equivalent to $\bR_{ij} = \bR_{is}[\bR_{ss}]^{-1}\bR_{sj}$. The conclusion in the general case follows from a density argument, namely that the space $\bC_{ss}\cL(\cH)$ is dense in $\bC_{ss}^{1/2}\cL(\cH)$ and the function $\bR_{sj} \mapsto (\dg \bC_{ss})^{1/2}\bR_{sj}(\bC_{jj})^{1/2} = \bC_{sj}$ is continuous.
\end{proof}

\subsection{Assumptions}

\begin{proof}[Proof of Remark \ref{rmk:lee-assumption}]
Assume that the regression operator $\bC_{-i,-i}^{\dagger}\bC_{-i,i}^{}$ is Hilbert-Schmidt. This implies that $\bC_{-i,i} = \bC_{-i,-i}\bA$ for some Hilbert-Schmidt operator $\bA$. Because $p(\dg \bC_{-i,-i}) \geq \bC_{-i,-i}$ and consequently $p^{2}(\dg \bC_{-i,-i})^{2} \geq \bC_{-i,-i}^{2}$, we can write using Douglas majorization \citep[see][Theorem 1]{douglas1966} that $\bC_{-i,-i} = (\dg \bC_{-i,-i})\bB$ for some bounded operator $\bB$. This implies that $\bC_{-i,i} = (\dg \bC_{-i,-i})\bB\bA$ for some Hilbert-Schmidt operator $\bB\bA$. It follows that $\bC_{ii}^{-1}\bC_{ij}^{}$ is Hilbert-Schmidt for every $i \neq j$. Now, let $\{(\mu_{k}, e_{k})\}_{k=1}^{\infty}$ $\{(\lambda_{l}, f_{l})\}_{l=1}^{\infty}$ be the eigenpairs of $\bC_{ii}$ and $\bC_{jj}$ respectively. Then $\bC_{ii}^{-1}\bC_{ij}^{}$ and $\bC_{jj}^{-1}\bC_{ji}^{}$ being Hilbert-Schmidt implies $\sum_{k,l=1}^{\infty} \langle e_{k}, \bC_{ij}f_{l}\rangle^{2}/\mu_{k}^{2} < \infty$ and $\sum_{k,l=1}^{\infty} \langle e_{k}, \bC_{ij}f_{l}\rangle^{2}/\lambda_{l}^{2} < \infty$. By Cauchy-Schwarz inequality, it follows that $\sum_{k,l=1}^{\infty} \langle e_{k}, \bC_{ij}f_{l}\rangle^{2}/\mu_{k}\lambda_{l} < \infty$ and thus, $\bR_{ij} = \bC_{ii}^{-1/2}\bC_{ij}^{}\bC_{jj}^{-1/2}$ is Hilbert-Schmidt for $i \neq j$ which is Assumption \ref{asm:equiv}/$1^{\ast}$. It follows that Hilbert-Schmidtness of regression operators $\bC_{-i,-i}^{\dagger}\bC_{-i,i}^{}$ implies that of the off-diagonal entries $\bR_{ij}$ of the correlation operator matrix $\bR$. 

To see why this is strict, consider the case $p = 2$ with $i = 1$ and $j = 2$, $\bC_{12} = \sum_{k=1}^{\infty} \alpha_{k}e_{k} \otimes f_{k}$ with $\lambda_{k} \sim 1/k^{1/2+\epsilon_{\lambda}}, \mu_{k} \sim 1/k^{1/2 + \epsilon_{\mu}}$ and $\alpha_{k} \sim 1/k^{1/2+\epsilon_{\alpha}}$ such that $\epsilon_{\lambda}> \epsilon_{\mu} > 0$ and $2\epsilon_{\alpha}-1 \in (\epsilon_{\mu}+\epsilon_{\lambda}, 2\epsilon_{\lambda})$. This ensures that $\sum_{k,l=1}^{\infty} \langle e_{k}, \bC_{ij}f_{l}\rangle^{2}/\mu_{k}\lambda_{l} < \infty$ while $\sum_{k,l=1}^{\infty} \langle e_{k}, \bC_{ij}f_{l}\rangle^{2}/\lambda_{l}^{2} = \infty$, implying that $\bR_{12}$ is Hilbert-Schmidt while the regression operators $\bC_{2,2}^{\dagger}\bC_{2,1}^{}$ isn't.
\end{proof}

\subsection{Methodology and Philosophy}

\begin{proof}[Proof of Lemma \ref{thm:likelihood}]
By Corollary 6.4.11 of \cite{bogachev1998}, we can write the log-likelihood $\log\left[\frac{d\bbP}{d\bbQ}\right]$ evaluated at $\mathrm{X}$ as 
\begin{equation*}
    \log \frac{d\bbP}{d\bbQ}(\mathrm{X}) = \frac{1}{2} \sum_{i=1}^{\infty} \left[ \frac{\lambda_{i}}{1 + \lambda_{i}} \left[ \sum_{j=1}^{\infty} \frac{1}{\sqrt{\mu_{j}}} \langle \varphi_{i} , \psi_{j} \rangle \langle \mathrm{X}, \psi_{j} \rangle \right]^{2} - \log (1 + \lambda_{i}) \right]
\end{equation*}
where $(\lambda_{j}, \varphi_{j})$ and $(\mu_{j}, \psi_{j})$ are the eigenpairs of $\bR_{0}$ and $\bC_{\bbQ}$ respectively. Let $\tilde{\lambda}_{j}$ be such that $1 + \tilde{\lambda}_{j} = (1 + \lambda_{j})^{-1}$. Then we have
\[
\begin{split}
    \int \log \frac{d\bbP}{d\bbQ}(\mathrm{X}) d\tilde{\bbP}(\mathrm{X}) 
    &= \frac{1}{2} \sum_{i=1}^{\infty} \left[ \frac{\lambda_{i}}{1 + \lambda_{i}} \int \left[ \sum_{j=1}^{\infty} \frac{1}{\sqrt{\mu_{j}}} \langle \varphi_{i} , \psi_{j} \rangle \langle \mathrm{X}, \psi_{j} \rangle \right]^{2} d\tilde{\bbP}(\mathrm{X}) - \log (1 + \lambda_{i}) \right] \\
    &= \frac{1}{2} \sum_{i=1}^{\infty} \left[ \frac{\lambda_{i}}{1 + \lambda_{i}} \int \left[ \sum_{j,j'=1}^{\infty} \frac{1}{\sqrt{\mu_{j}\mu_{j'}}} \langle \varphi_{i} , \psi_{j} \rangle \langle \varphi_{i} , \psi_{j'} \rangle \langle \mathrm{X}, \psi_{j} \rangle \langle \mathrm{X}, \psi_{j'} \rangle\right] d\tilde{\bbP}(\mathrm{X}) - \log (1 + \lambda_{i}) \right] \\
    &= \frac{1}{2} \sum_{i=1}^{\infty} \left[ \frac{\lambda_{i}}{1 + \lambda_{i}} \left[ \sum_{j,j'=1}^{\infty} \frac{1}{\sqrt{\mu_{j}\mu_{j'}}} \langle \varphi_{i} , \psi_{j} \rangle \langle \varphi_{i} , \psi_{j'} \rangle \int \langle \mathrm{X}, \psi_{j} \rangle \langle \mathrm{X}, \psi_{j'} \rangle d\tilde{\bbP}(\mathrm{X})\right] - \log (1 + \lambda_{i}) \right] \\
    &= \frac{1}{2} \sum_{i=1}^{\infty} \left[ \frac{\lambda_{i}}{1 + \lambda_{i}} \left[ \sum_{j,j'=1}^{\infty} \frac{1}{\sqrt{\mu_{j}\mu_{j'}}} \langle \varphi_{i} , \psi_{j} \rangle \langle \varphi_{i} , \psi_{j'} \rangle \langle \psi_{j}, \bC_{\tilde{\bbP}}\psi_{j'} \rangle\right] - \log (1 + \lambda_{i}) \right] \\
    &= \frac{1}{2} \sum_{i=1}^{\infty} \left[ \frac{\lambda_{i}}{1 + \lambda_{i}} \left\langle \varphi_{i}, \bC_{\bbQ}^{-1/2}\bC_{\tilde{\bbP}}^{\phantom{.}}\bC_{\bbQ}^{-1/2} \varphi_{i}\right\rangle - \log (1 + \lambda_{i}) \right] \\    
    &= \frac{1}{2} \sum_{i=1}^{\infty} \left[ \frac{\lambda_{i}}{1 + \lambda_{i}} (1 + \langle \varphi_{i}, \tilde{\bR}_{0}\varphi_{i}\rangle) - \log (1 + \lambda_{i}) \right] \\
    &= \frac{1}{2} \sum_{i=1}^{\infty} \left[ -\tilde{\lambda}_{i}\langle \varphi_{i}, \tilde{\bR}_{0}\varphi_{i}\rangle + \log (1 + \tilde{\lambda}_{i}) - \tilde{\lambda}_{i} \right] \\
    &= \frac{1}{2} \left[ -\left\langle \sum_{i=1}^{\infty} \tilde{\lambda}_{i}\varphi_{i} \otimes \varphi_{i}, \tilde{\bR}_{0}\right\rangle + \sum_{i=1}^{\infty}  \log (1 + \tilde{\lambda}_{i}) - \tilde{\lambda}_{i} \right] \\
    &= \frac{1}{2} \left[ -\tr(\bH\tilde{\bR}_{0}) + \log \det\nolimits_{2}(\bI + \bH) \right],
\end{split}
\]
where $\bH = (\bI + \bR_{0})^{-1} - \bI$. This establishes the claim.
\end{proof}

\subsection{Dual Problem}
\begin{proof}[Proof of Theorem \ref{thm:dual-problem}]
Let $\cG_{1}[\bA] = - \log\det\nolimits_{2}(\bI + \bA)$ and $\cG_{2}[\bA] = \lambda_{n}\|\bA\|_{2, 1}$. It can be shown that
\begin{eqnarray*}
    &\cG_{1}^{\ast}[\bB] &= \max_{\bA} \left[\tr(\bA\bB) + \log \det\nolimits_{2}(\bI + \bA)\right] \\
    &&= \tr([(\bI - \bB)^{-1} - \bI]\bB) + \log \det\nolimits_{2}(\bI + [(\bI - \bB)^{-1} - \bI]]) \\
    &&= -\left[ \tr([(\bI - \bB)^{-1} - \bI][-\bB]) - \log \det\nolimits_{2}(\bI + [(\bI - \bB)^{-1} - \bI]) \right].
\end{eqnarray*}
Note that if we replace $\bB$ with $-\bB_{0}$ such that $\dg \bB_{0} = \bzero$, the above expression is equal to twice the Kullback-Leibler divergence $\cD(\bB_{0})$ of the Gaussian measure with the correlation operator $\bI + \bB_{0}$ with respect to its product measure since $\cD(\bB_{0}) = -\frac{1}{2}\left[ \tr([(\bI + \bB_{0})^{-1} - \bI][\bB_{0}]) - \log \det\nolimits_{2}(\bI + \bB_{0}) \right]$. 

If $\bR_{0}$ is trace-class, so is $\tilde{\bR} = (\bI + \bR_{0})^{-1} - \bI$, and we can write
\begin{eqnarray*}
    &\cD(\bR_{0}) &= -\frac{1}{2}\left[ \tr(\tilde{\bR}\bR_{0}) - \log \det\nolimits(\bI + \tilde{\bR}) + \tr(\tilde{\bR})\right] \\
    &&= -\frac{1}{2}\left[ \tr(\tilde{\bR}[\bR_{0} + \bI]) - \log \det\nolimits(\bI + \tilde{\bR})\right] \\
    &&= -\frac{1}{2}\left[ \tr(\bI - [\bR_{0} + \bI]) - \log \det\nolimits(\bI + \tilde{\bR})\right] \\
    &&= -\frac{1}{2} \log \det\nolimits(\bI + \bR_{0}) = -\frac{1}{2} \log \det\nolimits_{2}(\bI + \bR_{0}). \\
\end{eqnarray*}
Since the expression is continuous in the Hilbert-Schmidt norm, the result holds even for Hilbert-Schmidt $\bR_{0}$. Finally,
\begin{eqnarray*}
    \cG_{2}^{\ast}[\bB] = \begin{cases}
        0 &\mbox{ if } \dg \bB = \bzero \mbox{ and } \|\bB_{0}\|_{2, \infty} \leq \lambda_{n}\\
        \infty &\mbox{ otherwise.}
    \end{cases}
\end{eqnarray*}
By combining these two using infimal convolution, we get
\begin{eqnarray*}
    &\cG^{\ast}(-\hat{\bR}_{0}) &= \inf_{\bB} \left[ \cG_{1}^{\ast}[\bB] + \cG_{2}^{\ast}(-\hat{\bR}_{0} - \bB) \right] \\
    &&= \inf \{ 2\cD(- \bB_{0}): \|-\hat{\bR}_{0} - \bB_{0}\|_{2, \infty} \leq \lambda_{n} \} \\
    &&= 2\inf \{ \cD(\bB_{0}): \|\bB_{0} - \hat{\bR}_{0}\|_{2, \infty} \leq \lambda_{n} \} \\
\end{eqnarray*}
\end{proof}
\subsection{Existence and Uniqueness of Minimizer}
The proof below relies on some basic results in convex analysis in Hilbert spaces that can be consulted in \cite{bauschke2011} or \cite{ekeland1999}.
\begin{lemma}\label{lem:cf_determinant}
    The functional $\bA \mapsto \log \det\nolimits_{2} (\bI + \bA)$ is twice differentiable in the G\^{a}teux sense, with the first and second G\^{a}teux derivatives at $\bA$ given by $(\bI + \bA)^{-1} - \bI$  and $[(\bI + \bA) \otimes (\bI + \bA)]^{-1}$, respectively. 
\end{lemma}
\begin{proof}
    We simply evaluate the derivative of $f$ at $t = 0$ by looking at its Taylor expansion. For $\bA, \bB \in \cL_{1}$:
	\begin{equation*}
		\begin{split}
			f(t) - f(0) 
			&= \log \det\nolimits_{2} (\bI + \bA + t\bB) - \log \det\nolimits_{2} (\bI + \bA) \\
			&= \log \det (\bI + \bA + t\bB) - \log \det (\bI + \bA) - \tr (\bA + t \bB) + \tr \bA\\
			&= \log \det \big[ \bI + t(\bI + \bA)^{-1}\bB) \big] - t\tr (\bB) \\
			&= \left[ t \tr \big[ (\bI + \bA)^{-1}\bB \big] - \frac{1}{2} t^{2} \tr \big\lbrace \big[ (\bI + \bA)^{-1}\bB \big]^{2} \big\rbrace + o(t^{3}) \right] - t\tr (\bB) \\
			&= t \tr \big\lbrace [(\bI + \bA)^{-1} - \bI] \bB \big\rbrace - \frac{1}{2!} t^{2} \tr \big\lbrace \big[ (\bI + \bA)^{-1}\bB \big]^{2} \big\rbrace + o(t^{3})
		\end{split}
	\end{equation*}
	The result follows from the continuity of expressions in $\| \cdot \|_{2}$ norm and the fact that $\cL_{1}$ is dense in $\cL_{2}$. 
\end{proof}

\begin{proof}[Proof of Lemma \ref{thm:existence}]
    The Carleman-Fredholm determinant is known to be strictly log-concave (see Lemma 2.1 of \cite{bakonyi1998}). From Theorem 6.5 of \cite{simon1977},
	\begin{eqnarray*}
		&| \det\nolimits_{2}(\bI + \bA) - \det\nolimits_{2}(\bI + \bB) | &\leq \| \bA - \bB \|_{2} \exp \Big[ \tfrac{1}{2}(\| \bA \|_{2} + \| \bB \|_{2} + 1)^{2}\Big].
	\end{eqnarray*}
    Thus the function $\bA \mapsto - \log \det\nolimits_{2} (\bI + \bA)$ is strictly convex and continuous in $\| \cdot \|_{2}$. Because the functions $\bH \mapsto \mathrm{tr} (\bH \hat{\bR}_{0})$ are $\bH \mapsto \|\bH_{0} \|_{1} = \sum_{i \neq j} \| \bH_{ij} \|_{2}$ are also convex, it follows that the function $\bH \mapsto F(\bH)$ is strictly convex and continuous.
	
    Using the method of Lagrange multipliers, we can rewrite the optimization problem (\ref{eqn:optimisation-problem}) in a constrained form as
    \begin{equation*}
        \inf_{\| \bH_{0} \|_{1} \leq r} \left[\tr(\bH\hat{\bR}_{0}) - \log \det\nolimits_{2} (\bI + \bH)\right] 
    \end{equation*}
    for some $r = r(\lambda_{n})$. Notice that $\mathrm{tr} (\bH \hat{\bR}_{0}) = \sum_{i \neq j} \tr_{i}(\bH_{ij}\hat{\bR}_{ji})$ depends only on $\bH_{0}$ and is hence bounded. Let $\lambda_{\infty} = \max_{j} \lambda_{j}(\bH)$. Using $\lambda_{j}(\bH) > - 1$ we can write
    \begin{eqnarray*}
        &[\det\nolimits_{2}(\bI + \bH)]^{-1} 
        &= \prod_{j=1}^{\infty} e^{\lambda_{j}(\bH)}(1 + \lambda_{j}(\bH))^{-1} \\
        &&= \prod_{j=1}^{\infty} 
        \left[ 
            1 + \frac{1}{(1 + \lambda_{j}(\bH))} 
            \sum_{k = 2}^{\infty} \frac{\lambda_{j}^{k}(\bH)}{k!}           
        \right] \\
        &&= \prod_{j=1}^{\infty} 
        \left[ 
            1 + \frac{1}{(1 + \lambda_{j}(\bH))} 
            \sum_{m = 1}^{\infty} \frac{\lambda_{j}^{2m}(\bH)}{(2m)!} \left( 1 + \frac{\lambda_{j}(\bH)}{2m+1} \right)           
        \right] \\
        &&\geq \prod_{j=1}^{\infty} 
        \left[ 
            1 + \frac{1}{(1 + \lambda_{\infty})} 
            \frac{\lambda_{j}^{2}(\bH)}{2!} \left( 1 - \frac{1}{3} \right)           
        \right] \\
        &&\geq
            1 + \frac{1}{3(1 + \lambda_{\infty})} 
            \sum_{j=1}^{\infty}\lambda_{j}^{2}(\bH) \\
        &&=
            1 + \frac{1}{3(1 + \lambda_{\infty})} 
            \|\bH\|_{2}^{2}
    \end{eqnarray*}
    In the first inequality, we used the fact that $1 + \lambda_{j}(\bH)/(2m+1) > 0$ and we retained only the first term of the infinite sum.
    Thus the Carleman-Fredholm determinant $-\log \det\nolimits_{2} (\bI + \bH) \to \infty$ as $\|\bH\|_{2} \to \infty$ and is therefore coercive. It immediately follows that $\mathcal{F}$ admits a unique minimum, say at $\hat{\bH}$ (Propostion 1.2, \cite{ekeland1999}). Consequently, it must satisfy the stationary condition (Theorem 16.3, \cite{bauschke2011}) at $\hat{\bH}$ given by
    \begin{equation*}
        \bzero \in \partial \mathcal{F}(\hat{\bH}).
    \end{equation*}
    Because $\bH \mapsto \tr(\bH\hat{\bR}_{0})$ and $\bH \mapsto \log \det\nolimits_{2} (\bI + \bH)$ are G\^{a}teaux differentiable with the G\^{a}teaux derivatives at $\hat{\bH}$ given by $\hat{\bR}_{0}$ and $(\bI + \hat{\bH})^{-1} - \bI$ respectively, this is equivalent to saying that there exists $\hat{\bZ} \in \partial\|\hat{\bH}_{0} \|_{1}$ such that
    \begin{equation*}
        \hat{\bR}_{0} - \left[(\bI + \hat{\bH})^{-1} - \bI\right] + \lambda_{n} \hat{\bZ} = \hat{\bR} - (\bI + \hat{\bH})^{-1} + \lambda_{n} \hat{\bZ} = \bzero.
    \end{equation*}
    Hence proved.
\end{proof}

\subsection{Proof of Main Result}
Our proof is a nontrivial adaptation of the multivariate proof in \cite{ravikumar2011} to the more general setting of elements and operators in Hilbert spaces. 
We closely follow the steps in the original proof and adapt them so as to not rely on topological properties such as the compactness of closed and bounded sets which are absent in infinite-dimensional spaces. 
We will use the new norms that we have developed to mimic the topology of Euclidean space in the product Hilbert space $\cH$. Crucially, controlling the difference between the oracle solution and the true solution (Lemma \ref{lem:control_of_D}) is achieved using the Banach fixed-point theorem instead of Brouwer's fixed-point theorem.

Let $\tilde{\bH}$ be the solution of the original optimization problem (\ref{eqn:optimisation-problem}) assuming that the support $S = \{ (i, j) :\bH^{\ast}_{ij} \neq \bzero\}$ of $\bH^{\ast}$ is known. Thus,
\begin{equation}\label{eqn:opt_problem_Sc}
    \tilde{\bH} = \argmin_{\bH_{S^{c}} = \bzero} \left[\tr(\bH\hat{\bR}_{0}) - \log \det\nolimits_{2} (\bI + \bH) + \lambda_{n}\|\bH_{0} \|_{2, 1}\right].
\end{equation}
The essence of the proof is to show that the \emph{oracle} solution $\tilde{\bH}$ assuming the support is known, is equal to the solution $\hat{\bH}$ with high probability. Because the oracle solution $\tilde{\bH}$ possesses nice properties, the same applies to $\hat{\bH}$ with high probability. To show that $\tilde{\bH}$ solves the original problem we will show that it satisfies the stationary condition (\ref{eqn:stationary_condn}) for some $\tilde{\bZ} \in \partial\|\tilde{\bH}_{0} \|_{2, 1}$. We will therefore construct a \emph{witness} $\tilde{\bZ} \in \partial\|\tilde{\bH}_{0} \|_{2, 1}$ such that
\begin{equation}\label{eqn:stationary_condn_Sc}
    \hat{\bR} - (\bI + \tilde{\bH})^{-1} + \lambda_{n} \tilde{\bZ} = \bzero
\end{equation}
holds. 

Notice that because $\tilde{\bH}$ solves (\ref{eqn:opt_problem_Sc}), it satisfies the corresponding stationary condition, which can be written as
\begin{equation*}
    \hat{\bR}_{S} - (\bI + \tilde{\bH})^{-1}_{S} + \lambda_{n} \tilde{\bZ}_{S} = \bzero
\end{equation*}
for some $\bZ_{S} \in \partial\left[\sum_{(i,j) \in S, i\neq j}\|\tilde{\bH}_{ij}\|_{2}\right]$ (note that this implies that $\|\bZ_{S}\|_{2, \infty} \leq 1$). Thus (\ref{eqn:stationary_condn_Sc}) is already satisfied for the entries $(i, j) \in S$. To ensure that (\ref{eqn:stationary_condn_Sc}) is satisfied for all entries, we simply define $\bZ_{S^{c}}$ as follows:
\begin{equation*}
    \tilde{\bZ}_{S^{c}} = \frac{1}{\lambda_{n}} \left[ -\hat{\bR}_{S^{c}} + [(\bI + \tilde{\bH})^{-1}]_{S^{c}} \right].
\end{equation*}
Now that (\ref{eqn:stationary_condn_Sc}) is satisfied, we need to show that $\tilde{\bZ} \in \partial\|\tilde{\bH}_{0}\|_{2,1}$, which obviously holds for the entries $(i,j) \in S$ by construction. For $(i,j) \in S^{c}$, notice that $\tilde{\bH}_{S^{c}} = \bzero$ and therefore its subdifferential $\partial\|\tilde{\bH}_{S^{c}} \|_{2,1}$ has a particularly simple form given by $\partial\|\tilde{\bH}_{S^{c}} \|_{2,1} = \{\bZ_{S^c} : \| \bZ_{S^c} \|_{2, \infty} < 1 \}$. The proof thus reduces to showing that $\tilde{\bZ}_{S^{c}}$ as defined above satisfies
\begin{equation}\label{eqn:strict_dual_feasibility_defn}
    \| \tilde{\bZ}_{S^c} \|_{2, \infty} < 1.
\end{equation}
with high probability. The condition (\ref{eqn:strict_dual_feasibility_defn}) is known as \emph{strict dual feasibility}.

We begin by showing in Lemma \ref{lem:strict_dual_feasibility} that strict dual feasibility is satisfied if the \emph{sampling noise} $\bW$ and the \emph{remainder term} $\cE(\bD)$ are small enough compared to $\lambda_{n}$ where 
\begin{equation}\label{eqn:noise}
    \bW = \hat{\bR} - (\bI + \bH^{\ast})^{-1} = \hat{\bR} - \bR^{\ast}, \mbox{ and }
\end{equation}
\begin{equation}\label{eqn:remainder}
    \cE(\bD) = (\bI + \tilde{\bH})^{-1} - (\bI + \bH^{\ast})^{-1} + (\bI + \bH^{\ast})^{-1}\bD(\bI + \bH^{\ast})^{-1} \mbox{ with } \bD = \tilde{\bH} - \bH^{\ast}.
\end{equation}
Then we show that $\cE(\bD)$ can be controlled by controlling $\bD$ in Lemma \ref{lem:control_of_remainder}. And finally, we prove in Lemma \ref{lem:control_of_D} that by controlling $\bW$ and $\lambda_{n}$ we can control $\bD$. In summary, choosing $\lambda_{n}$ appropriately and having enough samples to keep $\bW$ is small, ensures that $\tilde{\bH} = \hat{\bH}$, thus transferring the nice oracle properties that $\tilde{\bH}$ possesses to $\hat{\bH}$.

\begin{proof}[Proof of Theorem \ref{thm:general}]

    Notice that the deviations of $\bW = \hat{\bR}-\bR$ satisfy
    \begin{equation*}
        \bbP[\|\bW_{ij}\|_{2} \geq \delta] \leq 1/g(n, \delta)
    \end{equation*}
    where $g(n, \delta) = f(n, [\delta/\kappa]^{1+1/\beta \wedge 1})$. 
    
    According to Lemma \ref{lem:strict_dual_feasibility}, we need to choose the tuning parameter $\lambda_{n}$ such that for a large enough sample size $n$ we would have strict dual feasibility, which boils down to:
    \begin{eqnarray*}
        &\|\bW\|_{2,\infty} &\leq \frac{\alpha \lambda_{n}}{8}, \\
        &\|\cE(\bD)\|_{2,\infty} &\leq \frac{\alpha \lambda_{n}}{8}.
    \end{eqnarray*}
    The first of the above conditions is satisfied with probability greater than $1-1/p^{\tau}$ by requiring $\lambda_{n} = 8\bar{\delta}_{g}(n, p^{\tau})/\alpha$ since $\|\bW\|_{2, \infty} \leq \|\bW\|_{2,2} \leq \bar{\delta}_{g}(n, p^{\tau})$. It turns out that the second condition is also satisfied if we require that $n$ be large enough such that
    \begin{equation}\label{eqn:delta-bound}
        2\gamma \left( 1 + \frac{8}{\alpha} \right)^{2} \bar{\delta}_{g}(n, p^{\tau}) \leq \frac{1}{3\rho d \vee 6\rho^{3}\gamma d}.
    \end{equation}
    Indeed, by Lemma \ref{lem:control_of_D}, we have $\|\bD\|_{2,\infty} \leq 2 \gamma (\|\bW\|_{2,\infty} + \lambda_{n})$ since
    \begin{equation}\label{eqn:d-bound}
        2 \gamma (\|\bW\|_{2,\infty} + \lambda_{n}) \leq 2\gamma \left( 1 + \frac{8}{\alpha} \right) \bar{\delta}_{g}(n, p^{\tau}) \leq 2\gamma \left( 1 + \frac{8}{\alpha} \right)^{2} \bar{\delta}_{g}(n, p^{\tau}) \leq \frac{1}{3\rho d \vee 6\rho^{3}\gamma d} \leq \frac{1}{3 \rho d}
    \end{equation}
    and therefore, by Lemma \ref{lem:control_of_remainder}, we have from the bound (\ref{eqn:d-bound}) on $\|\bD\|_{2,\infty}$, 
    \begin{eqnarray*}
        &\|\cE(\bD)\|_{2,\infty} &\leq \frac{3}{2}d \|\bD\|_{2, \infty}^{2}\rho^{3} \\
        &&\leq 6 \rho^{3}\gamma^{2} d \cdot \left( 1 + \frac{8}{\alpha} \right)^{2} \bar{\delta}_{g}(n, p^{\tau})^{2}\\
        &&\leq 6 \rho^{3}\gamma^{2} d \cdot \left( 1 + \frac{8}{\alpha} \right)^{2} \bar{\delta}_{g}(n, p^{\tau}) \cdot \bar{\delta}_{g}(n, p^{\tau}) \\
        &&\leq 6 \rho^{3}\gamma^{2} d \cdot \frac{1}{6\rho\gamma d \vee 12 \rho^{3}\gamma^{2} d} \cdot \bar{\delta}_{g}(n, p^{\tau}) \\
        &&\leq \frac{1}{2} \cdot \bar{\delta}_{g}(n, p^{\tau}) \leq \bar{\delta}_{g}(n, p^{\tau}) = \frac{\alpha \lambda_{n}}{8}. \\
    \end{eqnarray*}
    Now, (\ref{eqn:delta-bound}) actually follows from the given condition on the sample size $n$ because it implies
    \begin{equation*}
        \bar{\delta}_{g}(n, p^{\tau}) = \kappa \bar{\delta}_{f}(n, p^{\tau})^{\beta\wedge 1/1+\beta\wedge 1} \leq \left[ 12 (1 + 8/\alpha)^{2} d (\rho \gamma \vee \rho^{3} \gamma^{2})   \right]^{-1}.
    \end{equation*}
    Therefore, we have $\hat{\bH} = \bH$ with probability $1 - p^{\tau}$ and when this is indeed true, we can write that $\bD = \|\hat{\bH} - \bH\|_{2, \infty} \leq 2\gamma(1+8/\alpha)\bar{\delta}_{g}(n, p^{\tau})$. It follows that if $\|\bH_{ij}\|_{2} > 2\gamma(1+8/\alpha)\bar{\delta}_{g}(n, p^{\tau})$, then $\hat{\bH}_{ij} \neq \bzero$. The result follows from rewriting $\bar{\delta}_{g}(n, p^{\tau})$ in terms of $\bar{\delta}_{f}(n, p^{\tau})$.
\end{proof}

\subsubsection{Strict Dual Feasibility} 
For an operator $\bA$, let $\overline{\bA}$ denote the column vector $[\bA_{ij}: (i, j) = (1,1), (1, 2), \dots, (p, p-1), (p,p)]^{\top}$ indexed by the pairs $(i,j)$, instead of $i$ and $j$ separately as in the original matrix form. In other words, $\overline{\bA}$ is the vectorized version of the operator matrix $\bA$. Denote the \emph{subvectors} $[\bA_{ij}: (i, j) \in S]^{\top}$ and $[\bA_{ij}:(i, j) \in S^{c}]^{\top}$ as $\overline{\bA}_{S}$ and $\overline{\bA}_{S^{c}}$ respectively.

\begin{lemma}[Strict Dual Feasibility]\label{lem:strict_dual_feasibility}
    Under the following condition, we have $\|\tilde{\bZ}_{S^{c}}\|_{2, \infty} < 1$ and hence, $\tilde{\bH} = \hat{\bH}$. 
    \begin{equation}\label{eqn:strict_dual_feasibility}
        \|\bW\|_{2, \infty} \vee \|\cE(\bD)\|_{2, \infty} \leq \frac{\alpha \lambda_{n}}{8}
    \end{equation}
\end{lemma}
\begin{proof}
    We rewrite Equation (\ref{eqn:stationary_condn}) using Equations (\ref{eqn:noise}) and (\ref{eqn:remainder}) as
    \begin{equation}\label{eqn:equiv_stationary_condn}
        (\bI + \bH^{\ast})^{-1}\bD(\bI + \bH^{\ast})^{-1} + \bW - \cE(\bD) + \lambda_{n}\tilde{\bZ} = \bzero.
    \end{equation}
    Let $\Gamma$ denote the outer product of $(\bI + \bH^{\ast})^{-1}$ with itself. Then,
    \begin{equation*}
        \overline{(\bI + \bH^{\ast})^{-1}\bD(\bI + \bH^{\ast})^{-1}} = \left[(\bI + \bH^{\ast})^{-1} \otimes (\bI + \bH^{\ast})^{-1}\right]\overline{\bD} = \Gamma \overline{\bD}.
    \end{equation*}
    By vectorizing Equation (\ref{eqn:equiv_stationary_condn}), subsetting on $S$ and $S^{c}$, and using $\overline{\bD}_{S^{c}} = \bzero$, we can write 
    \begin{eqnarray}
        \Gamma_{SS}\overline{\bD}_{S} + \overline{\bW}_{S} - \overline{\bE}_{S} + \lambda_{n}\overline{\tilde{\bZ}}_{S} 
        &=& \bzero, \\
        \Gamma_{S^{c}S}\overline{\bD}_{S} 
        + \overline{\bW}_{S^{c}} - \overline{\bE}_{S^{c}} + \lambda_{n}\overline{\tilde{\bZ}}_{S^{c}} 
        &=& \bzero.
    \end{eqnarray}
    for $\bE = \cE(\bD)$. Notice that $\Gamma_{SS}$ is invertible 
    and therefore we can solve the above system of equations for $\overline{\bD}_{S}$ and $\overline{\tilde{\bZ}}_{S^{c}}$ as follows:
    \begin{eqnarray}
        \overline{\bD}_{S} 
        &=& \Gamma_{SS}^{-1} \Big[-\overline{\bW}_{S} + \overline{\bE}_{S} - \lambda_{n} \overline{\tilde{\bZ}}_{S}\Big] \\
        \overline{\tilde{\bZ}}_{S^{c}} 
        &=& -\frac{1}{\lambda_{n}} \Gamma_{S^{c}S}^{\phantom{.}}\Gamma_{SS}^{-1}(\overline{\bW}_{S} - \overline{\bE}_{S}) 
        - \frac{1}{\lambda_{n}} (\overline{\bW}_{S^{c}} - \overline{\bE}_{S^{c}}) 
        + \Gamma_{S^{c}S}^{\phantom{.}}\Gamma_{SS}^{-1}\overline{\tilde{\bZ}}_{S}
    \end{eqnarray}
    Then by taking the $\|\cdot\|_{2, \infty}$-norm,
    \begin{eqnarray}
        \| \overline{\tilde{\bZ}}_{S^{c}}  \|_{2, \infty} 
        &\leq&
        \frac{1}{\lambda_{n}} \tnorm{\Gamma_{S^{c}S}^{\phantom{.}}\Gamma_{SS}^{-1}}_{2,\infty} \Big(\|\overline{\bW}_{S}\|_{2, \infty} + \|\overline{\bE}_{S}\|_{2, \infty}\Big) \\
        &+& \frac{1}{\lambda_{n}} \Big(\|\overline{\bW}_{S^{c}}\|_{2, \infty} + \|\overline{\bE}_{S^{c}}\|_{2, \infty}\Big) + \|\Gamma_{S^{c}S}^{\phantom{.}}\Gamma_{SS}^{-1}\overline{\tilde{\bZ}}_{S}\|_{2, \infty}.
    \end{eqnarray}
    By Assumption \ref{asm:incoherence}, we have $\tnorm{\Gamma_{S^{c}S}^{\phantom{.}}\Gamma_{SS}^{-1}}_{2,\infty} = \max_{e \in S^{c}} \|\Gamma_{eS}^{\phantom{.}}\Gamma_{SS}^{-1}\|_{2, 1} \leq 1 - \alpha$ and using $\|\overline{\tilde{\bZ}}_{S}\|_{2, \infty} \leq 1$ which follows by construction, we get
    \begin{eqnarray}
        \|\Gamma_{S^{c}S}^{\phantom{.}}\Gamma_{SS}^{-1}\overline{\tilde{\bZ}}_{S}\|_{2, \infty} 
        &=& \max_{e \in S^{c}} \|\Gamma_{eS}^{\phantom{.}}\Gamma_{SS}^{-1}\overline{\tilde{\bZ}}_{S}\|_{2} \\
        &\leq& \max_{e \in S^{c}} \|\Gamma_{eS}^{\phantom{.}}\Gamma_{SS}^{-1}\|_{2,1}\|\overline{\tilde{\bZ}}_{S}\|_{2, \infty} \leq (1 - \alpha).
    \end{eqnarray}
    The above bounds together with the inequality (\ref{eqn:strict_dual_feasibility}) imply that
    \begin{eqnarray}
        \| \overline{\tilde{\bZ}}_{S^{c}}  \|_{2, \infty} 
        &\leq& \frac{1-\alpha}{\lambda_{n}} \left(\frac{\alpha\lambda_{n}}{4}\right) + \frac{1}{\lambda_{n}} \left(\frac{\alpha\lambda_{n}}{4}\right) + (1 - \alpha) = 1 - \alpha/2 - \alpha^{2}/4 < 1.
    \end{eqnarray}
\end{proof}
Thus, controlling the noise level $\bW$ and the remainder term $\bE = \cE(\bD)$ enables us to enforce strict dual feasibility. We now show how the remainder term itself can be controlled by controlling $\bD = \tilde{\bH} - \bH^{\ast}$.

\subsubsection{Control of Remainder}
\begin{lemma}[Control of Remainder]\label{lem:control_of_remainder}
    Let $\bJ = \sum_{k = 0}^{\infty} (-1)^{k}\left[(\bI + \bH^{\ast})^{-1}\bD\right]^{k}$. Then,
    \begin{equation*}
        \cE(\bD) = (\bI + \bH^{\ast})^{-1}\bD(\bI + \bH^{\ast})^{-1}\bD\bJ(\bI + \bH^{\ast})^{-1}.
    \end{equation*}
    If $\|\bD\|_{2, \infty} \leq 1/3\rho d$, then $\tnorm{\bJ^{\top}}_{2,\infty} \leq 3/2$ and $\|\cE(\bD)\|_{2, \infty} \leq \frac{3}{2}d \|\bD\|_{2, \infty}^{2}\rho^{3}$.
\end{lemma}
\begin{proof}   
    Recall that 
    \begin{equation*}
        \cE(\bD) = (\bI + \bH^{\ast} + \bD)^{-1} - (\bI + \bH^{\ast})^{-1} + (\bI + \bH^{\ast})^{-1}\bD(\bI + \bH^{\ast})^{-1}.
    \end{equation*}
    Because $\bD$ has no more than $d$ non-zero entries in every row or column, we have $\tnorm{\bD}_{2,\infty} \leq d \|\bD\|_{2, \infty}$. Using sub-additivity and sub-multiplicativity of $\tnorm{\cdot}_{2,\infty}$ we can write
    \begin{equation*}
        \tnorm{(\bI + \bH^{\ast})^{-1}\bD}_{2,\infty} 
        = \tnorm{\bD + \bR^{\ast}_{0}\bD}_{2,\infty}        
        \leq \tnorm{\bD}_{2,\infty} + \tnorm{\bR^{\ast}_{0}}_{2,\infty}  \tnorm{\bD}_{2,\infty} \leq \rho d \|\bD\|_{2, \infty} < 1/3.
    \end{equation*}
    By expanding $(\bI + (\bI + \bH^{\ast})^{-1}\bD)^{-1}$ into a geometric series we can write
    \begin{eqnarray*}
        (\bI + \bH^{\ast} + \bD)^{-1} 
        &=& (\bI + (\bI + \bH^{\ast})^{-1}\bD)^{-1}(\bI + \bH^{\ast})^{-1} \\
        &=& (\bI + \bH^{\ast})^{-1} - (\bI + \bH^{\ast})^{-1}\bD(\bI + \bH^{\ast})^{-1} \\
        &+& \sum_{k = 2}^{\infty} (-1)^{k}\left[(\bI + \bH^{\ast})^{-1}\bD\right]^{k}  (\bI + \bH^{\ast})^{-1} \\
        &=& (\bI + \bH^{\ast})^{-1} - (\bI + \bH^{\ast})^{-1}\bD(\bI + \bH^{\ast})^{-1}\\
        &+& ~(\bI + \bH^{\ast})^{-1}\bD(\bI + \bH^{\ast})^{-1}\bD\bJ(\bI + \bH^{\ast})^{-1}
    \end{eqnarray*}
    where $\bJ = \sum_{k = 0}^{\infty} (-1)^{k}\left[(\bI + \bH^{\ast})^{-1}\bD\right]^{k}$ which satisfies,
    \begin{equation*}
        \tnorm{\bJ}_{2,1}
        = \textstyle\tnorm{\bJ^{\top}}_{2,\infty}         
        \leq \left[1 - \tnorm{\bD}_{2,\infty}\left(1 + \tnorm{\bR^{\ast}_{0}}_{2,\infty}\right)\right]^{-1} 
        \leq 3/2.
    \end{equation*}
    It follows from the above expansion that 
    \begin{equation*}
        \cE(\bD) = (\bI + \bH^{\ast})^{-1}\bD(\bI + \bH^{\ast})^{-1}\bD\bJ(\bI + \bH^{\ast})^{-1}
    \end{equation*}
    and therefore, we can control the remainder term as follows:
    \begin{eqnarray*}
        \|\cE(\bD)\|_{2, \infty} 
        &=& \|(\bI + \bH^{\ast})^{-1}\bD(\bI + \bH^{\ast})^{-1}\bD\bJ(\bI + \bH^{\ast})^{-1}\|_{2, \infty} \\
        &=& \max_{i,j}\|e_{i}^{\top}(\bI + \bH^{\ast})^{-1}\bD(\bI + \bH^{\ast})^{-1}\bD\bJ(\bI + \bH^{\ast})^{-1}e_{j}\|_{2} \\
        &\leq& \max_{i}\|e_{i}^{\top}(\bI + \bH^{\ast})^{-1}\bD\|_{2, \infty} \cdot \max_{j}\|(\bI + \bH^{\ast})^{-1}\bD\bJ(\bI + \bH^{\ast})^{-1}e_{j}\|_{2,1}
    \end{eqnarray*} 
    where $\{e_{i}\}_{i = 1}^{p}$ are ``unit vectors" given by $e_{i} = (\delta_{ij}\bI_{j})_{j = 1}^{p}$ where $\delta_{ij}$ is the Kronecker delta. The first term can be bounded as
    \begin{eqnarray*}
        \max_{i}\|e_{i}^{\top}(\bI + \bH^{\ast})^{-1}\bD\|_{2, \infty} 
        &=& \max_{i}\|e_{i}^{\top}(\bD + \bR^{\ast}_{0}\bD)\|_{2, \infty} \\
        &\leq& \max_{i}\|e_{i}^{\top}\bD\|_{2, \infty} + \max_{i}\|e_{i}^{\top}\bR^{\ast}_{0}\bD\|_{2, \infty} \\
        &\leq& \|\bD\|_{2, \infty} + \max_{i}\|e_{i}^{\top}\bR^{\ast}_{0}\|_{2, 1}\|\bD\|_{2, \infty} \\
        &=& (1 + \tnorm{\bR^{\ast}_{0}}_{2,\infty})\|\bD\|_{2, \infty} = \rho\|\bD\|_{2, \infty}
    \end{eqnarray*}
    and the second term as
    \begin{eqnarray*}
        \max_{j}\|(\bI + \bH^{\ast})^{-1}\bD\bJ(\bI + \bH^{\ast})^{-1}e_{j}\|_{2,1} 
        &=& \tnorm{(\bI + \bH^{\ast})^{-1}\bD\bJ(\bI + \bH^{\ast})^{-1}}_{2,1} \\
        &=& \tnorm{(\bI + \bH^{\ast})^{-1}\bJ^{\top}\bD(\bI + \bH^{\ast})^{-1}}_{2,\infty}. \\
    \end{eqnarray*}
    By substituting $(\bI+ \bH^{\ast})^{-1} = \bI + \bR^{\ast}_{0}$ and using sub-additivity and sub-multiplicativity of $\tnorm{\cdot}_{2,\infty}$ we can bound the above term by
    \begin{eqnarray*}
        &\leq& \textstyle \tnorm{\bJ^{\top}\bD}_{2,\infty} 
        + \tnorm{\bR^{\ast}_{0}\bJ^{\top}\bD}_{2,\infty} 
        + \tnorm{\bJ^{\top}\bD\bR^{\ast}_{0}}_{2,\infty}  
        + \tnorm{\bR^{\ast}_{0}\bJ^{\top}\bD\bR^{\ast}_{0}}_{2,\infty} \\
        &\leq& \frac{3}{2} \tnorm{\bD}_{2,\infty} (1 + \tnorm{\bR^{\ast}_{0}}_{2,\infty})^{2} \\
        &\leq& \frac{3}{2}d \|\bD\|_{2, \infty} \rho^{2}
    \end{eqnarray*}
    It follows by combining all the above estimates that
    \begin{equation*}
        \|\cE(\bD)\|_{2, \infty} \leq \frac{3}{2}d \|\bD\|_{2, \infty}^{2}\rho^{3}
    \end{equation*} 
\end{proof}
Now that we know how to control the remainder with the error term $\bD$, we will see how $\bD$ can be controlled using the noise level $\bW$ and the tuning parameter $\lambda_{n}$. 

\subsubsection{The Fixed Point Argument}
The proof uses the Banach fixed-point theorem instead of Brouwer's fixed-point theorem as in \cite{ravikumar2011} to make up for the lack of compactness in $\cL_{2}$.
\begin{lemma}[Control of $\bD$]\label{lem:control_of_D}
    Let $r_{0} = \min\{1/3\rho^{}d, 1/6\rho^{3}\gamma^{}d\}$. If $r = 2\gamma^{}(\|\bW\|_{2, \infty} + \lambda_{n}) \leq r_{0}$ then $\|\bD\|_{2, \infty} = \| \tilde{\bH} - \bH^{\ast} \|_{2, \infty}\leq r$. 
\end{lemma}
\begin{proof}
    Notice that $\tilde{\bH}_{S^{c}}^{} = \bH^{\ast}_{S} = \bzero$, so it suffices to bound $\bD_{S} = \tilde{\bH}_{S}^{} - \bH^{\ast}_{S}$. Notice further that $\tilde{\bH}_{S}^{}$ is the unique solution to the stationary condition (\ref{eqn:opt_problem_Sc}) and hence by subsetting on $S$ we can write
    that $\tilde{\bH}_{S}^{}$ is the unique solution to
    \begin{equation*}
        \cG(\bH_{S}^{}) = -[(\bI + \bH)^{-1}]_{S} + \hat{\bR}_{S} + \lambda_{n}\tilde{\bZ}_{S}
    \end{equation*}
    since we can assume $\bH_{S^{c}} = \bzero$. Define $\bbB_{r} = \{\bA \in \cL_{2} : \|\bA_{S}\|_{2, \infty} \leq r \mbox{ and } \bA_{S^{c}} = \bzero\}$. It follows that showing that $\|\bD_{S}\|_{2, \infty} \leq r$ is equivalent to showing that the equation $\cG(\bH^{\ast}_{S} + \bD_{S}^{}) = \bzero$ admits a (unique) solution $\bD_{S} \in \bbB_{r}$. 
   
    Note that a fixed point of the function $\cF: \cL_{2} \to \cL_{2}$ given in the vectorized form by 
    \begin{equation}
        \bar{\cF}(\overline{\bD}_{S}) = - \Gamma_{SS}^{-1}\left[\bar{\cG}(\bH^{\ast}_{S} + \bD_{S}^{})\right] + \overline{\bD}_{S}.
    \end{equation}
    in $\bbB_{r}$ corresponds to a solution of $\cG(\bH^{\ast}_{S} + \bD_{S^{}}) = \bzero$. It suffices for us to show that $\cF$ admits a unique fixed point in $\bbB_{r}$. We begin by simplifying the expression for $\cF$. By definition,
    \begin{eqnarray*}
        &\cG(\bH^{\ast}_{S} + \bD_{S}^{}) 
        &= -\left[(\bI + \bH^{\ast} + \bD)^{-1}\right]_{S} + \hat{\bR}_{S} + \lambda_{n}\tilde{\bZ}_{S} \\
        &&= -\left[(\bI + \bH^{\ast} + \bD)^{-1}\right]_{S} + \left[(\bI + \bH^{\ast})^{-1}\right]_{S} + \hat{\bR}_{S} - \left[(\bI + \bH^{\ast})^{-1}\right]_{S} + \lambda_{n}\tilde{\bZ}_{S} \\
        &&= -\left[(\bI + \bH^{\ast} + \bD)^{-1}\right]_{S} + \left[(\bI + \bH^{\ast})^{-1}\right]_{S} + \bW_{S} + \lambda_{n}\tilde{\bZ}_{S}.
    \end{eqnarray*}
    By Lemma \ref{lem:control_of_remainder},
    \begin{equation*}
        (\bI + \bH^{\ast} + \bD)^{-1} - (\bI + \bH^{\ast})^{-1} + (\bI + \bH^{\ast})^{-1}\bD(\bI + \bH^{\ast})^{-1} = \left[(\bI + \bH^{\ast})^{-1}\bD\right]^{2}\bJ(\bI + \bH^{\ast})^{-1}.
    \end{equation*}
    and this can be vectorized as
    \begin{equation*}
        \overline{(\bI + \bH^{\ast} + \bD)^{-1} - (\bI + \bH^{\ast})^{-1}} + \Gamma\overline{\bD} = \overline{\left[(\bI + \bH^{\ast})^{-1}\bD\right]^{2}\bJ(\bI + \bH^{\ast})^{-1}}.
    \end{equation*}
    By subsetting the above equation on $S$ we can rewrite $\cF$ as
    \begin{eqnarray*}
        &\bar{\cF}(\overline{\bD}_{S}) 
        &= - \Gamma_{SS}^{-1}\left[\bar{\cG}(\bH^{\ast}_{S} + \bD_{S})\right] + \overline{\bD}_{S}\\
        &&= \phantom{-}\Gamma_{SS}^{-1}\left[\overline{(\bI + \bH^{\ast} + \bD)^{-1} - (\bI + \bH^{\ast})^{-1}}\right]_{S} 
        -\Gamma_{SS}^{-1}\left(\overline{\bW_{S} + \lambda_{n}\tilde{\bZ}_{S}}\right) + \overline{\bD}_{S}\\
        &&= \phantom{-}\Gamma_{SS}^{-1}\left[\overline{\left[(\bI + \bH^{\ast})^{-1}\bD\right]^{2}\bJ(\bI + \bH^{\ast})^{-1}}\right]_{S}
        -\Gamma_{SS}^{-1}\left(\overline{\bW}_{S} + \lambda_{n}\overline{\tilde{\bZ}}_{S}\right)
    \end{eqnarray*}        
    Using essentially the same technique as in Lemma \ref{lem:control_of_remainder}, we can show that for $\bA, \bB \in \bbB_{r}$, we can write
    \begin{eqnarray*}
        &\| \cF(\bA) - \cF(\bB) \|_{2, \infty}        
        &\leq \gamma \|\left[(\bI + \bH^{\ast})^{-1}\bA\right]^{2}\bJ(\bI + \bH^{\ast})^{-1}
        - \left[(\bI + \bH^{\ast})^{-1}\bB\right]^{2}\bJ(\bI + \bH^{\ast})^{-1}\|_{2, \infty}\\
        &&\leq \gamma\|(\bI + \bH^{\ast})^{-1}(\bA - \bB)(\bI + \bH^{\ast})^{-1}\bA\bJ(\bI + \bH^{\ast})^{-1}\|_{2, \infty}\\ 
        &&+~ \gamma\|(\bI + \bH^{\ast})^{-1}\bB(\bI + \bH^{\ast})^{-1}(\bA - \bB)\bJ(\bI + \bH^{\ast})^{-1}\|_{2, \infty}\\
        &&\leq \frac{3}{2}d \rho^{3}\gamma^{}\left(\|\bA\|_{2, \infty} + \|\bB\|_{2, \infty}\right)\|\bA - \bB\|_{2, \infty}.
    \end{eqnarray*}
    Therefore, $\| \cF(\bA) - \cF(\bB) \|_{2, \infty} \leq \frac{3}{2}d \rho^{3}\gamma^{}\left(\|\bA\|_{2, \infty} + \|\bB\|_{2, \infty}\right)\|\bA - \bB\|_{2, \infty}$. It follows from $\| \bA \|_{2, \infty}, \| \bB \|_{2, \infty} \leq r$, that $\|\cF(\bA) - \cF(\bB)\|_{2, \infty} \leq \frac{1}{2} \| \bA - \bB \|_{2, \infty}$. Thus $\cF$ is contractive in $\bbB_{r}$. Moreover, $\cF$ maps $\bbB_{r}$ into itself since
    \begin{eqnarray*}
        &\| \cF(\bA) \|_{2, \infty} 
        &\leq \frac{1}{2} \|\bA\|_{2, \infty} + \| \cF(\bzero) \|_{2, \infty} \\
        &&\leq \frac{1}{2} \|\bA\|_{2, \infty} + \gamma(\|\bW\|_{2, \infty} + \lambda_{n}) \\
        &&\leq \frac{1}{2} r + \frac{1}{2} r = r.
    \end{eqnarray*}
    By Banach's fixed-point theorem, it follows that $\cF$ has a unique fixed point in $\bbB_{r}$.
\end{proof}

\subsection{Estimation of the Correlation Operator} 
Theorem \ref{thm:tail-correlation} is a consequences of a slightly stronger result, in the form of Lemma \ref{thm:correlation-estimation}.
\begin{proof}[Proof of Theorem \ref{thm:tail-correlation}]
    Apply Lemma \ref{thm:correlation-estimation} to $X' = (\bzero, \dots, \bzero, X_{i},\bzero, \dots,\bzero, X_{j}, \bzero, \dots, \bzero)$ and make the necessary simplifications for $\kappa$. 
\end{proof}

Lemma \ref{thm:correlation-estimation} is a variant of Theorem 6.2 from \cite{waghmare2023}, and can be proved similarly with the only important change being the use of the Hilbert-Schmidt norm for bounding the involved quantities. The difference is that the lemma concerns concentration in the Hilbert-Schmidt norm while Theorem 6.2 concerns the same in operator norm. Because the proof is somewhat complicated, we include an outline of the proof along with the most tedious calculations here. We then give the statement of the lemma, and break down its proof into further lemmas.\\

Recall that the empirical correlation matrix $\hat{\bR}$ is given by $\hat{\bR} = \bI + [\epsilon_{n}\bI + \dg \hat{\bC}]^{-1/2}\hat{\bC}_{0}[\epsilon_{n}\bI + \dg \bC]^{-1/2}$ and we are interested in quantifying how well it estimates the correlation operator matrix $\bR$. To this end, we define $\bR_{e}$ as
\begin{equation*}
	\bR_{e} = \bI + [\epsilon\bI + \dg \bC]^{-1/2}\bC_{0}[\epsilon\bI + \dg \bC]^{-1/2}
\end{equation*}
where it is implicit that $\epsilon \equiv \epsilon_{n}$. The operator matrix $\bR_{e}$ is essentially $\hat{\bR}$ assuming that $\bC$ is known and hence it can be thought of as an oracle estimator of $\bR$. The error of estimating $\bR$ with $\hat{\bR}$ can now be bounded above by estimation and approximation terms as follows:
\begin{equation*}
	\|\hat{\bR} - \bR\|_{2} \leq \|\hat{\bR} - \bR_{e}\|_{2} + \|\bR_{e} - \bR\|_{2}.
\end{equation*}
We will now bound the terms on the right hand side in terms of the regularization parameter $\epsilon$, the error $\| \hat{\bC} - \bC \|_{2}$ and a few other quantities which depend only on $\bC$. Finally, we will choose $\epsilon$ so as to minimize the bound and this will give us a rate of convergence in the form of the following result.
      
\begin{lemma}\label{thm:correlation-estimation}
    If $\dg(\hat{\bC})$ is positive, then
    \begin{equation*}
        \|\hat{\bR} - \bR\|_{2} \leq 2(1 + \| \bR_{0} \|_{2}) \left[\frac{\|\hat{\bC} - \bC\|_{2}^{2}}{\epsilon_{n}^{2}} + \frac{\|\hat{\bC} - \bC\|_{2}}{\epsilon_{n}} \right] + 2\epsilon_{n}^{\beta \wedge 1} \cdot \|\Phi\|_{2} \cdot \| \dg \bC \|^{2\beta - \beta \wedge 1}.
    \end{equation*}
    In particular, if $\| \hat{\bC} - \bC \|_{2} \leq \delta$ then for $\epsilon_{n} = \delta^{\frac{1}{1 + \beta \wedge 1}}$ we have $\|\hat{\bR} - \bR\|_{2} \leq \kappa \delta^{\frac{\beta \wedge 1}{1 + \beta \wedge 1}}$
    where $$\kappa = 8 \left[(1 + \| \bR_{0} \|_{2}) \vee \|\Phi\|_{2}\|\dg \bC\|^{2\beta - \beta \wedge 1}\right].$$ 
\end{lemma}
\begin{proof}
    The proof follows from Lemma \ref{lem:estmn_error} and \ref{lem:approx_error}.
\end{proof}
In the following discussion, Lemmas \ref{thm:correlation-estimation}, \ref{lem:estmn_error}, \ref{lem:approx_terms}, \ref{lem:approx_error}, \ref{lem:diff_sqrts_beta_ineq} are merely Hilbert-Schmidt counterparts of almost identical results (with almost identical proofs) in \cite{waghmare2023}, while Lemma \ref{lem:lambda_epsilon_ineq} is exactly identical and is included here for convenience. We begin by treating the estimation error.
\begin{lemma}\label{lem:estmn_error}
    If $\dg \hat{\bC}$ is positive then
	\begin{equation*}
		\| \hat{\bR} - \bR_{e} \|_{2} \leq 2(1 + \| \bR_{0} \|_{2}) \left[\frac{\|\hat{\bC} - \bC\|_{2}^{2}}{\epsilon^{2}} + \frac{\|\hat{\bC} - \bC\|_{2}}{\epsilon} \right]
	\end{equation*}
\end{lemma}
\begin{proof} It can be verified with some simple algebraic manipulation that
	\begin{eqnarray*}
		&\hat{\bR} - \bR_{e}
		&= \left[ [\epsilon\bI + \dg \hat{\bC}]^{-1/2} - [\epsilon\bI + \dg \bC]^{-1/2} \right][\hat{\bC}_{0} - \bC_{0}][\epsilon\bI + \dg \hat{\bC}]^{-1/2} \\
		& &+\quad \left[ [\epsilon\bI + \dg \hat{\bC}]^{-1/2} - [\epsilon\bI + \dg \bC]^{-1/2} \right]\bC_{0}\left[[\epsilon\bI + \dg \hat{\bC}]^{-1/2} - [\epsilon\bI + \dg \bC]^{-1/2}\right] \\
		&&+\quad \left[ [\epsilon\bI + \dg \hat{\bC}]^{-1/2} - [\epsilon\bI + \dg \bC]^{-1/2} \right]\bC_{0}[\epsilon\bI + \dg \bC]^{-1/2} \\
		&&+\quad [\epsilon\bI + \dg \bC]^{-1/2}[\hat{\bC}_{0} - \bC_{0}][\epsilon\bI + \dg \hat{\bC}]^{-1/2} \\
		&&+\quad [\epsilon\bI + \dg \bC]^{-1/2}\bC_{0}\left[[\epsilon\bI + \dg \hat{\bC}]^{-1/2} - [\epsilon\bI + \dg \bC]^{-1/2}\right].
	\end{eqnarray*}
	Using $\bC_{0} = [\dg \bC]^{1/2}\bR_{0}[\dg \bC]^{1/2}$, we can rewrite this expansion in terms of $\bD$ and $\bA$ given by
    \begin{eqnarray*}
		&\bD &= [\epsilon\bI + \dg \hat{\bC}]^{-1/2} - [\epsilon\bI + \dg \bC]^{-1/2} \\
		&\bA &= \left[ [\epsilon\bI + \dg \hat{\bC}]^{-1/2} - [\epsilon\bI + \dg \bC]^{-1/2} \right][\dg \bC]^{1/2}
	\end{eqnarray*}
    as
	\begin{eqnarray*}
		&&= \bD[\hat{\bC}_{0} - \bC_{0}][\epsilon\bI + \dg \hat{\bC}]^{-1/2} + \bA\bR_{0}\bA^{\ast} + \bA\bR_{0}[\dg \bC]^{1/2}[\epsilon\bI + \dg \bC]^{-1/2} \\
		&&+\quad [\epsilon\bI + \dg \bC]^{-1/2}[\hat{\bC}_{0} - \bC_{0}][\epsilon\bI + \dg \hat{\bC}]^{-1/2} + [\epsilon\bI + \dg \bC]^{-1/2}[\dg \bC]^{1/2}\bR_{0} \bA^{\ast}.
	\end{eqnarray*}
    Thus,
	\begin{eqnarray*}
		&\| \hat{\bR} - \bR_{e} \|_{2} 
		&\leq \| \bD \|\cdot \|\hat{\bC} - \bC\|_{2} \cdot \frac{1}{\sqrt{\epsilon}} + \| \bA \| \cdot \| \bR_{0} \|_{2} \cdot \| \bA \| + \| \bA\| \cdot \| \bR_{0}\|_{2} \cdot 1 \\
		&&\qquad+\quad \frac{1}{\sqrt{\epsilon}} \cdot \| \hat{\bC} - \bC \|_{2} \cdot \frac{1}{\sqrt{\epsilon}} + 1 \cdot \| \bR_{0} \|_{2} \cdot \| \bA \|.
	\end{eqnarray*}
	Using Lemma \ref{lem:approx_terms} (immediately below) for $\hat{\bA} = \dg \hat{\bC}$ and $\bA = \dg \bC$, and using $\|\dg \hat{\bC} - \dg \bC\|_{2} \leq \|\hat{\bC} - \bC\|_{2}$, we derive $\| \bD \| \leq \|\hat{\bC} - \bC\|_{2}/\epsilon^{3/2} $ and $\| \bA \| \leq \|\hat{\bC} - \bC\|_{2}/\epsilon$. It follows that
	\begin{eqnarray*}
		&\| \hat{\bR} - \bR_{e} \|_{2} 
		&\leq  \frac{\|\hat{\bC} - \bC\|_{2}^{2}}{\epsilon^{2}} + \| \bR_{0} \|_{2} \frac{\|\hat{\bC} - \bC\|_{2}^{2}}{\epsilon^{2}} + \| \bR_{0}\|_{2}\frac{\|\hat{\bC} - \bC\|_{2}}{\epsilon} \\
		&&\qquad+\quad \frac{\|\hat{\bC} - \bC\|_{2}}{\epsilon} + \| \bR_{0} \|_{2}\frac{\|\hat{\bC} - \bC\|_{2}}{\epsilon} \\
		&&= (1 + \| \bR_{0} \|_{2})\frac{\|\hat{\bC} - \bC\|_{2}^{2}}{\epsilon^{2}} + (1 + 2\| \bR_{0} \|_{2})\frac{\|\hat{\bC} - \bC\|_{2}}{\epsilon} \\
		&&\leq 2(1 + \| \bR_{0} \|_{2}) \left[\frac{\|\hat{\bC} - \bC\|_{2}^{2}}{\epsilon^{2}} + \frac{\|\hat{\bC} - \bC\|_{2}}{\epsilon} \right].
	\end{eqnarray*}
\end{proof}

\begin{lemma}\label{lem:approx_terms}
	If $\hat{\bA}$ is positive, then
	\begin{eqnarray*}
		&\| [\epsilon\bI + \hat{\bA}]^{-1/2} - [\epsilon\bI + \bA]^{-1/2} \| 
		&\leq \| \hat{\bA} - \bA \|_{2}/\epsilon^{3/2}, \\
		&\Big\| \left[[\epsilon\bI + \hat{\bA}]^{-1/2} - [\epsilon\bI + \bA]^{-1/2}\right]\bA^{1/2} \Big\| 
		&\leq \| \hat{\bA} - \bA \|_{2}/\epsilon.
	\end{eqnarray*}
\end{lemma}
\begin{proof}
	Notice that
	\begin{eqnarray*}
		&&[\epsilon\bI + \hat{\bA}]^{-1/2} - [\epsilon\bI + \bA]^{-1/2} \\
		&&= [\epsilon\bI + \hat{\bA}]^{-1/2}\left[ [\epsilon\bI + \hat{\bA}]^{1/2} - [\epsilon\bI + \bA]^{-1/2} \right][\epsilon\bI + \bA]^{1/2} \\
		&&= [\epsilon\bI + \hat{\bA}]^{-1/2}\left[ [\epsilon\bI + \hat{\bA}]^{1/2} + [\epsilon\bI + \bA]^{1/2} \right]^{-1}\left[ [\epsilon\bI + \hat{\bA}] - [\epsilon\bI + \bA] \right][\epsilon\bI + \bA]^{-1/2} \\
		&&= \left[ \epsilon\bI + \hat{\bA} + [\epsilon\bI + \bA]^{1/2}[\epsilon\bI + \hat{\bA}]^{1/2} \right]^{-1} [\hat{\bA} - \bA] [\epsilon\bI + \bA]^{-1/2}.
	\end{eqnarray*}
	Since $\hat{\bA} + [\epsilon\bI + \bA]^{1/2}[\epsilon\bI + \hat{\bA}]^{1/2}$ is positive, we can write
	\begin{eqnarray*}
		&& \| [\epsilon\bI + \hat{\bA}]^{-1/2} - [\epsilon\bI + \bA]^{-1/2} \| \\
		&&\leq \Big\| \left[ \epsilon\bI + \hat{\bA} + [\epsilon\bI + \bA]^{1/2}[\epsilon\bI + \hat{\bA}]^{1/2} \right]^{-1} \Big\| \|\hat{\bA} - \bA\| \|[\epsilon\bI + \bA]^{-1/2}\|\\
		&&\leq \frac{1}{\epsilon} \cdot \|\hat{\bA} - \bA\|_{2} \cdot \frac{1}{\epsilon^{1/2}}
	\end{eqnarray*}
	and similarly,
	\begin{eqnarray*}
		&&\Big\| \left[[\epsilon\bI + \hat{\bA}]^{-1/2} - [\epsilon\bI + \bA]^{-1/2}\right]\bA^{1/2} \Big\| \\
		&&\leq \Big\|\left[ \epsilon\bI + \hat{\bA} + [\epsilon\bI + \bA]^{1/2}[\epsilon\bI + \hat{\bA}]^{1/2} \right]^{-1} \Big\| \|\hat{\bA} - \bA\| \|[\epsilon\bI + \bA]^{-1/2}\bA^{1/2}\| \\
		&&\leq \frac{1}{\epsilon} \cdot \|\hat{\bA} - \bA\|_{2} \cdot 1.
	\end{eqnarray*}

\end{proof}

Now, we will find an upper bound for the approximation error under a regularity condition. 
\begin{lemma}\label{lem:approx_error}
	If $\bR_{0} = [\dg \bC]^{\beta}\Phi[\dg \bC]^{\beta}$ for some $\Phi \in \cS_{2}$ and $\beta > 0$, then
	\begin{equation*}
		\| \bR_{e} - \bR \|_{2} \leq 
		\begin{cases}
			2\epsilon^{\beta} \cdot \|\Phi\|_{2} \cdot \| \dg \bC \|^{\beta} & 0 < \beta \leq 1 \\
			2\epsilon^{\phantom{\beta}} \cdot \|\Phi\|_{2} \cdot \| \dg \bC \|^{2\beta-1} & 1 < \beta < \infty.
		\end{cases}
	\end{equation*}
\end{lemma}
\begin{proof}
	We decompose the difference as follows:
	\begin{eqnarray*}
		&\bR - \bR_{e} 
		&= [\epsilon \bI + \dg \bC]^{-1/2}[\dg \bC]^{1/2}\bR[\dg \bC]^{1/2}[\epsilon \bI + \dg \bC]^{-1/2} - \bR \\
		& &= \Big[ [\epsilon \bI + \dg \bC]^{-1/2} - [\dg \bC]^{-1/2} \Big][\dg \bC]^{1/2}\bR[\dg \bC]^{1/2}[\epsilon \bI + \dg \bC]^{-1/2} \\
		& &\qquad+\quad \bR[\dg \bC]^{1/2}\Big[ [\epsilon \bI + \dg \bC]^{-1/2} - [\dg \bC]^{-1/2} \Big] \\
		& &= \Big[ [\epsilon \bI + \dg \bC]^{-1/2} - [\dg \bC]^{-1/2} \Big][\dg \bC]^{1/2+\beta}\Phi[\dg \bC]^{1/2+\beta}[\epsilon \bI + \dg \bC]^{-1/2} \\
		& &\qquad+\quad [\dg \bC]^{\beta}\Phi[\dg \bC]^{1/2+\beta}\Big[ [\epsilon \bI + \dg \bC]^{-1/2} - [\dg \bC]^{-1/2} \Big].
	\end{eqnarray*}
	Using $\| [\dg \bC]^{1/2+\beta}[\epsilon \bI + \dg \bC]^{-1/2} \| \leq \| \bC \|^{\beta} $, it follows that
	\begin{eqnarray*}
		&\| \bR - \bR_{e} \|_{2} &\leq \Big\| \Big[ [\epsilon \bI + \dg \bC]^{-1/2} - [\dg \bC]^{-1/2} \Big][\dg \bC]^{1/2+\beta} \Big\| \|\Phi \|_{2} \| \dg \bC \|^{\beta} \\
		& &\qquad+\quad \| \dg \bC \|^{\beta} \| \Phi\|_{2} \Big\| [\dg \bC]^{1/2+\beta} \Big[ [\epsilon \bI + \dg \bC]^{-1/2} - [\dg \bC]^{-1/2} \Big] \Big\|. 
	\end{eqnarray*}
	The conclusion is now a direct consequence of Lemma \ref{lem:diff_sqrts_beta_ineq}, stated immediately below.
\end{proof}

\begin{lemma} \label{lem:diff_sqrts_beta_ineq}
	We have 
	\begin{equation*}
		\Big\| \Big[ [\epsilon \bI + \dg \bC]^{-1/2} - [\dg \bC]^{-1/2} \Big][\dg \bC]^{1/2+\beta} \Big\| \leq \begin{cases}
			\epsilon^{\beta} & 0 < \beta \leq 1 \\
			\epsilon \cdot \| \dg \bC \|^{\beta-1} & 1 < \beta < \infty.
		\end{cases}
	\end{equation*}
\end{lemma}
\begin{proof}
	By the spectral mapping theorem,
	\begin{eqnarray*}
		\Big\| \Big[ [\epsilon \bI + \dg \bC]^{-1/2} - [\dg \bC]^{-1/2} \Big][\dg \bC]^{1/2+\beta} \Big\| \leq \sup_{0 \leq \lambda \leq  \| \dg \bC \|} \left\lbrace \left| \frac{1}{\sqrt{\epsilon + \lambda}} - \frac{1}{\sqrt{\lambda}} \right| \cdot \lambda^{1/2+\beta} \right\rbrace.
	\end{eqnarray*}
	It can be shown using some elementary calculations that
	\begin{eqnarray*}
		\left| \frac{1}{\sqrt{\epsilon + \lambda}} - \frac{1}{\sqrt{\lambda}} \right| \cdot \lambda^{1/2+\beta}
		= \frac{\epsilon\lambda^{\beta}}{\sqrt{\epsilon + \lambda}(\sqrt{\lambda} + \sqrt{\epsilon+\lambda})} 
		\leq \begin{cases}
			\epsilon\left[\frac{\lambda^{\beta}}{\epsilon + \lambda}\right] & 0 < \beta < 1/2 \\
			\epsilon\left[\frac{\lambda^{2\beta-1}}{\epsilon + \lambda}\right]^{1/2} & 1/2 \leq \beta < 1 \\
			\epsilon \lambda^{\beta-1} & 1 \leq \beta < \infty.
		\end{cases}
	\end{eqnarray*}
	The conclusion now follows from Lemma \ref{lem:lambda_epsilon_ineq}, stated immediately below.
\end{proof}
\begin{lemma} \label{lem:lambda_epsilon_ineq}
	For $0 < x < 1$ and $\lambda \geq 0$, we have
	\begin{equation*}
		\frac{\lambda^{x}}{\epsilon + \lambda} \leq \frac{\epsilon^{x-1}}{2}
	\end{equation*}
\end{lemma}
\begin{proof}
	Consider the reciprocal expression. It follows from elementary differential calculus that the minimum of the reciprocal occurs at $\lambda_{\ast} = x\epsilon/(1-x)$. Therefore,
	\begin{equation*}
		\frac{\epsilon}{\lambda^{x}} + \lambda^{1-x} \geq \frac{\epsilon}{\lambda_{\ast}^{x}} + \lambda_{\ast}^{1-x} = \frac{\epsilon^{1-x}}{x^{x}(1-x)^{1-x}} \geq \frac{\epsilon^{1-x}}{\max_{0 < x < 1} [x^{x}(1-x)^{1-x}]} = 2\epsilon^{1-x}.
	\end{equation*}
\end{proof}

\subsection{Concentration of Sub-Gaussian Random Elements}

\begin{lemma}%
    Let $X$ be a random element in a Hilbert space $\cH$ such that $\bbE[X] = \bzero$ and $\| X \|$ is sub-Gaussian (in the sense defined in the main paper). Then for $\bC_{ij} = \bbE[X_{i} \otimes X_{j}] - \bbE[X_{i}] \otimes \bbE[X_{j}]$ and $\hat{\bC}_{ij} = \frac{1}{n}\sum_{k=1}^{n} X_{i}^{k} \otimes X_{j}^{k} - \bar{X}_{i} \otimes \bar{X}_{j}$ we have
    \begin{equation*}
        \bbP\{ \| \hat{\bC}_{ij} - \bC_{ij} \|_{2} \geq t \} \leq 2 \exp \left[ -\frac{cnt^{2}}{\| \| X_{i} \| \|_{\psi_{2}}^{2}\| \| X_{j} \| \|_{\psi_{2}}^{2}} \right]
    \end{equation*}
    for $0 \leq t \leq \| \| X_{i} \otimes X_{j} - \bbE\left[ X_{i} \otimes X_{j} \right] - \bar{X}_{i} \otimes \bar{X}_{j} + \bbE[X_{i}] \otimes \bbE[X_{j}] \| \|_{\psi_{1}}$ where $c$ is a universal constant. In particular, 
    \begin{equation*}
        \bbP\{ \| \hat{\bC}_{ij} - \bC_{ij} \|_{2} \geq t \} \leq 2 \exp \left[ -\frac{cnt^{2}}{\|X\|_{\infty}^{4}} \right]
    \end{equation*}
    for $0 \leq t \leq t_{X}$, where $\|X\|_{\infty} = \max_{i} \| \| X_{i} \| \|_{\psi_{2}}$ and $t_{X} = \min_{ij} \|\| X_{i} \otimes X_{j} - \bbE\left[ X_{i} \otimes X_{j} \right] - \bar{X}_{i} \otimes \bar{X}_{j} + \bbE[X_{i}] \otimes \bbE[X_{j}] \|\|_{\psi_{1}}$.
\end{lemma}
\begin{proof}
	Let $Y_{k} = X_{i}^{k} \otimes X_{j}^{k} - \bbE\left[ X_{i} \otimes X_{j} \right] - \bar{X}_{i} \otimes \bar{X}_{j} + \bbE[X_{i}] \otimes \bbE[X_{j}]$ for $1 \leq k \leq n$. Then, the $Y_{k}$ are sub-exponential random elements in the space of Hilbert-Schmidt operators on $\cH$. Indeed,
    \begin{eqnarray*}
		&\| \| Y \| \|_{\psi_{1}} &= \| \| X_{i} \otimes X_{j} - \bbE\left[ X_{i} \otimes X_{j} \right] - \bar{X}_{i} \otimes \bar{X}_{j} + \bbE[X_{i}] \otimes \bbE[X_{j}] \| \|_{\psi_{1}} \\
        & &\leq \| \| X_{i} \otimes X_{j} \| + \| \bbE\left[ X_{i} \otimes X_{j} \right] \| + \|\bar{X}_{i} \otimes \bar{X}_{j}\| + \|\bbE[X_{i}] \otimes \bbE[X_{j}]\|\|_{\psi_{1}} \\
		& &\leq \| \| X_{i} \otimes X_{j} \| \|_{\psi_{1}} + \| \| \bbE\left[ X_{i} \otimes X_{j} \right] \| \|_{\psi_{1}} + \|\|\bar{X}_{i} \|\|\bar{X}_{j}\|\|_{\psi_{1}} + \|\|\bbE[X_{i}]\| \|\bbE[X_{j}]\|\|_{\psi_{1}} \\
        & &\leq \| \| X_{i}\| \|X_{j} \| \|_{\psi_{1}} + \| \bbE\left[\| X_{i} \otimes X_{j} \|\right] \|_{\psi_{1}} + \|\|X_{i}\| \|X_{j}\|\|_{\psi_{1}} + \|\bbE[\|X_{i}\|] \bbE[\|X_{j}\|]\|_{\psi_{1}} \\
        & &\leq \| \| X_{i}\| \|X_{j} \| \|_{\psi_{1}} + \| \bbE\left[\| X_{i} \otimes X_{j} \|\right] \|_{\psi_{1}} + \|\|X_{i}\| \|X_{j}\|\|_{\psi_{1}} + \|\bbE[\|X_{i}\|]\|_{\psi_{2}} \|\bbE[\|X_{j}\|]\|_{\psi_{2}}.
	\end{eqnarray*}
    Using the centering property of the sub-exponential and sub-Gaussian norms and the fact that product of two sub-Gaussian random variables is sub-exponential (see Exercise 2.7.10 and Lemma 2.7.7 of \cite{vershynin2018}), we get
    \begin{eqnarray*}
        \| \| Y \| \|_{\psi_{1}} 
        &\lesssim& \| \| X_{i}\| \|X_{j} \| \|_{\psi_{1}} + \|\| X_{i} \otimes X_{j} \|\|_{\psi_{1}} + \|\|X_{i}\| \|X_{j}\|\|_{\psi_{1}} + \|\|X_{i}\|\|_{\psi_{2}} \|\|X_{j}\|\|_{\psi_{2}} \\
        &=& \| \| X_{i}\| \|X_{j} \| \|_{\psi_{1}} + \|\| X_{i}\| \|X_{j} \|\|_{\psi_{1}} + \|\|X_{i}\| \|X_{j}\|\|_{\psi_{1}} + \|\|X_{i}\|\|_{\psi_{2}} \|\|X_{j}\|\|_{\psi_{2}} \\
        &\lesssim& \|\|X_{i}\|\|_{\psi_{2}} \|\|X_{j}\|\|_{\psi_{2}} + \|\|X_{i}\|\|_{\psi_{2}} \|\|X_{j}\|\|_{\psi_{2}} + \|\|X_{i}\|\|_{\psi_{2}} \|\|X_{j}\|\|_{\psi_{2}} + \|\|X_{i}\|\|_{\psi_{2}} \|\|X_{j}\|\|_{\psi_{2}} \\
        &\lesssim& \|\|X_{i}\|\|_{\psi_{2}} \|\|X_{j}\|\|_{\psi_{2}}.
    \end{eqnarray*}    
    By Theorem 2.8.1 (Bernstein's Inequality) of \cite{vershynin2018}, we have
	\begin{eqnarray*}
		&\bbP \left\lbrace \left\| \frac{1}{n}\sum_{k=1}^{n} Y_{k} \right\| \geq t \right\rbrace 
		&= \textstyle \bbP \left\lbrace \left\| \sum_{k=1}^{n} Y_{k} \right\| \geq nt \right\rbrace \\
		& &\leq \textstyle \bbP \left\lbrace \sum_{k=1}^{n} \| Y_{k} \| \geq nt \right\rbrace \\
		& &\leq 2 \exp \left[ - c\left( \frac{n^{2}t^{2}}{n \| \| Y \| \|_{\psi_{1}}^{2}} \wedge \frac{nt}{ \| \| Y \| \|_{\psi_{1}}} \right) \right] \\
		& &\leq 2 \exp \left[ - \frac{cnt^{2}}{\| \| Y \| \|_{\psi_{1}}^{2}} \right]
	\end{eqnarray*}
	for $0 < t < \|\| Y \|\|_{\psi_{1}}$ where $c$ is an absolute constant. It follows that
	\begin{equation*}
		{\textstyle \bbP \left\lbrace \| \hat{\bC}_{ij} - \bC_{ij} \|_{2} \geq t \right\rbrace} \leq 2 \exp \left[ - \frac{cnt^{2}}{\| \| X_{i} \| \|_{\psi_{2}}^{2}\| \| X_{j} \| \|_{\psi_{2}}^{2}} \right]
	\end{equation*}
	for $0 < t < \| \| X_{i} \otimes X_{j} - \bbE\left[ X_{i} \otimes X_{j} \right] - \bar{X}_{i} \otimes \bar{X}_{j} + \bbE[X_{i}] \otimes \bbE[X_{j}] \| \|_{\psi_{1}}$ for some absolute constant $c > 0$.
\end{proof}
\begin{remark}
If $X$ has mean zero and we take $\hat{\bC} = \frac{1}{n}\sum_{k=1}^{n} X^{k} \otimes X^{k}$, then the above result is still true for $t_{X} = \min_{ij} \|\| X_{i} \otimes X_{j} - \bbE\left[ X_{i} \otimes X_{j} \right]\|\|_{\psi_{1}}$. Regardless, even if the mean is not zero, $\bar{X}_{i} \approx \bbE[X_{i}]$ for large enough $n$ and therefore, $$t_{X} \approx \min_{ij} \|\| X_{i} \otimes X_{j} - \bbE\left[ X_{i} \otimes X_{j} \right]\|\|_{\psi_{1}}.$$
\end{remark}

\subsection{Setups for the Simulation Study}

Below we describe how we perform a single draw of a multivariate functional datum in the three setups considered in our simulation study. In all cases, the sample size is chosen as $n=100$, so this process is repeated independently one hundred times. From the resulting sample, the empirical covariance operator is calculated, which is subsequently transformed into the correlation operator using generalized inverses. All the simulations were run on a computer cluster with the total runtime of about six hundred CPU hours.

\subsubsection{Setup 1}

This is the only setup where a precision matrix $\mathbf Q \in \R^{p r\times pr}$ is explicitly formed, and then inverted to obtain $\mathbf C \in \R^{p r\times p r}$. The scores $\delta = (\delta_1,\ldots,\delta_p)^\top \in \R^{pr}$ for the first $r=10$ Fourier eigenfunctions in all $p=100$ nodes are then drawn jointly from the multivariate Gaussian distribution with mean 0 and covariance $\mathbf C$. Finally, the multivariate functional datum is formed as 
\[
X=(X_1,\ldots,X_p)^\top = \left(\; \sum_{l=1}^r \delta_{1,l} e_l(t) \;,\; \ldots \;,\; \sum_{l=1}^r \delta_{p,l} e_l(t) \;\right)^\top
\]
where $e_l(t)$ is the $l$-th Fourier basis function. In practice, these are naturally evaluated on an equidistant grid (we use the grid size $K=30$ throughout this simulation study).

Note that the conditional dependencies are thus directly governed by the block sparsity pattern of the finite-dimensional precision matrix $\mathbf Q \in \R^{p r\times pr}$. This simulation setup is thus not truly functional.

It remains to specify the choice of $\mathbf Q = (\mathbf Q_{i,j})_{j,j=1}^p$, where $\mathbf Q_{i,j} \in \R^{r \times r}$ for $i,j=1,\ldots,p$ are the respective blocks. For $i=1,\ldots,p$, we take 
\[
\begin{split}
    \mathbf Q_{i,i} &= \diag(1,\ldots,10)/10, \\
    \mathbf Q_{i,i-1} &= 0.4 \diag(0,\ldots, 0, 6, \ldots, 10)/10, \\
    \mathbf Q_{i,i-2} &= 0.2 \diag( 0,\ldots, 0, 6, \ldots, 10)/10,
\end{split}
\]
and $\mathbf Q_{i,i} = 0$ for $|i-j| > 2$. This choice constitutes an AR(2) process. Note that due to the zeros at the end of the diagonals of $\mathbf Q_{i,i-1}$ and $\mathbf Q_{i,i-2}$ the dependencies are only created between eigenspaces spanned by $e_6(t),\ldots,e_{10}(t)$. Those actually have lower corresponding eigenvalues than $e_1(t),\ldots,e_{5}(t)$, since $\mathbf Q$ is the precision matrix and the relative importance in the spectrum gets reversed when inverting to obtain $\mathbf C$.

\subsubsection{Setup 2}

Here we create a process by applying linear operators to an initial multivariate functional datum $(Z_1,\ldots, Z_p)^\top$ with independent nodes. Firstly, we draw $Z_j$ for $j=1,\ldots,p$ (again with $p=100$) independently from a Gaussian distribution with mean zero and a covariance $\Sigma$ that has the Fourier basis eigenfunctions and quadratically decaying eigenvalues, i.e.~$\lambda_l = 1/l^2$ for $l=1,2,\ldots$. Then, we create $X=(X_1,\ldots,X_p)^\top$ as
\[
\begin{split}
    X_1 &= Z_1, \\
    X_2 &= \frac{2}{5} \mathcal{A}_1(X_1) + Z_2, \\
    X_j &= \frac{2}{5} \mathcal{A}_1(X_{j-1}) + \frac{1}{5} \mathcal{A}_2(X_{j-2}) + Z_j, \quad j=3,\ldots,p,
\end{split}
\]
where $\mathcal{A}_1$ and $\mathcal{A}_2$ are zero-extended restriction operators on the first and last tenth of the functional domain. Specifically, if $\Delta_E(t)=1$ for $t \in E$ and 0 otherwise, where $E \subset [0,1]$, we define $\mathcal{A}_1(f)(t)=\Delta_{[0,1/10]}(t) f(t)$ and $\mathcal{A}_2(f)(t)=\Delta_{[9/10,1]}(t) f(t)$ for any $f$. Since clearly $\|\mathcal{A}_k(f)\| \leq \|f\| $ for $k=1,2$, we have that $\| \mathcal{A}_k \| \leq 1$. Thus, the formulas above define a mean-reverting process that has the Markov property of order 2.

The goal of this construction is to have AR(2)-type dependencies, that are however created only locally in the time domain, as opposed to global dependencies in the spectral domain, like in Setup 1. We believe this setup constitutes a more realistic scenario, e.g.~from the perspective of neuroimaging applications.

\subsubsection{Setup 3}

In this final set of simulations, we generate $Z_j$, $j=1,\ldots,p$ (with $p=99)$ similarly to the previous setup as independent Gaussian processes with mean zero and the covariance $\Sigma_Z$ being rank-5 with Fourier basis eigenfunctions and the five non-zero eigenvalues all equal to one. For $j=3k-1$ where $k=1,\ldots,33$, we then generate $W_j$ as independent zero-mean Gaussian processes (also independent of $Z_j$'s) with the covariance $\Sigma_W$ and the corresponding kernel $k(t,s) = \frac{1}{2} (|t|^{2H} + |s|^{2H} - |t-s|^{2H})$, where $H=0.2$. For $j=3k$ or $j=3k-2$ where $k=1,\ldots,33$, we set $W_j := W_{3k-1}$. That is, the $W_j$'s are fractional Brownian motions with relatively rough sample paths (less smooth than those of  standard Brownian motion, which corresponds to $H=0.5$), and they are dependent (in fact, identical) in subsequent triplets. A single multivariate functional datum $X=(X_1,\ldots,X_p)^\top$ is then composed as 
\begin{equation}\label{eq:sim_setup3}
    X_j=3 Z_j + W_j
\end{equation}
and the actual measurements on the equidistant grid are also superposed with additional Gaussian white noise with variance $1/5$.

First, it is easy to verify by calculating conditional covariances that the graphical model is $(V,E)$ with the vertex set $V=\{1,\ldots,p\}$ and the edge set $E=\{ (i,j) \mid i,j=1,\ldots,99 : \lfloor i/3 \rfloor = \lfloor j/3 \rfloor \}$ in this case, i.e.~subsequent triplets of nodes are connected. Next, note that this is the only setup where we add measurement error to the generated values. The reason why we did not do this in the two previous setups (where the signal is smoothly varying) is that this would heavily favor the competing methods, which use denoising as the first step. Of course, we could also use some form of denoising, but we wish to avoid any specific approaches to estimate the covariance, since we view the methodology developed in this paper as one of a plug-in type. In this setup, however, the signal is relatively rough and so the competing approaches working with low-rank projections of the data are at a disadvantage. It is thus reasonable to add measurement error here, which also exemplifies that our method can naturally cope with it. Finally, the constant 3 in formula \eqref{eq:sim_setup3} and the white noise variance $1/5$ is chosen such that the total variability (as captured by the trace of the respective covariance operator) of the smooth component $Z_j$, the rough signal $W_j$ and the additional white noise, respectively, are in proportions 3:1:2 with each other. Thus while the form of the dependency is particularly simple in this setup, the signal is not very strong in the data.

\renewcommand{\baselinestretch}{1}
\bibliographystyle{rss}
\bibliography{biblio}
\end{document}